\documentclass[review,onefignum,onetabnum]{siamonline220329}
\usepackage{amsfonts,amssymb}
\usepackage{hyperref}
\usepackage{titlesec}
\usepackage{titletoc}
\usepackage{amsmath}
\usepackage{listings}
\usepackage{verbatim}
\usepackage{graphicx}
\usepackage{hyperref}
\usepackage{mathrsfs}
\usepackage{enumitem}
\usepackage{cite}
\usepackage{bbm}
\usepackage{mathtools}
\usepackage{xcolor, tikz}
\newtheorem{assumption}{Assumption}
\usepackage{bm}
\newsiamremark{remark}{Remark}

\DeclareMathOperator{\rank}{rank}
\DeclareMathOperator{\rad}{rad}

\DeclareMathOperator{\supp}{supp}
\DeclareMathOperator{\loc}{loc}

\DeclareMathOperator{\ba}{\textbf{a}}

\DeclareMathOperator{\bPhi}{\mathbf{\Phi}}

\DeclareMathOperator{\calA}{\mathcal{A}}
\DeclareMathOperator{\calB}{\mathcal{B}}
\DeclareMathOperator{\calQ}{\mathcal{Q}}
\DeclareMathOperator{\calC}{\mathcal{C}}
\DeclareMathOperator{\calD}{\mathcal{D}}
\DeclareMathOperator{\calE}{\mathcal{E}}
\DeclareMathOperator{\calG}{\mathcal{G}}
\DeclareMathOperator{\calZ}{\mathcal{Z}}
\DeclareMathOperator{\scrH}{\mathscr{H}}

\newsiamremark{rem}{Remark}

\ifpdf
\hypersetup{
  pdftitle={Uniform Recovery Guarantees for Quantized Corrupted Sensing Using Structured or Generative Priors},
  pdfauthor={J. Chen, Z. Liu, M. Ding, M. K. Ng}
}
\fi

\title{Uniform Recovery Guarantees for Quantized Corrupted Sensing Using \\Structured or Generative Priors}
\author{Junren Chen\thanks{J.~Chen is with  Department of Mathematics, The University of Hong Kong (\email{chenjr58@connect.hku.hk}). J. Chen was supported by a Hong Kong Ph.D. Fellowship from the Hong Kong Research Grant Council (HKRGC).}
	\and Zhaoqiang Liu\thanks{Z.~Liu is with the School of Computer Science and Engineering, University of Electronic Science and Technology of China (\email{zqliu12@gmail.com}). Z.~Liu was supported by the Shenzhen Fundamental Research Program (No. JCYJ20220530164812027), Sichuan Science and Technology Program (No. 2021YFS0374), and Sichuan Science and Technology Program (No. 2022YFS0600).} \and Meng Ding\thanks{M. Ding is with the School of Mathematics, Southwest Jiaotong University (\email{dingmeng56@163.com}). M. Ding was supported by the National Natural Science Foundation of China under Grant 12201522 and the Fundamental Research Funds for the Central Universities under Grant 2682023CX069.}	\and Michael K. Ng\thanks{M. K. Ng is with Department of Mathematics,  Hong Kong Baptist University (\email{michael-ng@hkbu.edu.hk}). M. K. Ng was supported by HKRGC GRF 12300218, 12300519, 17201020, 17300021, C1013-21GF, C7004-21GF and Joint NSFC-RGC   N-HKU76921.}
		}
		
\headers{Uniform Guarantees for Quantized Corrupted Sensing}{J. Chen, Z. Liu, M. Ding, M. K. Ng}

\begin{document}
\nolinenumbers
\maketitle

\begin{abstract}
This paper studies quantized corrupted sensing where the   measurements are contaminated by unknown corruption and then quantized by a dithered uniform quantizer. We establish uniform guarantees for Lasso that ensure the accurate recovery of {\it all} signals and corruptions using a single draw of the sub-Gaussian sensing matrix and uniform dither.  For signal and corruption with structured priors (e.g., sparsity, low-rankness), our uniform error rate for constrained Lasso typically coincides with the non-uniform one up to logarithmic factors, indicating that the uniformity costs very little. By contrast,  our uniform error rate for unconstrained Lasso exhibits worse dependence on the structured parameters due to regularization parameters larger than the ones for non-uniform recovery. These results complement the non-uniform ones recently obtained in [Sun, Cui and Liu, 2022] and provide more insights for understanding actual applications where the sensing ensemble is typically fixed and the corruption may be adversarial. 
  For signal and corruption living in  the ranges of some Lipschitz continuous generative models (referred to as generative priors), we achieve uniform recovery via constrained Lasso  with a measurement number proportional to the latent dimensions of the generative models. We present experimental results to corroborate our theories.  From the technical side, our treatments to the two kinds of priors are (nearly) unified and share the common key ingredients of (global) {\it  quantized product embedding} (QPE) property, which states that the dithered uniform quantization (universally) preserves inner product. As a by-product, our QPE result refines the one in [Xu and Jacques, 2020] under sub-Gaussian random matrix, and in this specific instance we are able to sharpen the uniform error decaying rate  (for the {\it projected-back projection} estimator with signals in some    convex symmetric set) 
  presented therein from $O(m^{-1/16})$ to $O(m^{-1/8})$.
\end{abstract}
\section{Introduction}\label{sec:introduction}

In  corrupted sensing, our goal is to recover   the   signal  $\bm{x^\star}\in \mathbb{R}^n$ and   corruption  $\bm{v^\star}\in \mathbb{R}^m$ from relatively few measurements \begin{equation}
    \label{1.1}\bm{y}=\bm{\Phi x^\star}+\sqrt{m}\bm{v^\star}+\bm{\epsilon},
\end{equation}
where $\bm{\Phi}\in \mathbb{R}^{m\times n}$ is the sensing matrix, $\bm{\epsilon}\in \mathbb{R}^m$ represents the noise vector, $\bm{y}$ denotes the measurements from which we seek to recover $\bm{x^\star}$ and $\bm{v^\star}$. When the corruption $\bm{v^\star}$ does not appear, \cref{1.1} reduces to the classical compressed sensing problem \cite{foucart2013invitation,candes2006robust,donoho2006compressed}; hence, corrupted sensing is a more challenging generalization of compressed sensing. While the   corrupted sensing problem is ill-posed in general, faithful recovery   can be achieved even in a high-dimensional regime (i.e., $m\ll n$) by utilizing additional structures of $(\bm{x^\star},\bm{v^\star})$, such as sparsity and low-rankness.
In the literature,
recovery methods with theoretical guarantees have   been well developed in a long line of works, first for some specific instances of \cref{1.1} like sparse signal recovery or low-rank matrix sensing under sparse corruption \cite{li2013compressed,nguyen2013exact,nguyen2012robust,xu2012outlier,chen2013low}, and then for the more general cases where $\bm{x^\star}$ and $\bm{v^\star}$ exhibit some structures that are often captured by Gaussian width \cite{chen2018stable,mccoy2014sharp,foygel2014corrupted}. {\color{black} We note that \cref{1.1}  captures a series of applications in imaging problems.} Specifically, associated with various operators $\bm{\Phi}$, the reduced model  $\bm{y}= \bm{\Phi x^\star}+\bm{\epsilon}$ (without corruption $\bm{v^\star}$) already models most problems in computational imaging \cite{barrett2013foundations}, from classical tasks like deblurring, inpainting and super-resolution, to a wide range of tomographic imaging applications such as magnetic resonance imaging and X-ray computed tomography (e.g., \cite{elbakri2002statistical,fessler2010model}). These inverse problems of recovering $\bm{x^\star}$ from $\bm{y}$ are collectively referred to as {\it image reconstruction}, where prior knowledge on the image signal is often available (e.g., sparsity in some dictionary or basis, (approximately) low-rankness, smoothness \cite{yu2011solving,gu2017weighted}). However, the flexibility of including the corruption $\bm{v^\star}$, as per \cref{1.1}, becomes necessary in certain imaging problems, in which the recovery of the corruption often provides useful information. We provide two  specific examples: 
\begin{itemize}
 [leftmargin=5ex,topsep=0.25ex]
    \item  The       
{\it face recognition} example~\cite{wright2008robust,candes2011robust} can be modeled  by \cref{1.1}, where the columns of the dictionary $\bm{\Phi}$ are the training face images,  the prior on the present face image $\bm{y}$ is that it can be represented as a {\it sparse} linear combination of the training faces. However, it is unreasitic to assume that the prior exactly holds true (e.g., due to undesired parts in $\bm{y}$ such as glasses and shadows), and a useful remedy is to complement our prior via a sparse corruption $\bm{v^\star}$. 

\item  In some {\it image reconstruction} problems our goal is to recover the image  and the impulsive signal  $\bm{v^\star}$, with   $\bm{\Phi}$ being certain dictionary that generates the image or  the sensing matrix that produces the compressive measurements of the signal $\bm{x^\star}$. See the star-galaxy separation example described in \cite{sun2022quantized} for instance. 
\end{itemize}
 Additionally, the corrupted sensing model \cref{1.1} has found applications in  sensor network analysis \cite{haupt2008compressed}, subspace clustering \cite{elhamifar2013sparse}, latent variable modeling \cite{chandrasekaran2011rank}, among others.

Note that the data is inevitably quantized to finite precision in digital signal processing \cite{gray1998quantization,widrow2008quantization}, and working with coarsely quantized data has proven   effective in many large scale machine learning or signal processing systems \cite{zhang2017zipml,hanna2021quantization,yang2023plug}. Consequently, recent years have witnessed rapidly increasing literature on quantized compressed sensing. These works proposed various quantization schemes that are associated with accurate (post-quantization) recovery methods, including 1-bit  quantization \cite{boufounos20081,plan2012robust,plan2013one,chen2022high,thrampoulidis2020generalized,dirksen2021non,jung2021quantized,jacques2013robust}, uniform (multi-bit) quantization \cite{chen2022quantizing,xu2020quantized,thrampoulidis2020generalized,jung2021quantized}, and other adaptive quantization methods as surveyed in \cite{dirksen2019quantized}. {\color{black}The recent work  \cite{tachella2023learning}  even extended the theoretical foundation of 1-bit compressed sensing \cite{jacques2013robust} to learning a signal set.} However, nearly all of them are restricted to   classical compressed sensing without     accounting for the additional corruption $\bm{v^\star}$.  The single exception is a recent work \cite{sun2022quantized},  in which  the authors     analyzed corrupted sensing under a dithered uniform quantizer $\mathcal{Q}_\delta(\cdot+\bm{\tau})$: 
\begin{equation} \label{eq:problem}   \bm{\dot{y}}:=\mathcal{Q}_\delta(\bm{y}+\bm{\tau})=\mathcal{Q}_\delta(\bm{\Phi x^\star}+\sqrt{m}\bm{v^\star}+\bm{\epsilon}+\bm{\tau}),
\end{equation}  
where $\bm{y}=\bm{\Phi x^\star}+\sqrt{m}\bm{v^\star}+\bm{\epsilon}$ is the unquantized measurements in corrupted sensing as per \cref{1.1}, $\bm{\tau}\in\mathbb{R}^m$ is the uniform dither, $\mathcal{Q}_\delta(\cdot)$ is the uniform quantizer with resolution $\delta$, see \cref{sec2.4}. 
Under  sub-Gaussian $\bm{\Phi}$, they showed that accurate recovery of $(\bm{x^\star},\bm{v^\star})$ can be achieved by Lasso, thus confirming the compatibility of the dithered uniform quantization and the recovery of the additional structured corruption $\bm{v^\star}$. Nonetheless, their recovery guarantees are {\it non-uniform} 
and only ensure the recovery of a  single  pair of 
 $(\bm{x^\star},\bm{v^\star})$ fixed before drawing $(\bm{\Phi},\bm{\epsilon},\bm{\tau})$,  with the implication that a new realization of $(\bm{\Phi},\bm{\epsilon},\bm{\tau})$ is needed for the sensing and recovery of a different pair of $(\bm{x^\star},\bm{v^\star})$. Thus, the possibility of uniform recovery in quantized corrupted sensing \cref{eq:problem} remains unaddressed: 
 \begin{align*}
     &\textit{Is it possible to recover all signals and corruptions using a fixed sensing ensemble}\\
     & (\bm{\Phi}, \bm{\epsilon},\bm{\tau})?~\textit{If yes, what is the cost of uniformity compared to non-uniform recovery?}
 \end{align*}

   \textbf{Importance of Uniformity:} Compressed sensing theories are developed to promote the understanding of its many real-world applications, {\color{black}and the above uniform recovery question  is important for theory of \cref{eq:problem} for the following reasons:}
\begin{itemize}
[leftmargin=5ex,topsep=0.25ex] 
    \item Uniform recovery is a highly sought-after notion   in compressed/corrupted sensing theory since the sensing matrix $\bm{\Phi}$ is typically fixed in applications (e.g., think of the above face recognition example), and one expects that the single fixed sensing ensemble  works for all possible signals (and corruptions) that may arise. In fact,  uniformity is a defining property for achieving {\it compression} in some applications, and non-uniform recovery with new sensing ensemble for new signal
    could be unrealistic (if not impossible) since the memory of $(\bm{\Phi},\bm{\tau})$ is already heavier than the signal itself. 


    \item In view of the corruption $\bm{v^\star}$ in \cref{eq:problem}, a uniform guarantee also offers {\it stronger robustness} than a non-uniform guarantee. Specifically, a uniform guarantee tolerates  $\bm{v^\star}$ generated   in an {\it adversarial} manner according to the knowledge of $(\bm{\Phi},\bm{\epsilon},\bm{\tau},\bm{x^\star})$, and the error bound remains valid as long as $\bm{v^\star}$ satisfies certain structured assumption like sparsity. In contrast, a non-uniform guarantee only works for a fixed $\bm{v^\star}$   {\it oblivious} to $(\bm{\Phi},\bm{\epsilon},\bm{\tau})$.
\end{itemize}

\textbf{Generative Prior:} Beyond the classical   structured priors, it was proposed in
  \cite{bora2017compressed} to assume that the desired signal in compressed sensing lies in the range of a pre-trained generative model, known as a generative prior.  This new perspective has led to successful numerical results such as a significant reduction of the required number of measurements for accurate recovery, as well as attracted much research attention with various extensions such as nonlinear models \cite{liu2020generalized,liu2020sample,qiu2020robust},  MRI applications \cite{jalal2021robust,quan2018compressed}, and information-theoretic bounds \cite{liu2020information,kamath2020power}, among others. {\color{black} Note that generative prior has now been widely applied to imaging-related inverse problems; we refer interested readers to \cite{ongie2020deep}.} For the specific quantized corrupted sensing problem \cref{eq:problem}, the results in \cite{sun2022quantized} are restricted to classical structured priors     promoted by certain norms, and it is unclear whether their theory extends to generative priors. Without considering quantization, 
the linear case of corrupted sensing with generative priors has been studied in  \cite{berk2020deep,berk2021deep} under the name of ``generative demixing''
  (here, \cite{berk2020deep} is the extended   version of the published conference paper \cite{berk2021deep}).

\subsection{Main Results}
In this paper, we 
establish   {\it uniform} recovery guarantees for \cref{eq:problem} with structured priors on $(\bm{x^\star},\bm{v^\star})$ using (un)constrained Lasso. Our uniform guarantees state that {\it a single realization} of $(\bm{\Phi},\bm{\epsilon},\bm{\tau})$  can be used for simultaneously recovering all $(\bm{x^\star},\bm{v^\star})$ in the structured sets,\footnote{This is a generalization of traditional signal structures such as sparsity and low-rank, see   \cref{defi1} and \cref{assump3}.} thus answering the above question in affirmative. We also provide careful comparison with \cite{sun2022quantized} to unveil the cost of uniformity.   
 Noticing  the recent trend of using generative prior for compressed sensing, we also present a  uniform recovery guarantee for quantized corrupted sensing where $\bm{x^\star}$ and $\bm{v^\star}$ are equipped with generative priors (see   \cref{assump5} for details). We highlight and summarize our major results as follows:  
  \begin{itemize}
       [leftmargin=5ex,topsep=0.25ex]

\item \textbf{Structured Priors via Constrained Lasso:}
We establish uniform recovery guarantees for quantized corrupted sensing \cref{eq:problem} with structured priors on $(\bm{x^\star},\bm{v^\star})$ via Lasso. In constrained Lasso, our uniform error rate  in \cref{thm1}  exhibits a decaying rate of $O(m^{-1/2})$ and typically coincides with the non-uniform one in \cite{sun2022quantized} up to logarithmic factors (see \cref{rem1} and  \cref{coro1}, \cref{coro2}), indicating that the uniformity costs very little. 
To our best knowledge, even  going back to   compressed sensing (without corruption) associated with the dithered uniform quantizer,
our \cref{thm1} provides the sharpest uniform error rate  for a computationally feasible decoder (see \cref{rem3}).

\item \textbf{Structured Priors via Unconstrained Lasso:} In unconstrained Lasso, our result in \cref{thm2} decays in $m$ with a rate of $O(m^{-1/2})$ but exhibits a  worse dependence on the structured parameter. To our best knowledge, this is the first uniform guarantee for quantized compressed sensing via unconstrained Lasso (see   \cref{rem4} and   \cref{coro3}, \cref{coro4}). The results in this and the previous dot point strengthen the non-uniform ones in \cite{sun2022quantized} and shed more light on the understanding of actual applications with fixed sensing ensemble and possibly adversarial corruption.

\item \textbf{Generative Priors:} We present the first result in \cref{thm3} for quantized corrupted sensing using generative priors,  which assume that the signal and corruption lie in the ranges of some Lipschitz continuous generative models with latent dimensions $k$ and $k'$. Our result guarantees that roughly $\tilde{O}\big(\frac{(1+\delta^2)(k+k')}{\mu^2}\big)$ measurements (up to logarithmic factors) suffice for achieving a uniform $\ell_2$-norm recovery error of $\mu$ via constrained Lasso. Note that this also implies a decaying rate of $O(m^{-1/2})$; see \cref{rem:gene_rate}. 
  \end{itemize}



\subsection{Technical Contributions and Technically Related Works} Considering the theoretical nature of our work, we provide an introduction from a more technical aspect. 

\textbf{A Sketch of our Techniques:} While most existing works focused on either structured prior or generative prior, we provide   near unified treatments for the two kinds of priors, with the common key ingredient  for achieving uniformity  being the global {\it quantized product embedding} (QPE) property, which states that the dithered quantization uniformly preserves inner product, i.e., $|\langle \mathcal{Q}_\delta(\bm{a}+\bm{\tau})-\bm{a}, \bm{b}\rangle|$ is uniformly small over all $(\bm{a},\bm{b})$ in certain constraint sets. The most general version of QPE is presented in \cref{thm:globalqpe} in \cref{app:generalqpe}, and a version sufficient for proving our main theorems is given in \cref{coro:qpe_structured} in \cref{app:special}. The proofs of these QPE results, despite following the conceptually simple covering strategy, appear to be the most technical and tedious part of this work. To avoid being overly technical in the main body, we collect the statements and proofs of QPE in \cref{qpe}. Equipped with QPE and a set of useful concentration inequalities (see \cref{appendixA1}), the major differences in proving our three main theorems (\cref{thm1}, \ref{thm2}, \ref{thm3}) lie in estimating the Gaussian width and Kolmogorov entropy   of various constrain sets, which will be  settled in \cref{appendixA.2}.

We provide   comparisons with the     works   most relevant to this paper in techniques (Readers less interested in proof techniques could safely skip the discussions below):

\begin{itemize}
[leftmargin=5ex,topsep=0.25ex]
    \item \textbf{Comparing with \cite{xu2020quantized} on QPE:} Under the name of {\it limited projection distortion} (LPD) property, Xu and Jacques \cite{xu2020quantized} utilized   QPE for analyzing the projected-back projection (PBP) estimator. In this work, we show that global QPE also serves as the key ingredient in analyzing uniform recovery via Lasso. Compared to the global QPE in \cite{xu2020quantized}, our  \cref{thm:globalqpe} is a generalization and provides (instance-wise) improvement under sub-Gaussian sensing matrix; see \cref{coro:recover} and \cref{coro:improved_qpe} in \cref{app:implication}. As an interesting by-product of \cref{coro:improved_qpe}, under sub-Gaussian sensing matrix, we are able to improve the uniform error rate of PBP over a symmetric convex signal set from $O(m^{-1/16})$ to $O(m^{-1/8})$; see \cref{prop:improve} in \cref{app:improve}. 

    \item {\bf Comparing with \cite{genzel2022unified,chen2023unified} on the Approach to Uniformity:} The recent work \cite{genzel2022unified} due to Genzel and Stollenwerk developed a unified approach to proving uniform recovery guarantees for   constrained Lasso in non-linear compressed sensing $y_i = f_i(\bm{\Phi}_i^\top\bm{x^\star})$, where the possibly random $f_i(\cdot)$ captures some non-linearity that can be unknown and/or discontinuous.  However, under discontinuous $f_i(\cdot)$, their general strategy leads to a uniform decaying rate $O(m^{-1/4})$ inferior to our $O(m^{-1/2})$. More recently, Chen et al. \cite{chen2023unified} extended the scope of \cite{genzel2022unified} to non-linear compressed sensing with generative prior.  They observed that using a different concentration inequality yields tighter bound for the generative case, thus they managed to prove a uniform decaying rate of $O(m^{-1/2})$ for discontinuous $f_i(\cdot)$ (e.g.,   various quantization models). However, as we will discuss in \cref{rem:insensible} and \cref{rem:gene_related}, some hurdle arises if we follow the general strategy in \cite{genzel2022unified,chen2023unified} to prove our main theorems, thus our techniques are not implied by these two works. Indeed, our QPE-based analysis suggests a  possible strategy to improve the rate $O(m^{-1/4})$ in \cite{genzel2022unified}   to $O(m^{-1/2})$ under discontinuous $f_i(\cdot)$. 
     
\end{itemize}




\subsection{Paper Outline} In   \cref{sec2:pre}, we provide preliminaries and set up notations. In   \cref{sec3:main}, we present our main theorems.  Experimental results are reported in   \cref{sec4:expresult}.  We provide some remarks to conclude the paper in   \cref{sec5:conclu}. We provide technical lemmas and useful propositions in \cref{appenA} to support our analysis. We develop a   general global QPE result in \cref{app:generalqpe} and present a version sufficient for proving our main theorems in \cref{app:special}. The proofs of results in the main body, if missing, are deferred to  \cref{appendixc}. In \cref{app:by}, we present more implications of our QPE result (\cref{app:implication}) and then obtain a by-product for PBP estimator (\cref{app:improve}). In \cref{app:table}, we provide a list of recurring notation (\cref{tbl:notation}) to improve the readability of this paper.

\section{Preliminaries}\label{sec2:pre}
We first collect some generic notations. We represent matrices and vectors by boldface letters, scalars by regular letters. For positive integer $m$ we write $[m]:=\{1,...,m\}$. {\color{black}We use $|\mathcal{S}|$ to denote the cardinality of any finite set $\mathcal{S}$.}
For a vector $\bm{x}=[x_i]\in \mathbb{R}^d$, we work with the $\ell_p$-norm $\|\bm{x}\|_{p}=(\sum_i |x_i|^p)^{1/p}$ ($p\geq 1$), max norm $\|\bm{x}\|_\infty=\max_i|x_i|$, and zero norm $\|\bm{x}\|_0$ that counts the number of non-zero entries. 
We write the standard Euclidean sphere in $n$-dimensional space as $\mathbb{S}^{n-1}=\{\bm{x}\in \mathbb{R}^n:\|\bm{x}\|_2=1\}$, the set of $s$-sparse vectors as $\Sigma^n_s=\{\bm{x}\in\mathbb{R}^n:\|\bm{x}\|_0\leq s\}$. The inner product of $\bm{x},\bm{y}\in\mathbb{R}^n$ is $\langle\bm{x},\bm{y}\rangle=\bm{x}^\top\bm{y}$. Given a matrix $\bm{A}\in \mathbb{R}^{p\times q}$, we denote its operator norm (that equals the maximal singular value), Frobenius norm, nuclear norm (i.e., sum of singular values) by $\|\bm{A}\|_{\mathrm{op}}$, $\|\bm{A}\|_{\mathrm{F}}$, $\|\bm{A}\|_{\mathrm{nu}}$, respectively. The inner product between matrices $\bm{A},\bm{B}$ is   $\langle\bm{A},\bm{B}\rangle=\mathrm{Tr}(\bm{A}^\top\bm{B})$. The set of matrices with rank not exceeding $r$ is denoted by $M^{p,q}_r=\{\bm{A}\in \mathbb{R}^{p\times q}:\rank(\bm{A})\leq r\}$.

Given a norm $f(\cdot)$ in $\mathbb{R}^n$ (resp. $\mathbb{R}^{p\times q}$), we denote the corresponding ball with radius $r$ by $\mathbb{B}_f^n(r)=\{\bm{x}\in\mathbb{R}^n\,:\, f(\bm{x})\leq r\}$ (resp. $\mathbb{B}_f^{p,q}(r)=\{\bm{A}\in \mathbb{R}^{p\times q}\,:\, f(\bm{A})\leq r\}$). We let $\mathbb{B}_f^n:=\mathbb{B}_f^n(1)$, $\mathbb{B}_f^{p,q}:=\mathbb{B}_f^{p,q}(1)$ be the unit ball. 
For instance, $\mathbb{B}_2^n$, $\mathbb{B}_1^n$ are respectively the $\ell_2$-ball, $\ell_1$-ball in $\mathbb{R}^n$, $\mathbb{B}_\mathrm{F}^{p,q}$, $\mathbb{B}_{\mathrm{nu}}^{p,q}$ are respectively the Frobenius norm ball, nuclear norm ball in $\mathbb{R}^{p\times q}$. The dual norm of $f(\cdot)$ is defined as $f^*(\bm{x}):=\sup_{\bm{y}\in \mathbb{B}_f}\langle\bm{x},\bm{y}\rangle$, and we note the H\"older's inequality $\langle\bm{x},\bm{y}\rangle\leq f(\bm{x})\cdot f^*(\bm{y})$. The descent cone of $f$ at a point $\bm{x}$ and its normalized counterpart are defined as \begin{equation}\label{descentcone}
    \mathcal{D}_f(\bm{x}):=\{\bm{u}\,:\,\exists~t>0,~\text{s.t. }f(\bm{x}+t\bm{u})\leq f(\bm{x})\},~\mathcal{D}_f^*(\bm{x}):=\mathcal{D}_f(\bm{x^\star})\cap\mathbb{S}^{n-1}. 
\end{equation} 

Throughout this paper, $\mathbbm{P}(\cdot)$, $\mathbbm{E}(\cdot)$, $\mathbbm{1}(\cdot)$ stand for probability, expectation, indicator function, respectively. We make no attempt to refine multiplicative constants, and we use $C,C_i,c,c_i$ to denote absolute constants whose value may vary from line to line. For some quantities $I_1,I_2$, We write $I_1=O(I_2)$ or $I_1\lesssim I_2$ if $I_1\leq CI_2$  holds for some absolute constant $C$; Conversely, we write $I_1=\Omega(I_2)$ or $I_1\gtrsim I_2$ if $I_1\geq cI_2$ for some $c$; we refer to $(C,c)$ behind $\lesssim$ or $\gtrsim$ as the implied constant. {\color{black} For some quantity $I$, we may write ``certain event holds with probability exceeding $1-\exp(-\Omega(I))$'' to state that  this event holds with probability at least $1-\exp(cI)$ for some absolute constant $c>0$.} We will write $I_1\asymp I_2$ to  state  that $I_1=O(I_2)$ and $I_1=\Omega(I_2)$ simultaneously hold.  Given $W\subset \mathbb{R}^d$ we use $\mathscr{U}(W)$ to denote the uniform distribution over $W$. Also, $\mathcal{N}(\bm{\mu},\bm{\Sigma})$ represents Gaussian variable/vector with mean $\bm{\mu}$ and covariance $\bm{\Sigma}$.

\subsection{Sub-Gaussian Random Variable}\label{sec2.1}
The sub-Gaussian norm of a random variable $X$ is defined as $\|X\|_{\psi_2}=\inf\{t>0: \mathbbm{E}(X^2/t^2)\leq 2\}$, and we have   $\|X\|_{\psi_2}\asymp \sup_{p\geq 1}p^{-1/2}(\mathbbm{E}|X|^p)^{1/p}$ and hence $\|X\|_{\psi_2}=O(K)$ if $|X|\leq K$. $X$ is said to be sub-Gaussian if $\|X\|_{\psi_2}<\infty$, and sub-Gaussian $X$ has probability tail resembling that of a Gaussian variable:  
\begin{equation}\label{eq:sg_tail}
    \mathbbm{P}(|X|\geq t)\leq 2\exp\Big(-\frac{ct^2}{\|X\|_{\psi_2}^2}\Big) 
\end{equation}
holds for any $t>0$. 
Moreover, for independent {\it zero-mean} random variables $X_i$'s we have (see \cite[Prop. 2.6.1]{vershynin2018high}) \begin{align} \label{eq:sum_psi2}
    \Big\|\sum_i X_i\Big\|_{\psi_2}^2\leq C\sum_i\|X_i\|_{\psi_2}^2.
\end{align} 
A random vector $\bm{X}\in \mathbb{R}^n$ is sub-Gaussian if it has finite sub-Gaussian norm, which is  defined   as $\|\bm{X}\|_{\psi_2}=\sup_{\bm{v}\in \mathbb{S}^{n-1}}\|\bm{v}^\top\bm{X}\|_{\psi_2}$. Assume 
$\bm{X}$ has independent {\it zero-mean} entries $X_i$'s satisfying $\|X_i\|_{\psi_2}\leq K$, then $\|\bm{X}\|_{\psi_2}=O(K)$ \cite[Lem. 3.4.2]{vershynin2018high}. Readers may refer to \cite[Sections 2--3]{vershynin2018high} for more details.

\subsection{Covering Number and Kolmogorov Entropy}\label{sec2.2}
Given $\mathcal{K}\subset\mathbb{R}^n$, a subset $\mathcal{G}\subset \mathcal{K}$ is said to be an $\varepsilon$-net (with respect to Euclidean distance) of $\mathcal{K}$, if for any $\bm{x}\in\mathcal{K}$ there exists $\bm{x}'\in\mathcal{G}$ satisfying $\|\bm{x}-\bm{x}'\|_2\leq \varepsilon$, i.e., $\mathcal{K}\subset \cup_{\bm{x}\in\mathcal{G}}\big(\bm{x}+\mathbb{B}_2^n(\varepsilon)\big)$. The covering number of $\mathcal{K}$ under radius $\varepsilon$, denoted $\mathscr{N}(\mathcal{K},\varepsilon)$, is defined to be the smallest possible cardinality of an $\varepsilon$-net of $\mathcal{K}$.   We note the following monotonicity of covering number \cite[Exercise 4.2.10]{vershynin2018high}:
\begin{equation}\label{monoto}
    \mathscr{N}(\mathcal{K},\varepsilon)\leq \mathscr{N}(\mathcal{K}',\varepsilon/2),~\text{if}~\mathcal{K}\subset \mathcal{K}'.
\end{equation}
We will more frequently work with  the equivalent  notion called Kolmogorov entropy:  
\begin{align}
    \mathscr{H}(\mathcal{K},\varepsilon)=\log\mathscr{N}(\mathcal{K},\varepsilon). 
\end{align}

\subsection{Gaussian Width and Gaussian Complexity}\label{sec2.3}
Suppose that $\bm{g}$ has i.i.d. $\mathcal{N}(0,1)$ entries, then the Gaussian width of $\mathcal{K}\subset \mathbb{R}^n$  is defined as $\omega(\mathcal{K})=\mathbbm{E}\sup_{\bm{x}\in\mathcal{K}}\langle\bm{g},\bm{x}\rangle$,  which is a geometric quantity that precisely and stably captures the intrinsic dimension of $\mathcal{K}$. Gaussian width is closely related to Kolmogorov entropy. Specifically, we can bound the Kolmogorov entropy by Gaussian width via Sudakov's inequality \cite[Thm. 8.1.13]{vershynin2018high} \begin{equation}
    \label{sudakov}
    \mathscr{H}(\mathcal{K},\varepsilon) \leq \frac{\omega^2(\mathcal{K})}{\varepsilon^2}.
\end{equation}
We can also estimate $\omega(\mathcal{K})$ by $\mathscr{H}(\mathcal{K},\cdot)$ via Dudley's inequality \cite[Thm. 8.1.10]{vershynin2018high} \begin{equation}\label{dudley}
    \omega(\mathcal{K})\leq C\int_0^\infty \sqrt{\mathscr{H}(\mathcal{K},\varepsilon)}~\mathrm{d}\varepsilon.
\end{equation}
A slightly different notion is the Gaussian complexity defined as $\gamma(\mathcal{K})=\mathbbm{E}\sup_{\bm{x}\in\mathcal{K}}|\langle\bm{g},\bm{x}\rangle|$. In many cases $\omega(\mathcal{K})$ and $\gamma(\mathcal{K})$ are of similar scaling. For instance, $\omega(\mathcal{K})\asymp \gamma(\mathcal{K})$ holds if $0\in \mathcal{K}$.
More generally, we note the following relation from \cite[Exercise 7.6.9]{vershynin2018high}
\begin{equation}\label{widthcomple}
    \frac{1}{3}\Big(\omega(\mathcal{K})+\|\bm{x}\|_2\Big) \leq \gamma(\mathcal{K})\leq 2\Big(\omega(\mathcal{K})+\|\bm{x}\|_2\Big)
\end{equation}
that holds for any $\bm{x}\in\mathcal{K}$. 
 Given $\mathcal{K}\subset \mathbb{R}^n$ and some $\rho>0$, we will work with the localized version of $\mathcal{K}$ defined as \begin{align}
     \label{eq:localize}
     \mathcal{K}^{(\rho)}_{\loc} = (\mathcal{K}-\mathcal{K})\cap  \mathbb{B}_2^n(\rho). 
 \end{align}
 We denote the radius of $\mathcal{K}$ by $\rad(\mathcal{K})=\sup_{\bm{a}\in\mathcal{K}}\|\bm{a}\|_2$. 
\subsection{Dithered Uniform Quantization}\label{sec2.4}
For some resolution $\delta>0$, the uniform quantizer $\mathcal{Q}_\delta(\cdot)$ quantizes a scalar $a$ to \begin{equation}
\mathcal{Q}_\delta(a)=\delta\left(\Big\lfloor\frac{a}{\delta}\Big\rfloor+\frac{1}{2}\right).\label{quantizer}
\end{equation} Note that $\mathcal{Q}_\delta(\cdot)$ enjoys the bounded distortion property, i.e., 
\begin{align}
    |\mathcal{Q}_\delta(a)-a|\leq \frac{\delta}{2}
\end{align} 
holds for any $a$. In this paper, we use a dithered uniform quantizer that involves a uniform dither $\tau\sim \mathscr{U} [-\frac{\delta}{2},\frac{\delta}{2}]$. 
Specifically, we quantize $a$ to $\mathcal{Q}_\delta(a+\tau)$, and we refer to $\xi:=\mathcal{Q}_\delta(a+\tau)-a$ as quantization noise, which is always bounded because 
\begin{align}\label{eq:boundedness}
    |\xi|\leq |\mathcal{Q}_\delta(a+\tau)-(a+\tau)|+|\tau|\leq \frac{\delta}{2}+\frac{\delta}{2}=\delta.
\end{align} 
With dithering, the nice property is that the quantization noise is zero-mean: 
\begin{align}\label{eq:zero_mean}
   \mathbbm{E}(\xi) =  \mathbbm{E}[\mathcal{Q}_\delta(a+\tau)]-a= 0, 
\end{align}  see \cite{gray1993dithered,thrampoulidis2020generalized,chen2022quantizing,gray1998quantization} for instance.\footnote{It is revealed by \cref{eq:zero_mean} that the benefit of dithering is to whiten the quantization noise. As a brief introduction, we mention that the use of dithering (prior to quantization) dates back to early engineering works \cite{jayant1972application,limb1969design} and theoretical analysis \cite{schuchman1964dither}, while in the past few years it has regained a surge of research interest in various estimation/recovery problems, including compressed sensing \cite{jung2021quantized,dirksen2021non,xu2020quantized,thrampoulidis2020generalized,sun2022quantized}, matrix completion \cite{davenport20141,chen2022high,chen2022quantizing,cai2013max}, and more recently covariance estimation \cite{dirksen2022covariance,dirksen2023tuning,chen2023parameter} and reduced-rank regression \cite{chen2023quantized}.} To quantize a   vector $\bm{a}\in \mathbb{R}^m$, we apply the dithered uniform quantizer to each entry in a {\it memoryless} manner. That is, we draw a random uniform dither $\bm{\tau}\sim \mathscr{U} ([-\frac{\delta}{2},\frac{\delta}{2}]^m)$ and then quantize $\bm{a}$ to $\mathcal{Q}_\delta(\bm{a}+\bm{\tau})$. Let $\bm{\xi}=\mathcal{Q}_\delta(\bm{a}+\bm{\tau})-\bm{a}$ be the quantization noise. It follows that, for a fixed $\bm{a}\in \mathbb{R}^m$, entries of $\bm{\xi}$ are independent (since entries of $\bm{\tau}$ are independent), zero-mean (due to \cref{eq:zero_mean}), and bounded by $\delta$ (see \cref{eq:boundedness}), and hence the sub-Gaussian norm of each entry also scales as $O(\delta)$. Taken collectively, 
we arrive at (see   \cref{sec2.1})
 \begin{align}
     \|\bm{\xi}\|_{\psi_2}=O(\delta). \label{eq:sgxi}
 \end{align}

\section{Main Results}\label{sec3:main}
Recall that   the corrupted sensing problem can be formulated as 
\begin{equation}\label{3.1}
    \bm{y}=\bm{\Phi x^\star}+\sqrt{m}\bm{v^\star}+\bm{\epsilon},
\end{equation}
where $\bm{v^\star}$ is the corruption mixed with the clean measurements $\bm{\Phi x^\star}$ of the signal $\bm{x^\star}$, and $\bm{\epsilon}$ represents measurement noise. In this work, we study 
a more challenging nonlinear model that involves quantization of $\bm{y}$, adopting a dithered uniform quantizer  
following prior works \cite{thrampoulidis2020generalized,xu2020quantized,sun2022quantized,chen2022quantizing,jung2021quantized}. For some quantization level $\delta>0$,\footnote{Smaller $\delta$ corresponds to higher resolution. Specifically, letting $\delta\to 0$ returns the unquantized (full-data) setting.}  
%
 we  acquire the quantized measurements as
 \begin{equation}
\label{3.2}   
\bm{\dot{y}}:=\mathcal{Q}_\delta(\bm{y}+\bm{\tau})=\mathcal{Q}_\delta(\bm{\Phi x^\star}+\sqrt{m}\bm{v^\star}+\bm{\epsilon}+\bm{\tau}),
\end{equation}
where $\bm{\tau}\sim \mathscr{U}([-\frac{\delta}{2},\frac{\delta}{2}]^m)$ is the uniform dither independent of $\bm{\Phi}$ and $\bm{\epsilon}$. 
 We denote the quantization noise (that depends on $(\bm{x^\star},\bm{v^\star})$) by 
\begin{align}\label{eq:xi_xv}
    \bm{\xi}_{\bm{x^\star},\bm{v^\star}}: &= \bm{\dot{y}} - \bm{y} =   \mathcal{Q}_\delta(\bm{\Phi x^\star}+\sqrt{m}\bm{v^\star}+\bm{\epsilon}+\bm{\tau})- (\bm{\Phi x^\star}+\sqrt{m}\bm{v^\star}+\bm{\epsilon}), 
\end{align}
with the $k$-th entry denoted by $(\bm{\xi}_{\bm{x^\star},\bm{v^\star}})_k$. Then we can also write \cref{3.2} as
\begin{align}\label{eq:quan_noise}
    \bm{\dot{y}} = \bm{\Phi x^\star}+\sqrt{m}\bm{v^\star}+\bm{\epsilon} + \bm{\xi}_{\bm{x^\star},\bm{v^\star}}. 
\end{align}
Moreover, let $\bm{\Phi}_i^\top$ be the $i$-th row of $\bm{\Phi}$, $v_i^\star$ and $\epsilon_i$ be respectively the $i$-th entry of $\bm{v^\star}$ and $\bm{\epsilon}$, then the $i$-th quantized measurement $\dot{y}_i$ is given by  
\begin{align}
        \dot{y}_i & =\mathcal{Q}_\delta(\bm{\Phi}_i^\top\bm{x^\star}+\sqrt{m}v^\star_i + \epsilon_i+\tau_i)=\bm{\Phi}_i^\top\bm{x^\star}+\sqrt{m}v^\star_i + \epsilon_i+(\bm{\xi}_{\bm{x^\star},\bm{v^\star}})_i.
\end{align}

Note that the recovery behaviour under a sub-Gaussian sensing matrix serves as an important benchmark in compressed sensing. In this work, we make the following assumption.  
\begin{assumption}[Random Sensing Ensemble]\label{assump1}
$\bm{\Phi}_1,...,\bm{\Phi}_m$ are independent, zero-mean, isotropic (i.e., $\mathbbm{E}(\bm{\Phi}_i\bm{\Phi}_i^\top)=\bm{I}_n$),   sub-Gaussian {\color{black}sensing vectors} satisfying $\|\bm{\Phi}_i\|_{\psi_2} \leq K$ for all $i$ and for some absolute constant $K$; $\epsilon_1,...,\epsilon_m$ are independent of each other and of $\bm{\Phi}$, and sub-Gaussian satisfying $\|\epsilon_i\|_{\psi_2}\leq E$; $\tau_1,...,\tau_m$ are independent of each other and of $(\bm{\Phi},\bm{\epsilon})$, and $\tau_i\sim\mathscr{U}[-\frac{\delta}{2},\frac{\delta}{2}]$.  
\end{assumption}

To handle the high-dimensional regime where $m\ll n$, it is standard and necessary to utilize the low-complexity structures of $(\bm{x^\star},\bm{v^\star})$. Following \cite{chen2018stable} we assume that the structure can be promoted by some norm, for instance, $\ell_1$-norm for sparsity, $\ell_1/\ell_2$-norm for group sparsity, and nuclear norm for low-rankness. 
\begin{assumption}[Structures Promoted by Norms]\label{assump2} 
    For some low-complexity sets $\mathcal{K}_{\bm{x}}\subset \mathbb{R}^n$ and $\mathcal{K}_{\bm{v}}\subset \mathbb{R}^m$, we assume  that $\bm{x^\star}\in \mathcal{K}_{\bm{x}}$ and $\bm{v^\star}\in \mathcal{K}_{\bm{v}}$. 
    The structures of $\bm{x^\star}$ and $\bm{v^\star}$ can be promoted by the norms $f(\bm{x}) $ and $g(\bm{v})$, respectively.
\end{assumption}

We will investigate the recovery performance of two types of Lasso: the constrained Lasso\footnote{Though it is possible the pursue a relaxation (e.g., \cite{plan2016generalized,genzel2022unified}), 
we follow prior works such as \cite{sun2022quantized,chen2018stable,chen2023quantized,thrampoulidis2020generalized} to consider the constrained Lasso with {\it the best possible constraint} to allow for an descent-cone-based analysis.}
\begin{equation}
   \label{3.4} (\bm{\hat{x}},\bm{\hat{v}})=\mathrm{arg}\min_{\substack{\bm{x}\in \mathbb{R}^n\\\bm{v}\in \mathbb{R}^m}}~ \|\bm{\dot{y}}-\bm{\Phi x}-\sqrt{m}\bm{v}\|_2,~\mathrm{s.t.}~f(\bm{x})\leq f(\bm{x^\star}),~g(\bm{v})\leq g(\bm{v^\star}),
\end{equation}
and the unconstrained Lasso
\begin{equation}
    \label{3.5}(\bm{\hat{x}},\bm{\hat{v}})=\mathrm{arg}\min_{\substack{\bm{x}\in \mathbb{R}^n\\\bm{v}\in \mathbb{R}^m}}~\|\bm{\dot{y}}-\bm{\Phi x}-\sqrt{m}\bm{v}\|^2_2+\lambda_1\cdot f(\bm{x})+\lambda_2\cdot g(\bm{v}).
\end{equation}
Note that the loss function $\mathcal{L}(\bm{x},\bm{v}):=\|\bm{\dot{y}}-\bm{\Phi x}-\sqrt{m}\bm{v}\|_2^2$   is simply the regular $\ell_2$-loss with full observations $\bm{y}$ substituted by the quantized ones $\bm{\dot{y}}$.

 The non-uniform guarantees    in \cite{sun2022quantized} state that for any fixed $(\bm{x^\star},\bm{v^\star})$, \cref{3.4} and  \cref{3.5}  deliver comparably accurate recovery using a random realization of $(\bm{\Phi},\bm{\dot{y}})$ according to   \cref{assump1}. By contrast, the primary  goal of this paper is to establish uniform recovery guarantees that ensure the  accurate recovery of {\it all} possible pairs of $(\bm{x^\star},\bm{v^\star})$ using $(\bm{\Phi},\bm{\dot{y}})$, where the quantized measurements $\bm{\dot{y}}$ are produced by a single draw of $(\bm{\Phi},\bm{\epsilon},\bm{\tau})$. Compared to non-uniform guarantees, our uniform ones provide more insights in understanding the actual applications with fixed sensing ensemble and possibly adversarial corruption, as we explained in 
 \cref{sec:introduction}. 

\subsection{Structured Priors with Constrained Lasso}

Our first set of results are for \cref{3.2} where  $\bm{x^\star}$ and $\bm{v^\star}$ lie in some structured   sets, which we follow \cite[Section 3.1]{xu2020quantized} and define as follows. 

\begin{definition}[Structured Set]\label{defi1}
If $\mathcal{K}\subset \mathbb{R}^p$ is a cone satisfying 
\begin{equation}\label{structuredset}
    \mathscr{H}(\mathcal{K}\cap \mathbb{B}_2^{p},\eta) \leq C\cdot\omega^2(\mathcal{K}\cap \mathbb{B}_2^{p})\log\Big(1+\frac{1}{\eta}\Big)
\end{equation}
for all $\eta>0$ and for some absolute constant $C$, then we say $\mathcal{K}$ is a structured set. 
\end{definition}

Structured   set is a generalization of various prototypical structures utilized in compressed sensing, e.g., (group) sparse vectors, low-rank matrices ({\em cf.}   \cref{pro6}), subspaces, union of subspaces, among others.  In   \cref{defi1},  the distinguishing feature of a structured set is that its Kolmogorov entropy  is only logarithmically dependent on its covering radius, in contrast to Sudakov's inequality     \cref{sudakov} that holds for arbitrary subset. We note that notions analogous to structured sets have been widely adopted in the literature \cite{dirksen2016dimensionality,chen2022uniform,oymak2015near}.

To derive a uniform error bound in quantized compressed sensing, it is standard to  concentrate on $\bm{x^\star},\bm{v^\star}$ with bounded $\ell_2$-norm (e.g., \cite[Thm. 3]{jung2021quantized}, \cite[Thm. 4.1]{xu2020quantized}). Without loss of generality, we   assume that $\bm{x^\star}\in \mathbb{B}_2^n$ and $\bm{v^\star}\in \mathbb{B}_2^m$ and   state our structured set assumption as follows. 

\begin{assumption}[Structured Priors]
    \label{assump3} 
    Given a pair of structured  sets $\mathcal{K}_{\bm{x}}^0$ and $\mathcal{K}_{\bm{v}}^0$ as per   \cref{defi1},  
      we let  $\mathcal{K}_{\bm{x}}=\mathcal{K}^0_{\bm{x}}\cap \mathbb{B}_2^n$ and $\mathcal{K}_{\bm{v}}=\mathcal{K}_{\bm{v}}^0\cap \mathbb{B}_2^m$. According to   \cref{defi1}, for any $\eta>0$ we have
      \begin{equation}\label{3.77}
          \mathscr{H}(\mathcal{K}_{\bm{x}},\eta)\lesssim \omega^2(\mathcal{K}_{\bm{x}})\log\Big(1+\frac{1}{\eta}\Big),~ \mathscr{H}(\mathcal{K}_{\bm{v}},\eta)\lesssim \omega^2(\mathcal{K}_{\bm{v}})\log\Big(1+\frac{1}{\eta}\Big).
      \end{equation}
\end{assumption}

We first consider constrained Lasso \cref{3.4} and present an upper bound that holds uniformly for all $(\bm{x^\star},\bm{v^\star})\in \mathcal{K}_{\bm{x}}\times \mathcal{K}_{\bm{v}}$. 

\begin{theorem}[Uniform Recovery via Constrained Lasso]\label{thm1}
Under   \cref{assump1}--\cref{assump3}, 
 we define the constraint sets 
\begin{gather}
    \label{3.91}\mathcal{D}_{\bm{x}}:=\cup_{\bm{x^\star}\in \mathcal{K}_{\bm{x}}}\mathcal{D}_f(\bm{x^\star}),~\mathcal{D}_{\bm{v}}:=\cup_{\bm{v^\star}\in \mathcal{K}_{\bm{v}}} \mathcal{D}_g(\bm{v^\star})\\\label{3.101}
\mathcal{D}_{\bm{x}}^*:=\mathcal{D}_{\bm{x}}\cap \mathbb{S}^{n-1}=\cup_{\bm{x^\star}\in \mathcal{K}_{\bm{x}}}\mathcal{D}_f^*(\bm{x^\star}),~\mathcal{D}_{\bm{v}}^*:=\mathcal{D}_{\bm{v}}\cap \mathbb{S}^{m-1}=\cup_{\bm{v^\star}\in \mathcal{K}_{\bm{v}}}\mathcal{D}_g^*(\bm{v^\star}).
\end{gather}
Suppose that the positive scalars $(\zeta,\rho_1,\rho_2)$ and the sample size $m$ satisfy 
\begin{gather}\label{eq:qpecon1}
    \zeta  = \frac{4\delta(\mathscr{H}(\mathcal{K}_{\bm{x}},\rho_1)+\mathscr{H}(\mathcal{K}_{\bm{v}},\rho_2))}{m},\\\label{eq:qpecon2}
    \rho_1 \le \frac{c_1\zeta}{(\log\frac{\delta}{\zeta})^{1/2}},~\omega\big((\mathcal{K}_{\bm{x}})^{(\rho_1)}_{\loc}\big)\le c_2\zeta\sqrt{\frac{m\zeta}{\delta}},~\rho_2\le c_3\zeta\sqrt{\frac{\zeta}{\delta}} 
\end{gather}
for  small enough $(c_1,c_2,c_3)$. If for large enough $C_4$ it holds that 
\begin{align}
    \label{eq:thm1_sam_size}
    m\ge C_4 \Big(\gamma^2(\mathcal{D}_{\bm{x}}^*)+\gamma^2(\mathcal{D}_{\bm{v}}^*)+\mathscr{H}(\mathcal{K}_{\bm{x}},\rho_1)+\mathscr{H}(\mathcal{K}_{\bm{v}},\rho_2)\Big)
\end{align}
then with probability exceeding  
\begin{equation}\label{promisedthm1}
    1 -  12\exp\big(-\Omega(\mathscr{H}(\mathcal{K}_{\bm{x}},\rho_1)+\mathscr{H}(\mathcal{K}_{\bm{v}},\rho_2))\big)-6\exp\big(-\Omega(\gamma^2(\mathcal{D}^*_{\bm{x}})+\gamma^2(\mathcal{D}^*_{\bm{v}}))\big)
\end{equation}
 on a single draw of $(\bm{\Phi},\bm{\epsilon},\bm{\tau})$, 
 the following uniform error bound holds true for all $(\bm{x^\star},\bm{v^\star})\in \mathcal{K}_{\bm{x}}\times \mathcal{K}_{\bm{v}}$:
\begin{align}
    \label{eq:thm1bound}
\sqrt{\|\bm{\Delta_x}\|_2^2+\|\bm{\Delta_v}\|_2^2}\lesssim \frac{(E+\delta)\big(\gamma(\mathcal{D}^*_{\bm{x}})+\gamma(\mathcal{D}^*_{\bm{v}})\big)+\delta(\mathscr{H}(\mathcal{K}_{\bm{x}},\rho_1)+\mathscr{H}(\mathcal{K}_{\bm{v}},\rho_2))^{1/2}}{\sqrt{m}},
\end{align}
 where $\bm{\Delta_x}=\bm{\hat{x}}-\bm{x^\star}$ and $\bm{\Delta_v}=\bm{\hat{v}}-\bm{v^\star}$, with $(\bm{\hat{x}},\bm{\hat{v}})$ being the solution  to \cref{3.4}
\end{theorem}
\begin{proof}
We present the proof in three steps.

\subsubsection*{Step 1: Problem Reduction}
Our first step is to reduce the problem to bounding some random processes. We   assume that   $\bm{\Delta_x}$ and  $\bm{\Delta_v}$ are non-zero with no loss of generality. 

\textbf{Using Optimality:} Starting with the optimality of $(\bm{\hat{x}},\bm{\hat{v}})$ that implies   \begin{align}
        \|\bm{\dot{y}}-\bm{\Phi \hat{x}}-\sqrt{m}\bm{\hat{v}}\|_2^2\leq \|\bm{\dot{y}}-\bm{\Phi x^\star}-\sqrt{m}\bm{v^\star}\|_2^2,
    \end{align}   we substitute $\bm{\hat{x}}=\bm{x^\star}+\bm{\Delta_x}$ and $\bm{\hat{v}}=\bm{v^\star}+\bm{\Delta_v}$ into the left-hand side and then expand the square to obtain  
  \begin{align}\label{3.7}
    \|\bm{\Phi\Delta_x}+\sqrt{m}\bm{\Delta_v}\|^2_2& \leq 2\langle\bm{\dot{y}}-(\bm{\Phi x^\star}+\sqrt{m}\bm{v^\star}),\bm{\Phi\Delta_x}+\sqrt{m}\bm{\Delta_v}\rangle\\\label{eq:sub_quan_noi}
        & = 2 \langle \bm{\epsilon},\bm{\Phi\Delta_x}+\sqrt{m}\bm{\Delta_v}\rangle + 2 \langle \bm{\xi}_{\bm{x^\star},\bm{v^\star}},\bm{\Phi\Delta_x}+\sqrt{m}\bm{\Delta_v}\rangle,
    \end{align}
    where in \cref{eq:sub_quan_noi} we substitute \cref{eq:quan_noise}, and recall that $\bm{\xi}_{\bm{x^\star},\bm{v^\star}}$ is the quantization noise associated with $(\bm{x^\star},\bm{v^\star})$ defined in 
    \cref{eq:xi_xv}. In pursuit of uniform error bound, we must ensure that each step proceeds universally for all $(\bm{x^\star},\bm{v^\star})\in \mathcal{K}_{\bm{x}}\times \mathcal{K}_{\bm{v}}$; for clarity, we will take supremum/infimum at this early stage, which requires us to identify contraint sets that accommodate the estimation errors. 
    
    \textbf{Identifying Constraint Sets:}  Because $f(\bm{\hat{x}})\leq f(\bm{x^\star})$, $g(\bm{\hat{v}})\leq g(\bm{v^\star})$, we have $\bm{\Delta_x}\in\mathcal{D}_f(\bm{x^\star})$ and $\bm{\Delta_v}\in \mathcal{D}_g(\bm{v^\star})$. Combining with $\mathcal{D}_{\bm{x}}$ and $\mathcal{D}_{\bm{v}}$ defined in \cref{3.91},   we have\begin{align}\label{eq:uni_constraint}
        (\bm{\Delta_x},\bm{\Delta_v})\in  \calD _{\bm{x},\bm{v}}:=\mathcal{D}_{\bm{x}}\times \mathcal{D}_{\bm{v}},~\forall (\bm{x^\star},\bm{v^\star})\in \mathcal{K}_{\bm{x}}\times \mathcal{K}_{\bm{v}}.
    \end{align}     Besides $\mathcal{D}_{\bm{x},\bm{v}}$ in \cref{eq:uni_constraint}, we further define its localized version $\mathcal{D}_{\bm{x},\bm{v}}^* = \mathcal{D}_{\bm{x},\bm{v}}\cap\mathbb{S}^{m+n-1}$; along with \cref{3.101}, we have that 
    \begin{align}
        \frac{\bm{\Delta}_{\bm{x}}}{\|\bm{\Delta}_{\bm{x}}\|_2} \in \mathcal{D}_{\bm{x}}^*,~\frac{\bm{\Delta}_{\bm{v}}}{\|\bm{\Delta}_{\bm{v}}\|_2} \in \mathcal{D}_{\bm{v}}^*,~\frac{(\bm{\Delta_x},\bm{\Delta_v})}{(\|\bm{\Delta_x}\|_2^2+\|\bm{\Delta_v}\|_2^2)^{1/2}} \in \mathcal{D}_{\bm{x},\bm{v}}^*
    \end{align}
    holds universally for all $(\bm{x^\star},\bm{v^\star})\in \mathcal{K}_{\bm{x},\bm{v}}$. 
    
    \textbf{Bounding Both Sides of \cref{3.7}:} Uniformly for all $(\bm{x^\star},\bm{v^\star})$, the left-hand side of \cref{3.7} is lower bounded by 
    \begin{align}
        \big(\|\bm{\Delta}_{\bm{x}}\|_2^2+\|\bm{\Delta}_{\bm{v}}\|_2^2\big)\cdot\inf_{(\bm{a},\bm{b})\in \mathcal{D}_{\bm{x},\bm{v}}^*}\big\|\bm{\Phi a}+\sqrt{m}\bm{b}\big\|_2^2:= \big(\|\bm{\Delta}_{\bm{x}}\|_2^2+\|\bm{\Delta}_{\bm{v}}\|_2^2\big)\cdot I_1,
    \end{align}
    and \cref{eq:sub_quan_noi} is (upper) bounded by 
    \begin{align}
       & 2\|\bm{\Delta_x}\|_2\cdot\sup_{\bm{c}\in \mathcal{D}_{\bm{x}}^*}\langle\bm{\epsilon},\bm{\Phi c}\rangle + 2\|\bm{\Delta_v}\|_2 \cdot\sup_{\bm{d}\in \mathcal{D}_{\bm{v}}^*}\langle\bm{\epsilon},\sqrt{m}\bm{d}\rangle \\&\quad+ 2 \big(\|\bm{\Delta_x}\|_2^2+\|\bm{\Delta_v}\|_2^2\big)^{1/2}\cdot \sup_{\bm{a}\in \mathcal{K}_{\bm{x}}}\sup_{\bm{b}\in \mathcal{K}_{\bm{v}}}\sup_{(\bm{c},\bm{d})\in \mathcal{D}_{\bm{x},\bm{v}}^*}\langle \bm{\xi}_{\bm{a},\bm{b}}, \bm{\Phi c}+\sqrt{m}\bm{d}\rangle \label{eq:qpe_to_bound}\\
       & \quad := 2\|\bm{\Delta_x}\|_2\cdot I_2 + 2\|\bm{\Delta_v}\|_2\cdot I_3 +  2 \big(\|\bm{\Delta_x}\|_2^2+\|\bm{\Delta_v}\|_2^2\big)^{1/2}\cdot I_4, 
    \end{align}
    where in \cref{eq:qpe_to_bound} we introduce a generic notation for quantization noise similar to \cref{eq:xi_xv}: 
    \begin{align}
        \label{eq:xiabfirst}
        \bm{\xi}_{\bm{a},\bm{b}} := \mathcal{Q}_\delta(\bm{\Phi a}+\sqrt{m}\bm{b}+\bm{\epsilon}+\bm{\tau})- (\bm{\Phi a}+\sqrt{m}\bm{b}+\bm{\epsilon}),
    \end{align}
    and for convenience we denote the random processes that arose by $I_1,I_2,I_3,I_4$. Therefore, we obtain that \begin{align}
        \label{I1234}\big(\|\bm{\Delta}_{\bm{x}}\|_2^2+\|\bm{\Delta}_{\bm{v}}\|_2^2\big)\cdot I_1\le 2\Big(\|\bm{\Delta_x}\|_2\cdot I_2 + \|\bm{\Delta_v}\|_2\cdot I_3 +   \big(\|\bm{\Delta_x}\|_2^2+\|\bm{\Delta_v}\|_2^2\big)^{1/2}\cdot I_4\Big)
    \end{align}
holds uniformly for all $(\bm{x^\star},\bm{v^\star})\in \mathcal{K}_{\bm{x}}\times \mathcal{K}_{\bm{v}}$. 

\subsubsection*{Step 2: Bounding $I_1,I_2,I_3,I_4$} With \cref{I1234}, all that remains is to bound $I_1,I_2,I_3,I_4$  from the correct side.
We provide a sketch of our techniques in this step:  
   \begin{itemize}
   [leftmargin=5ex,topsep=0.25ex]
       \item We apply the extended matrix deviation inequality (\cref{pro1}) to get a lower bound on $I_1$;

       \item We   apply \cref{pro4} with the randomness of $\bm{\Phi}$  to bound $I_2$; 

       \item We apply \cref{pro2} with the randomness of $\bm{\epsilon}$ to bound $I_3$; 
       \item We apply  the global QPE property for structured sets (\cref{coro:qpe_structured}) to bound $I_4$. 
   \end{itemize} 
\vspace{1mm}

\textbf{Bounding $I_1$:} 
 For any $t\ge 0$, \cref{pro1} yields that the event
 \begin{equation}\label{3.10}
     \sup_{(\bm{a},\bm{b})\in \mathcal{D}^*_{\bm{x},\bm{v}}}\big|\|\bm{\Phi a}+\sqrt{m}\bm{b}\|_2-\sqrt{m}\big|\leq C \big(\gamma(\mathcal{D}^*_{\bm{x},\bm{v}})+t\big)
 \end{equation}
 holds with probability  exceeding $ 1-\exp(-t^2)$. 
 Note that   \cref{pro7} gives $\gamma(\mathcal{D}^*_{\bm{x},\bm{v}})\asymp \gamma(\mathcal{D}^*_{\bm{x}})+\gamma\big(\mathcal{D}^*_{\bm{v}})$, so  the sample complexity \cref{eq:thm1_sam_size} implies $m\gtrsim \gamma^2(\mathcal{D}^*_{\bm{x},\bm{v}})$, and we  can set $t=\gamma(\mathcal{D}_{\bm{x},\bm{v}}^*)$ in \cref{3.10} and obtain that the event
 \begin{align}\label{eq:I1lower1}
\sup_{(\bm{a},\bm{b})\in\mathcal{D}_{\bm{x},\bm{v}}^*}\big|\|\bm{\Phi a}+\sqrt{m}\bm{b}\|_2-\sqrt{m}\big|\leq \frac{\sqrt{m}}{2}
 \end{align}
{\color{black}holds with     probability exceeding $1-\exp(-\gamma^2(\mathcal{D}_{\bm{x},\bm{v}}^*))$.} Combining with triangle inequality, \cref{eq:I1lower1} gives 
 \begin{align}
     \sqrt{I_1}&= \inf_{(\bm{a},\bm{b})\in \mathcal{D}_{\bm{x},\bm{v}}^*}\big\|\bm{\Phi a}+\sqrt{m}\bm{b}\big\|_2\\
     &\ge \sqrt{m} - \sup_{(\bm{a},\bm{b})\in\mathcal{D}_{\bm{x},\bm{v}}^*}\big|\|\bm{\Phi a}+\sqrt{m}\bm{b}\|_2-\sqrt{m}\big| \ge \frac{\sqrt{m}}{2}, \label{eq:bound_thm1I1}
 \end{align}
thus yielding the desired lower bound on $I_1$: $I_1\ge \frac{m}{4}$.

 \textbf{Bounding $I_2$:} Conditioning on $\bm{\epsilon}$, we invoke \cref{pro4} to obtain  that for any $t\ge 0$, the event   \begin{align} 
        I_{2}\leq C_1\|\bm{\epsilon}\|_2\big(\omega (\mathcal{D}^*_{\bm{x}})+t\big) 
   \end{align}
    holds with probability exceeding $ 1-2\exp(-t^2)$.  We set $t=\gamma(\mathcal{D}^*_{\bm{x},\bm{v}})$ to obtain that the event $I_{2}\leq C_1\|\bm{\epsilon}\|_{2}\cdot \omega(\mathcal{D}^*_{\bm{x},\bm{v}})$ holds with probability exceeding $1-2\exp(-\gamma^2(\mathcal{D}^*_{\bm{x},\bm{v}}))$. Then we deal with the randomness of $\bm{\epsilon}$. By $\|\bm{\epsilon}\|_{\psi_2}=O(E)$ from \cref{assump1},  we can use \cite[Exercise 6.3.5]{vershynin2018high} to bound $\|\bm{\epsilon}\|_2$ and obtain 
    \begin{equation}   \label{3.1999}\mathbbm{P}\big(\|\bm{\epsilon}\|_2\geq C_2E\sqrt{m}+t\big)\leq \exp\Big(-\frac{c_3t^2}{E^2}\Big)
   \end{equation}
   for any $t\ge 0$. 
   We set $t\asymp E\sqrt{m}$ to obtain that  $\|\bm{\epsilon}\|_2\lesssim E\sqrt{m}$ holds with  probability exceeding $1-\exp(-\Omega(m))$. Therefore, we arrive at the desired bound 
   \begin{align}\label{eq:bound_thm1_I2}
       I_2 \lesssim E\sqrt{m}\cdot \gamma(\mathcal{D}^*_{\bm{x},\bm{v}})
   \end{align}
   that holds with probability exceeding $1-2\exp(-\gamma^2(\mathcal{D}^*_{\bm{x},\bm{v}}))-\exp(-\Omega(m))$.

   \textbf{Bounding $I_3$:} For any $\bm{d}_1,\bm{d}_2\in \mathcal{D}_{\bm{v}}^*\cup\{0\}$ we have \begin{align}\label{eq:verify_proa2}
       \|\langle \bm{\epsilon},\sqrt{m}\bm{d}_1\rangle-\langle \bm{\epsilon},\sqrt{m}\bm{d}_2\rangle\|_{\psi_2}=\sqrt{m}\|\langle\bm{\epsilon},\bm{d}_1-\bm{d}_2\rangle\|_{\psi_2}\le E\sqrt{m}\|\bm{d}_1-\bm{d}_2\|_2.
   \end{align}
   Thus, \cref{pro2} implies that for any $t\ge 0$,  the event $I_3 \le C_4E\sqrt{m}(\omega(\mathcal{D}_{\bm{v}}^*)+t)$ holds with probability exceeding $1-2\exp(-t^2)$. Setting $t=\gamma(\mathcal{D}^*_{\bm{x},\bm{v}})$ gives 
   \begin{align}\label{eq:bound_thm1_I3}
       I_3 \lesssim E\sqrt{m}\cdot \gamma(\mathcal{D}^*_{\bm{x},\bm{v}})
   \end{align}
   that holds with probability at least $1-2\exp(-\gamma^2(\mathcal{D}^*_{\bm{x},\bm{v}}))$.

   \textbf{Bounding $I_4$:} This is the most challenging part in our analysis, but we leave the development of QPE to \cref{qpe} to allow for a clean analysis in the main body.  
   With \cref{eq:qpecon1}, \cref{eq:qpecon2} and \cref{eq:thm1_sam_size} we can apply \cref{coro:qpe_structured} to obtain that 
   \begin{align}\label{eq:bound_thm1_I4}
       I_4 \le C_5\delta\sqrt{m}\Big(\omega(\mathcal{D}^*_{\bm{x},\bm{v}})+\sqrt{\mathscr{H}(\mathcal{K}_{\bm{x}},\rho_1)+\mathscr{H}(\mathcal{K}_{\bm{v}},\rho_2)}\Big)
   \end{align}
  holds with probability at least $1-12\exp(-\Omega(\mathscr{H}(\mathcal{K}_{\bm{x}},\rho_1)+\mathscr{H}(\mathcal{K}_{\bm{v}},\rho_2)))$.

\subsubsection*{Step 3: Combining Everything}
We are in the position to combine everything together to conclude the proof. Substituting the bounds \cref{eq:bound_thm1I1} on $I_1$, \cref{eq:bound_thm1_I2}, \cref{eq:bound_thm1_I3} and \cref{eq:bound_thm1_I4} on $(I_2,I_3,I_4)$ into \cref{I1234} yields 
\begin{align}
    \nonumber&m\big(\|\bm{\Delta_x}\|_2^2+\|\bm{\Delta_v}\|_2^2\big) \\&\quad\quad\lesssim \sqrt{m}\big(\|\bm{\Delta_x}\|_2^2+\|\bm{\Delta_v}\|_2^2\big)^{1/2}\Big((E+\delta)\cdot\gamma(\mathcal{D}^*_{\bm{x},\bm{v}})+\delta\sqrt{\mathscr{H}(\mathcal{K}_{\bm{x}},\rho_1)+\mathscr{H}(\mathcal{K}_{\bm{v}},\rho_2)}\Big) 
\end{align}
that holds uniformly for all $(\bm{x^\star},\bm{v^\star})\in\mathcal{K}_{\bm{x}}\times \mathcal{K}_{\bm{v}}$.  Rearranging, along with $\gamma(\mathcal{D}^*_{\bm{x},\bm{v}})\asymp\gamma(\mathcal{D}^*_{\bm{x}})+\gamma(\mathcal{D}^*_{\bm{v}})$ from \cref{pro7},   yields the desired bound \cref{eq:thm1bound}. All that remains is to count the probability terms: We rule out probability terms of $\exp(-\gamma^2(\mathcal{D}^*_{\bm{x},\bm{v}}))$  to ensure \cref{eq:bound_thm1I1}, $2\exp(-\gamma^2(\mathcal{D}^*_{\bm{x},\bm{v}}))+\exp(-\Omega(m))$ for \cref{eq:bound_thm1_I2}, $2\exp(-\gamma^2(\mathcal{D}^*_{\bm{x},\bm{v}}))$ for \cref{eq:bound_thm1_I3}, $12\exp(-\Omega(\mathscr{H}(\mathcal{K}_{\bm{x}},\rho_1)$ $+\mathscr{H}(\mathcal{K}_{\bm{v}},\rho_2)))$ for \cref{eq:bound_thm1_I4}. Combining with $\gamma^2(\mathcal{D}^*_{\bm{x},\bm{v}})\asymp \gamma^2(\mathcal{D}^*_{\bm{x}})+\gamma^2(\mathcal{D}^*_{\bm{v}})$ and \cref{eq:thm1_sam_size}, we can promise that the uniform error bound holds with the probability stated in \cref{promisedthm1}. 
\end{proof}

\begin{rem}[The Cost of Uniformity: Constrained Lasso with Structured Priors] \label{rem1}
To see the implication of   \cref{thm1} on structured priors (\cref{assump3}), we substitute \cref{3.77} into \cref{eq:thm1bound} to obtain the uniform $\ell_2$-norm error bound
\begin{equation}        \label{3.35}\tilde{O}\left(\frac{(E+\delta)\big[\gamma(\mathcal{D}_{\bm{x}}^*)+\gamma(\mathcal{D}_{\bm{v}}^*)\big]+\delta\cdot\big[\omega(\mathcal{K}_{\bm{x}}) +\omega(\mathcal{K}_{\bm{v}})\big]}{\sqrt{m}}\right),
    \end{equation}
    where we use $\tilde{O}(\cdot)$ to omit some logarithmic factors on $(\rho_1,\rho_2)$. 
       We compare \cref{3.35} with   the non-uniform bound \cite[Thm. 1]{sun2022quantized} \begin{equation}
        O\left(\frac{(E+\delta)\big[\omega\big(\mathcal{D}^*_f(\bm{x^\star})\big)+\gamma\big(\mathcal{D}^*_g(\bm{v^\star})\big)\big]}{\sqrt{m}}\right)\label{sunnon}
    \end{equation} 
    and elaborate the cost of uniformity 
    by
    noting two differences:
    \begin{itemize}  [leftmargin=5ex,topsep=0.25ex]
        \item  First,  the term $\omega\big(\mathcal{D}_f^*(\bm{x^\star})\big)+\gamma\big(\mathcal{D}_g^*(\bm{v^\star})\big)$ regarding some fixed $(\bm{x^\star},\bm{v^\star})$ in \cref{sunnon} is substituted with     $\gamma(\mathcal{D}_{\bm{x}}^*)+\gamma(\mathcal{D}_{\bm{v}}^*)$ in \cref{3.35}, where $\mathcal{D}_{\bm{x}}^*$ and $\mathcal{D}_{\bm{v}}^*$ are defined in \cref{3.101}. This appears rather natural to us --- in uniform recovery, the local complexity quantity (e.g., $\omega(\mathcal{D}_f^*(\bm{x^\star}))$) upgrades to a global one (e.g., $\omega(\mathcal{D}^*_{\bm{x}})$) concerning   all signals and corruptions.
        \item Second, our uniform bound  also presents the additional term $\tilde{O}\big(\frac{\delta}{\sqrt{m}}\big[\omega(\mathcal{K}_{\bm{x}})+\omega(\mathcal{K}_{\bm{v}})\big]\big)$. Note that this term vanishes in   non-uniform recovery where $\mathcal{K}_{\bm{x}}$ and $\mathcal{K}_{\bm{v}}$ only contain one point. Thus, we can think of that \cref{sunnon} also implicitly includes this term, and from this viewpoint the two bounds \cref{3.35} and \cref{sunnon} stay consistent.  
    \end{itemize}
    We further note that similar phenomena were also observed from the results of \cite{genzel2022unified} (e.g., Theorem 1 therein). As we shall see, in the most interesting cases of structured priors (such as sparsity and low-rank), \cref{3.35} and \cref{sunnon}   are  typically of the same scaling up to logarithmic factors, indicating that the uniformity costs very little. 
\end{rem}
\begin{rem}[The  Role of Quantization Resolution $\delta$]\label{rem2}
  Due to random dithering, the quantization resolution $\delta$ appears in \cref{eq:thm1bound} as multiplicative factors,   agreeing with similar findings in \cite{chen2022quantizing,xu2020quantized,chen2023quantized,thrampoulidis2020generalized,sun2022quantized} and confirming the  intuition that the recovery worsens under coarser quantization (i.e., larger $\delta$). 
\end{rem}

To illustrate the implications of   \cref{thm1}, we provide two concrete examples in \cref{coro1} and~\cref{coro2}. The proofs of these corollaries can be found in   \cref{appendixc}. Equipped with the general \cref{thm1}, the proofs can be done by selecting $(\rho_1,\rho_2)$ and estimating the geometric quantities. Due to the feature of structured signal sets as per \cref{3.77}, using some extremely small $(\rho_1,\rho_2)$ to render \cref{eq:qpecon2} only leads to logarithmic degradation. That being said, we still (slightly) refine the choice of $(\rho_1,\rho_2)$ to lessen the logarithmic factor (our delicate QPE result \cref{coro:qpe_structured} allows us to do so); see \cref{rem:elimination}, \cref{rem:simplifi} below.

\begin{corollary}[Sparse Signal and Sparse Corruption]
\label{coro1}
 Under  \cref{assump1}--\cref{assump3}, we assume that $\bm{x^\star}$ is $s$-sparse and $\bm{v^\star}$ is $k$-sparse, i.e., $\mathcal{K}_{\bm{x}}=\Sigma^n_s\cap \mathbb{B}_2^n$ and  $\mathcal{K}_{\bm{v}}=\Sigma^m_k\cap \mathbb{B}_2^m$ in   \cref{assump3}, and accordingly we use $f(\bm{x})=\|\bm{x}\|_1$ and $g(\bm{v})=\|\bm{v}\|_1$ in   \cref{assump2}. If  $ m\ge C_1 s\log (\frac{nm^{3/2}}{s^{5/2}\delta})+C_1k\log(\frac{m^{5/2}}{k^{5/2}\delta})$ holds with large enough $C_1$, then with probability exceeding $
    1- C_2\exp(-\Omega(s\log\frac{en}{s}+k\log\frac{em}{k}))$
  on a single draw of $(\bm{\Phi},\bm{\epsilon},\bm{\tau})$,
 the following uniform error bound holds true for all $(\bm{x^\star},\bm{v^\star})\in \mathcal{K}_{\bm{x}}\times \mathcal{K}_{\bm{v}}$:
 \begin{equation}   \label{boundcoro1}\sqrt{\|\bm{\Delta_x}\|_2^2+\|\bm{\Delta_v}\|_2^2} \lesssim \frac{E\sqrt{s\log\frac{en}{s}+k\log\frac{em}{k}}+\delta \sqrt{s\log\frac{nm^{3/2}}{s^{5/2}\delta}+k\log\frac{m^{5/2}}{k^{5/2}\delta}}}{\sqrt{m}},
\end{equation}
where $\bm{\Delta_x}=\bm{\hat{x}}-\bm{x^\star}$, $\bm{\Delta_v}=\bm{\hat{v}}-\bm{v^\star}$, with $(\bm{\hat{x}},\bm{\hat{v}})$ being the solution to \cref{3.4}.
\end{corollary}
\begin{rem}[Elimination of Logarithmic Factors]\label{rem:elimination}
Provided the  additional scaling conditions   $m\lesssim n$ and $\delta \gtrsim \big(\frac{\min\{s,k\}}{m}\big)^N$ for some positive integer $N$ 
(note that they are very mild and cover most interesting settings), we have $\log(\frac{nm^{3/2}}{s^{5/2}\delta})\lesssim \log(\frac{n}{s})$ and $\log(\frac{m^{5/2}}{k^{5/2}\delta})\lesssim \log (\frac{m}{k})$.
Thus notably, our \cref{coro1}     coincides with its non-uniform counterpart in \cite[Coro. 1]{sun2022quantized} without suffering from   logarithmic degradation (rather, the cost is at most a larger multiplicative constant which both works do not aim to refine).  
\end{rem}
\begin{corollary}[Low-Rank Signal and Sparse Corruption]
\label{coro2}
Under  \cref{assump1}--\cref{assump3}, we assume that $\bm{x^\star}\in \mathbb{R}^{p\times q}$ is of rank no greater than $r$,\footnote{When substituted into \cref{3.1}, we view $\bm{x^\star}$ as a $(n:=pq)$-dimensional vector by vectorization.} and $\bm{v^\star}$ is $k$-sparse, i.e., $\mathcal{K}_{\bm{x}}=M^{p,q}_{r}\cap \mathbb{B}_\mathrm{F}^{p,q}$ and  $\mathcal{K}_{\bm{v}}=\Sigma^m_k\cap\mathbb{B}^m_2$ in \cref{assump3}, and accordingly we use $f(\bm{x})=\|\bm{x}\|_{\mathrm{nu}}$ and $g(\bm{v})=\|\bm{v}\|_1$ in \cref{assump2}. If $m\ge C_1 r(p+q)\log(\frac{m^{3/2}}{\delta(r(p+q))^{3/2}})+ C_1k\log(\frac{m^{5/2}}{k^{5/2}\delta})$ for some large enough $C_1$,
    then with probability exceeding $1-C_2\exp(-\Omega(r(p+q)+k\log\frac{em}{k}))$ on a single draw of $(\bm{\Phi},\bm{\epsilon},\bm{\tau})$,   
    the following uniform error bound holds true  for all $(\bm{x^\star},\bm{v^\star})\in \mathcal{K}_{\bm{x}}\times \mathcal{K}_{\bm{v}}$:
    \begin{equation}   \label{eq:bound_low_rank}\sqrt{\|\bm{\Delta_x}\|_{\rm F}^2+\|\bm{\Delta_v}\|_2^2}\lesssim \frac{E\sqrt{r(p+q)+k\log\frac{em}{k}}+\delta\sqrt{r(p+q)\log(\frac{m^{3/2}}{\delta(r(p+q))^{3/2}})+k\log(\frac{m^{5/2}}{k^{5/2}\delta})}}{\sqrt{m}}, 
    \end{equation} 
    where $\bm{\Delta_x}=\bm{\hat{x}}-\bm{x^\star}$, $\bm{\Delta_v}=\bm{\hat{v}}-\bm{v^\star}$, with $(\bm{\hat{x}},\bm{\hat{v}})$ being the solution to \cref{3.4}.
\end{corollary}
\begin{rem} \label{rem:simplifi}
Analogously to \cref{rem:elimination}, provided the additional scaling conditions of $\delta\gtrsim(\frac{r(p+q)}{m})^N$ and $\delta\gtrsim(\frac{k}{m})^N$ for some positive integer $N$, one can further simplify \cref{eq:bound_low_rank} to 
\begin{align}
    \frac{\delta\sqrt{r(p+q)\log(\frac{m}{r(p+q)})}+E\sqrt{r(p+q)}+(\delta+E)\sqrt{k\log\frac{em}{k}}}{\sqrt{m}}.
\end{align}
This only exhibits an additional factor of $\log^{1/2}(\frac{m}{r(p+q)})$ compared to the non-uniform counterpart in \cite[Coro. 2]{sun2022quantized}.   
\end{rem}

We close this subsection by comparing with relevant results and claiming our contributions.
\begin{rem}[Related Works and the Novelty of Our Results]\label{rem3}
      Restricted to Lasso,  the prior developments are as follows:
      \begin{itemize}
      [leftmargin=5ex,topsep=0.25ex]
          \item Non-uniform guarantees were presented in \cite[Thm. III.1]{thrampoulidis2020generalized} for compressed sensing and in \cite[Thm. 1]{sun2022quantized} for corrupted sensing.
          \item The only existing uniform guarantee for constrained Lasso was obtained in \cite[Coro. 4]{genzel2022unified},  but only applies to classical compressed sensing (without the need of recovering an additional structured corruption) and typically yields an error rate of $O(\sqrt{\delta}[\frac{\omega^2(\mathcal{K}_{\bm{x}})}{m}]^{1/4})$, which is slower than our \cref{3.35} under structured priors (\cref{assump3}). We will further note that it may not be sensible (if not impossible) to follow the proof technique in \cite{genzel2022unified} to prove \cref{thm1}; see \cref{rem:insensible}. 
      \end{itemize}
    Note that constrained Lasso is   a general recipe for nonlinear compressed sensing models \cite{plan2016generalized,liu2020generalized,genzel2016high}, but there have also been uniform guarantees for other recovery methods (see also the less extensive discussion in \cite[Sec. 4.3]{genzel2022unified}):
    \begin{itemize}
      [leftmargin=5ex,topsep=0.25ex]
        \item Jung et al. devised and theoretically analyzed a more specialized recovery method for quantized compressed sensing \cite[Thm. 3]{jung2021quantized}. However,  under the dithered uniform quantizer, their uniform error rate translates into $O\big(\big[\frac{\delta\omega^2(\mathcal{K}_{\bm{x}})}{m}\big]^{1/3}\big)$ in the worst case, and it is unknown whether their result can be sharpened for structured sets like $\Sigma_s^n$ since their statement requires $\mathcal{K}_{\bm{x}}$ to be convex;

        \item    Xu and Jacques analyzed the projected-back projection (PBP) estimator in \cite[Sec. 7.3\textbf{A}]{xu2020quantized}, providing a rate of   $\tilde{O}\big((1+\delta)\frac{\omega(\mathcal{K}_{\bm{x}})}{\sqrt{m}}\big)$ (logarithmic factors omitted) for structured sets. Although the rate is comparable to ours, one downside of the PBP estimator is that it does not achieve exact reconstruction  in a noiseless unquantized case \cite[Sec. 7.3\textbf{C}]{xu2020quantized}. (In contrast, Lasso achieves exact reconstruction in a noiselesss unquantized case; see \cite{chen2018stable,foygel2014corrupted} for instance.) 
    \end{itemize} 
    In a nutshell, in compressed/corrupted sensing associated with the dithered uniform quantizer, our \cref{thm1} presents the {\it sharpest} uniform error rate over structured set (\cref{defi1}), and note that this is achieved by   constrained Lasso which returns exact recovery in a noisyless unquantized setting.  
   \end{rem}
    \begin{rem}[Optimality]
    We note that the rates in \cref{coro1} and \cref{coro2} are near minimax optimal when the sub-Gaussian noise $\bm{\epsilon}$ is severer than the quantization noise (i.e., $E\gtrsim \delta$). Such optimality is implied by 
    adding together the lower bounds (e.g., from \cite[Thm. 4.2]{plan2017high}) for the two  estimation problems $\bm{y}_1 = \bm{\Phi} \bm{x^\star}+\bm{\epsilon}$ and 
    $\bm{y}_2 = \sqrt{m}\bm{v^\star}+\bm{\epsilon}$.\footnote{The minimax lower bound for $\bm{y}_1=\bm{\Phi}\bm{x^\star}+\bm{\epsilon}$ applies to the   estimation of $\bm{x^\star}$ from \cref{3.2}, since the additional corruption and quantization can only decrease our ability to estimate $\bm{x^\star}$. Similarly, the minimax lower bound for $\bm{y}_2=\sqrt{m}\bm{v^\star}+\bm{\epsilon}$ stands when estimating $\bm{v^\star}$ from \cref{3.2}.} 
    Nonetheless, for the noiseless case with quantization (i.e., $E=0$, $\delta>0$), the information theoretic limit exhibits a  decaying rate of $O(m^{-1})$ that is faster than our $O(m^{-1/2})$
 (e.g., \cite{boufounos2015quantization}), and we suspect that such faster rate cannot be achieved by Lasso due to some fundamental performance limit.\footnote{Though we are not aware of a rigorous analysis, to our best knowledge, all proved rates for {\it Lasso} in   quantized compressed sensing with memorylesss quantizer  are no faster than  $O(m^{-1/2})$; see   similar discussion in \cite[P. 34]{genzel2022unified}.} 
\end{rem}
\begin{rem}[Technical Comparison with \cite{genzel2022unified}]
    \label{rem:insensible}
    Genzel and Stollenwerk \cite{genzel2022unified} developed a general strategy to achieve uniform recovery, which consists of two ingredients: (i) Constructing Lipschitz approximation for handling the   discontinuity of $f_i(\cdot)$ (if any), and (ii) Applying the concentration inequality \cite[Thm. 8]{genzel2022unified} (due to Mendelson \cite{mendelson2016upper}) to bounding the product processes arising in the analysis. 
     However, it might not be sensible (if not impossible) to follow their techniques to prove \cref{thm1} for two reasons. First, when $f_i(\cdot)$ contains some discontinuity (as with our      \cref{eq:problem}), their general strategy leads to a uniform decaying rate of $O(m^{-1/4})$ (see \cite[Sec. 4]{genzel2022unified}) that is inferior to our $O(m^{-1/2})$, and it is unclear how to get faster uniform rate without incorporating existing embedding result available in the literature.\footnote{This is a workaround proposed in \cite[Sec. 5]{genzel2022unified}. Nonetheless, the needed embedding result may not exist in the literature for the problem at hand, and this is the case for our quantized corrupted sensing problem.} Second, the appearance of the corruption $\bm{v^\star}$ poses additional hurdle to the approach in \cite{genzel2022unified}, since $\bm{v^\star}$ leads to random processes beyond the scope of \cite[Thm. 8]{genzel2022unified}.  Conversely, our work of   getting $O(m^{-1/2})$ decaying rate based on QPE suggests the possibility of improving the slow rate of $O(m^{1/4})$ in \cite{genzel2022unified} under more general discontinuous $f_i(\cdot)$ --- one may deal with discontinuity of $f_i(\cdot)$ by proper {\it product embedding} property  (or {\it limited projection distortion} as termed by \cite{xu2020quantized}) rather than constructing Lipschitz approximation. 
\end{rem}
\subsection{Structured Priors with Unconstrained Lasso}
We now turn our attention to unconstrained Lasso \cref{3.5}, which is more practical than the constrained Lasso in \cref{3.4} in the sense that   it does not require prior estimates of $\big(f(\bm{x}^*),g(\bm{v^}\star)\big)$. Rather, as we shall see, a fixed large enough choice of the regularization parameters $(\lambda_1,\lambda_2)$ works  uniformly for all pairs of signal and corruption. 
To proceed, we first define the restricted compatibility constant between $f(\cdot)$ and $\ell_2$-norm over some constraint set $\mathcal{X}\subset \mathbb{R}^n$ as \begin{align}\label{eq:compatible}
    \alpha_f(\mathcal{X})= \sup_{\bm{x}\in \mathcal{X}}\frac{f(\bm{x})}{\|\bm{x}\|_2}.
\end{align}
Compared to constrained Lasso, the analysis of unconstrained Lasso is more technical in the following senses:
\begin{itemize}
  [leftmargin=5ex,topsep=0.25ex]
    \item The derivation of a low-complexity constraint set that contains the estimation error  becomes non-trivial, as contrasted to the straightforward $\bm{\Delta_x}\in\mathcal{D}_f(\bm{x^\star})$ and $\bm{\Delta_v}\in\mathcal{D}_g(\bm{v^\star})$ for constrained Lasso; 
    \item  Some other additional efforts are needed, e.g.,   bounding the compatibility constant that is in general technically challenging. 
\end{itemize}

Without pursuing full generality, 
 we make the following decomposable assumption on $f(\cdot)$ and $g(\cdot)$ to facilitate the estimation of compatibility constant.  

\begin{assumption}[Decomposable Norm]\label{assump4} Regarding the   sets $\mathcal{K}_{\bm{x}},\mathcal{K}_{\bm{v}}$   and associated norms $f(\cdot),g(\cdot)$  for promoting certain structure (\cref{assump2}), we assume that:
 \begin{itemize}
 [leftmargin=5ex,topsep=0.25ex]
     \item {\rm (Decomposibility)} Given any $\bm{a}\in \mathcal{K}_{\bm{x}}$, there exists a pair of linear subspaces $(\mathcal{X}_{\bm{a}},\overline{\mathcal{X}}_{\bm{a}})$  (possibly depending on $\bm{a}$) with $\mathcal{X}_{\bm{a}}\subset \overline{\mathcal{X}}_{\bm{a}}$, such that $\bm{a}\in \mathcal{X}_{\bm{a}}$, and  $f(\cdot)$   is decomposable over $(\mathcal{X}_{\bm{a}},\overline{\mathcal{X}}_{\bm{a}}^{\bot})$:\footnote{For a given linear subapce $\mathcal{X}$,   we denote  its orthogonal complement by $\mathcal{X}^\bot$.}
    \begin{equation}
        \begin{aligned}
           \label{3.39} &f(\bm{x}_1+\bm{x}_2)=f(\bm{x}_1)+f(\bm{x}_2),~~\forall~ \bm{x}_1\in \mathcal{X}_{\bm{a}},\bm{x}_2\in \overline{\mathcal{X}}_{\bm{a}}^\bot.
        \end{aligned}
    \end{equation} 
    Similarly, given any $\bm{b}\in \mathcal{K}_{\bm{v}}$ there exists a pair of linear subspaces $(\mathcal{V}_{\bm{b}},\overline{\mathcal{V}}_{\bm{b}})$ (possibly depending on $\bm{b}$) with $\mathcal{V}_{\bm{b}}\subset\overline{\mathcal{V}}_{\bm{b}}$, such that $\bm{b}\in \mathcal{V}_{\bm{b}}$, and the decomposibility $g(\bm{v}_1+\bm{v}_2)=g(\bm{v}_1)+g(\bm{v}_2)$ holds for any $\bm{v}_1\in \mathcal{V}_{\bm{a}}$, $\bm{v}_2\in\overline{\mathcal{V}}^\bot_{\bm{b}}$.
    
    \item {\rm(Uniform Bound on Compatibility Constant)} There exist $\alpha_{\bm{x}}$ and $\alpha_{\bm{v}}$ such that $\alpha_f(\overline{\mathcal{X}}_{\bm{a}})\leq \alpha_{\bm{x}}$ holds uniformly for all $\bm{a}\in \mathcal{K}_{\bm{x}}$, and that $\alpha_g(\overline{\mathcal{V}}_{\bm{b}})\leq \alpha_{\bm{v}}$ holds uniformly for all $\bm{b}\in \mathcal{K}_{\bm{v}}$, where $(\overline{\mathcal{X}}_{\bm{a}},\overline{\mathcal{V}}_{\bm{b}})$ are the   linear subspaces identified for a specific $(\bm{a},\bm{b})\in\mathcal{K}_{\bm{x}}\times \mathcal{K}_{\bm{v}}$ in the preceding dot point, $\alpha_f(\overline{\mathcal{X}}_{\bm{a}})$ and $\alpha_g(\overline{\mathcal{V}}_{\bm{b}})$ are the compatibility constants defined as per \cref{eq:compatible}.
 \end{itemize}
\end{assumption}

It is well-known that this decomposibility assumption is    satisfied   $\|\cdot\|_1$, $\|\cdot\|_{\rm nu}$,  $\|\cdot\|_{\ell_1/\ell_2}$ and so on   \cite{negahban2012restricted,negahban2012unified}, thus covering the most interesting cases of structured priors. As a canonical example, for  $s$-sparse structured prior together with $f(\bm{x})=\|\bm{x}\|_1$, we will let $\mathcal{X}_{\bm{a}}=\overline{\mathcal{X}}_{\bm{a}}=\{\bm{v}\in \mathbb{R}^n:\supp(\bm{v})\subset \supp(\bm{a})\}$ for any $\bm{a}\in \mathcal{K}_{\bm{x}}=\Sigma^n_s\cap \mathbb{B}_2^n$, under which it is evident that $\alpha_{\bm{x}}=\sqrt{s}$ is a uniform bound on $\alpha_f(\overline{\mathcal{X}}_{\bm{a}})$ (see more details in the proof of \cref{coro3}).

\begin{theorem}[Uniform Recovery via Unconstrained Lasso]
    \label{thm2}
    Under \cref{assump1}--\cref{assump4},   suppose that the positive scalars $(\zeta,\rho_1,\rho_2)$ and the sample size $m$ satisfy 
\begin{gather}\label{eq:qpecon1un}
    \zeta  = \frac{4\delta(\mathscr{H}(\mathcal{K}_{\bm{x}},\rho_1)+\mathscr{H}(\mathcal{K}_{\bm{v}},\rho_2))}{m},\\\label{eq:qpecon2un}
    \rho_1 \le \frac{c_1\zeta}{(\log\frac{\delta}{\zeta})^{1/2}},~\omega\big((\mathcal{K}_{\bm{x}})^{(\rho_1)}_{\loc}\big)\le c_2\zeta\sqrt{\frac{m\zeta}{\delta}},~\rho_2\le c_3\zeta\sqrt{\frac{\zeta}{\delta}} 
\end{gather}
for  small enough $(c_1,c_2,c_3)$, and {\color{black}we also suppose that $f(\bm{x})\geq \|\bm{x}\|_2$ holds for any $\bm{x}\in \mathbb{R}^n$, $g(\bm{v})\geq \|\bm{v}\|_2$ holds for any $\bm{v}\in \mathbb{R}^m$.}\footnote{This   is a very mild condition because $f(\cdot)$ and $g(\cdot)$ are norms that promote low-complexity structure, thus naturally dominating $\ell_2$-norm; see, e.g., \cite{raskutti2019convex}.}  We set
    \begin{align}\label{eq:lam1choice}
            &\lambda_1= C_4 (E+\delta)\sqrt{m}\cdot \omega(\mathbb{B}_f^n)+ C_4\delta \sqrt{m}\cdot \sqrt{\mathscr{H}(\mathcal{K}_{\bm{x}},\rho_1)+\mathscr{H}(\mathcal{K}_{\bm{v}},\rho_2)},
            \\&\lambda_2=C_5(E+\delta)\sqrt{m}
            \cdot \omega(\mathbb{B}_g^m)+ C_5\delta \sqrt{m}\cdot \sqrt{\mathscr{H}(\mathcal{K}_{\bm{x}},\rho_1)+\mathscr{H}(\mathcal{K}_{\bm{v}},\rho_2)},\label{eq:lam2choice}
        \end{align}
    for some large enough $C_4,C_5$. 
   If for some sufficiently large implied constant, it holds that 
    \begin{align}\label{eq:thm2sample}
        m\gtrsim  \Big(\alpha_{\bm{x}}+\frac{\lambda_2\alpha_{\bm{v}}}{\lambda_1}\Big)^2\omega^2(\mathbb{B}_f^n)+\Big(\alpha_{\bm{v}}+\frac{\lambda_1\alpha_{\bm{x}}}{\lambda_2}\Big)^2\omega^2(\mathbb{B}^m_g) +\mathscr{H}\big(\mathcal{K}_{\bm{x}},\rho_1\big)+\mathscr{H}\big(\mathcal{K}_{\bm{v}},\rho_2\big),
        \end{align}
    then  with probability exceeding 
    \begin{align}\label{eq:thm2_prob}
        1-C_6\exp\big(-c_7\min\big\{\omega^2(\mathbb{B}_f^n),\omega^2(\mathbb{B}_g^m)\big\}\big)-24 
     \exp\big(-\Omega(\mathscr{H}(\mathcal{K}_{\bm{x}},\rho_1)+\mathscr{H}(\mathcal{K}_{\bm{v}},\rho_2))\big)
     \end{align} 
     on a single draw of $(\bm{\Phi},\bm{\epsilon},\bm{\tau})$,     the following uniform error bound holds true for all $(\bm{x^\star},\bm{v^\star})\in\mathcal{K}_{\bm{x}}\times \mathcal{K}_{\bm{v}}$:
    \begin{equation}\label{thm2bound}
\sqrt{\|\bm{\Delta_x}\|_2^2+\|\bm{\Delta_v}\|_2^2}\lesssim \frac{\lambda_1\alpha_{\bm{x}}+\lambda_2\alpha_{\bm{v}}}{m},
    \end{equation}
    where $\bm{\Delta_x}=\bm{\hat{x}}-\bm{x^\star}$ and $\bm{\Delta_v}=\bm{\hat{v}}-\bm{v^\star}$, with $(\bm{\hat{x}},\bm{\hat{v}})$ being the solution  to \cref{3.5}
\end{theorem}

\begin{proof}
    We   assume that $\bm{\Delta_x}$ and $\bm{\Delta_v}$ are non-zero without loss of generality. Note that the optimality of $(\bm{\hat{x}},\bm{\hat{v}})$ gives
    \begin{align}\label{eq:un_optimal}
    \|\bm{\dot{y}}-\bm{\Phi \hat{x}}-\sqrt{m}\bm{\hat{v}}\|_2^2+\lambda_1f(\bm{\hat{x}})+\lambda_2g(\bm{\hat{v}})\leq \|\bm{\dot{y}}-\bm{\Phi x^\star}-\sqrt{m}\bm{v^\star}\|_2^2+\lambda_1f(\bm{x^\star})+\lambda_2g(\bm{v^\star}).     
    \end{align}
     We perform some   calculation and reformulate the inequality as follows: 
        \begin{align}
            \label{3.44}&\|\bm{\Phi\Delta_x}+\sqrt{m}\bm{\Delta_v}\|_2^2\leq 2\langle\bm{\dot{y}}-\bm{\Phi x^\star}-\sqrt{m}\bm{v^\star},\bm{\Phi\Delta_x}+\sqrt{m}\bm{\Delta_v}\rangle\\\label{eq:351}
            &~~~~~~~~~~~~~~~~~~~\quad\quad\quad\quad+\lambda_1\big(f(\bm{x^\star})-f(\bm{\hat{x}})\big)+\lambda_2\big(g(\bm{v^\star})-g(\bm{\hat{v}})\big)\\\label{eq:2nd_ine}
            &\leq 2\Big(f(\bm{\Delta_x})\cdot\sup_{\bm{c}\in \mathbb{B}_f^n}\langle\bm{\epsilon}+\bm{\xi}_{\bm{x^\star},\bm{v^\star}},\bm{\Phi c}\rangle+ g(\bm{\Delta_v})\cdot \sup_{\bm{d}\in \mathbb{B}_g^m}\langle\bm{\epsilon}+\bm{\xi}_{\bm{x^\star},\bm{v^\star}},\sqrt{m}\bm{d}\rangle\Big)\\
        &~~~~~~~~~~~~~~~~~~~\quad\quad\quad\quad+\lambda_1\big(f(\bm{x^\star})-f(\bm{\hat{x}})\big)+\lambda_2\big(g(\bm{v^\star})-g(\bm{\hat{v}})\big)\\
        &\le 2 \Big(f(\bm{\Delta_x})\cdot \big[I_1+I_2\big]+ g(\bm{\Delta_v})\cdot\big[I_3+I_4\big]\Big)+ \lambda_1\big(f(\bm{x^\star})-f(\bm{\hat{x}})\big)+\lambda_2\big(g(\bm{v^\star})-g(\bm{\hat{v}})\big) \label{eq:last_ineq}
        \end{align}
    where the first inequality \cref{3.44}--\cref{eq:351} is obtained from \cref{eq:un_optimal} by substituting $\bm{\hat{x}}= \bm{x^\star}+\bm{\Delta_x}$ and $\bm{\hat{v}}= \bm{v^\star}+\bm{\Delta_v}$ into $\|\bm{\dot{y}}-\bm{\Phi\hat{x}}-\sqrt{m}\bm{\hat{v}}\|_2^2$ and then expanding the square; then, in \cref{eq:2nd_ine} we substitute $\bm{\dot{y}}-\bm{\Phi x^\star}-\sqrt{m}\bm{v^\star}=\bm{\epsilon}+\bm{\xi}_{\bm{x^\star},\bm{v^\star}}$ \cref{eq:quan_noise} and then take the supremum over $\bm{c} = \frac{\bm{\Delta_x}}{f(\bm{\Delta_x})}\in \mathbb{B}_f^n$ and $\bm{d} = \frac{\bm{\Delta_v}}{g(\bm{\Delta_v})}\in \mathbb{B}_g^m$; moreover, in \cref{eq:last_ineq} we further take the supremum with respect to $\bm{\xi}_{\bm{x^\star},\bm{v^\star}}$ over $(\bm{x^\star},\bm{v^\star})\in \mathcal{K}_{\bm{x}}\times \mathcal{K}_{\bm{v}}$ and introduce the shorthand 
    \begin{gather}\label{eq:I1I2_uncon}
         I_1: = \sup_{\bm{c}\in \mathbb{B}_f^n}\langle \bm{\epsilon},\bm{\Phi c}\rangle,~I_2 := \sup_{\bm{a}\in \mathcal{K}_{\bm{x}}}\sup_{\bm{b}\in \mathcal{K}_{\bm{v}}}\sup_{\bm{c}\in \mathbb{B}_f^n}\langle\bm{\xi}_{\bm{a},\bm{b}},\bm{\Phi c}\rangle,\\\label{eq:I3I4_uncon}
         I_3:=\sup_{\bm{d}\in \mathbb{B}_g^m}\langle\bm{\epsilon},\sqrt{m}\bm{d}\rangle,~I_4 := \sup_{\bm{a}\in \mathcal{K}_{\bm{x}}}\sup_{\bm{b}\in \mathcal{K}_{\bm{v}}}\sup_{\bm{d}\in \mathbb{B}_g^m}\langle \bm{\xi}_{\bm{a},\bm{b}},\sqrt{m}\bm{d}\rangle,
    \end{gather}
    where $\bm{\xi}_{\bm{a},\bm{b}}$ is the quantization noise as per \cref{eq:xiabfirst}. We pause to provide an outline of the remainder of this proof:
    \begin{itemize}
    [leftmargin=5ex,topsep=0.25ex]
        \item {\bf Step 1:}  Bounding  $I_1,I_2,I_3,I_4$ in \cref{eq:I1I2_uncon}--\cref{eq:I3I4_uncon}
        by the techniques similar to those in the proof of \cref{thm1}.  The high-probability bounds imply  $\lambda_1\geq 4(I_1+I_2)$, $\lambda_2\geq 4(I_3+I_4)$; 

        \item {\bf Step 2:} Based on   \cref{assump4},   identifying the constraint set  where $(\bm{\Delta_x},\bm{\Delta_v})$ lives and establishing a uniform lower bound on $\|\bm{\Phi\Delta_x}+\sqrt{m}\bm{\Delta_v}\|_2^2$ via \cref{pro1};  

        \item {\bf Step 3:}   Combining everything to conclude the proof. 
    \end{itemize}

    \subsubsection*{Step 1: Bounding $I_1,I_2,I_3,I_4$} 
    Similarly to the proof of \cref{thm1}, we bound the error terms associated with $\bm{\epsilon}$ (i.e., $I_1,I_3$) via \cref{pro2} and \cref{pro4}, and bound the terms associated with quantization noise (i.e., $I_2,I_4$) via QPE property specialized to structured sets (\cref{coro:qpe_structured}).

    {\bf Bounding $I_{1}$:} Note that 
    $\mathbb{B}_f^n\subset \mathbb{B}_2^n$ since $f(\bm{x})\geq \|\bm{x}\|_2$ holds for all $\bm{x}\in \mathbb{R}^n$. 
    Conditioning on $\bm{\epsilon}$,  \cref{pro4} provides that $I_{1}\leq C_1\|\bm{\epsilon}\|_2\cdot \omega(\mathbb{B}_f^n)$ holds with probability exceeding $1-\exp(-\omega^2(\mathbb{B}_f^n))$. Moreover, we can still bound $\|\bm{\epsilon}\|_2$ as in \cref{3.1999}, which implies $\|\bm{\epsilon}\|_2=O(E\sqrt{m})$ with probability exceeding $1-\exp(-\Omega(m))$. Thus, the bound on $I_1$ \begin{align}
        \label{eq:un_I1_bound}
        I_1 \lesssim E \sqrt{m}\cdot\omega(\mathbb{B}_f^n)
    \end{align} 
    holds with probability exceeding $1-\exp(-\omega^2(\mathbb{B}_f^n))-\exp(-\Omega(m))$.

   {\bf Bounding $I_{2}$:} With the assumptions \cref{eq:qpecon1un}--\cref{eq:qpecon2un} and \cref{eq:thm2sample}, we can invoke \cref{coro:qpe_structured} with $(\mathcal{A},\mathcal{B},\mathcal{E}) = (\mathcal{K}_{\bm{x}},\mathcal{K}_{\bm{v}},\mathbb{B}_f^n\times\{0\})$, yielding that the event \begin{align}
       I_2 \lesssim \delta\sqrt{m}\cdot \Big(\omega(\mathbb{B}_f^n)+\sqrt{\mathscr{H}(\mathcal{K}_{\bm{x}},\rho_1)+\mathscr{H}(\mathcal{K}_{\bm{v}},\rho_2)}\Big)\label{eq:bound_I2_un}
   \end{align}
    holds with probability exceeding $1-12\exp(-\Omega(\mathscr{H}(\mathcal{K}_{\bm{x}},\rho_1)+\mathscr{H}(\mathcal{K}_{\bm{v}},\rho_2)))$. 
    
    \textbf{Bounding $I_{3}$:} We note that  $\mathbb{B}_g^m\subset \mathbb{B}_2^m$ due to $g(\bm{v})\geq \|\bm{v}\|_2$ holds for any $\bm{v}\in \mathbb{R}^m$. Analogously to ``Bounding $I_3$'' in the proof of \cref{thm1},  we can apply   \cref{pro2} to obtain to get $I_3\lesssim E\sqrt{m}(\omega(\mathbb{B}_g^m)+t)$ with probability exceeding $1-2\exp(-t^2)$. Setting $t= \omega(\mathbb{B}_g^m)$ yields the bound on $I_3$\begin{align}
        \label{eq:bound_I3_un}
        I_3 \lesssim E\sqrt{m}\cdot\omega(\mathbb{B}_g^m)
    \end{align}
    with probability exceeding $1-\exp(-\omega^2(\mathbb{B}_g^m))$.

    {\bf Bounding $I_4$:} With the assumptions \cref{eq:qpecon1un}--\cref{eq:qpecon2un} and \cref{eq:thm2sample}, we can invoke \cref{coro:qpe_structured} with $(\mathcal{A},\mathcal{B},\mathcal{E})=(\mathcal{K}_{\bm{x}},\mathcal{K}_{\bm{v}},\{0\}\times \mathbb{B}_g^m)$. This yields that the event 
    \begin{align}
       I_4 \lesssim \delta\sqrt{m}\cdot \Big(\omega(\mathbb{B}_g^m)+\sqrt{\mathscr{H}(\mathcal{K}_{\bm{x}},\rho_1)+\mathscr{H}(\mathcal{K}_{\bm{v}},\rho_2)}\Big)\label{eq:bound_I4_un}
   \end{align}
    holds with probability exceeding $1-12\exp(-\Omega(\mathscr{H}(\mathcal{K}_{\bm{x}},\rho_1)+\mathscr{H}(\mathcal{K}_{\bm{v}},\rho_2)))$.

    Compared to our choice of $(\lambda_1,\lambda_2)$ in  \cref{eq:lam1choice}--\cref{eq:lam2choice}, we arrive at
    \begin{align}
    \lambda_1\geq 4(I_1+I_2),~  \lambda_2\geq 4(I_3+I_4) \label{eq:lam12ge}  
    \end{align}
    that hold with the promised probability.

  \subsubsection*{Step 2: Constraining Estimation Errors} Substituting \cref{eq:lam12ge} into 
  \cref{eq:last_ineq}  yields 
  \begin{align}
      \|\bm{\Phi\Delta_x}+\sqrt{m}\bm{\Delta_v}\|_2^2 \le \frac{\lambda_1}{2}f(\bm{\Delta_x})+\frac{\lambda_2}{2}g(\bm{\Delta_v})+\lambda_1\big(f(\bm{x^\star})-f(\bm{\hat{x}})\big)+ \lambda_2 \big(g(\bm{v^\star})-g(\bm{\hat{v}})\big). \label{eq:final_use}
  \end{align}
  Since $\|\bm{\Phi\Delta_x}+\sqrt{m}\bm{\Delta_v}\|^2_2\geq 0$, we obtain 
  \begin{equation}\label{3.49}
      \lambda_1\big(f(\bm{\hat{x}})-f(\bm{x^\star})\big)+\lambda_2\big(g(\bm{\hat{v}})-g(\bm{v^\star})\big)\leq \frac{\lambda_1}{2}f(\bm{\Delta_x})+\frac{\lambda_2}{2}g(\bm{\Delta_v}).
  \end{equation}
  For a linear subspace $\mathcal{X}$, we use $\mathcal{P}_{\mathcal{X}}(\cdot)$ to denote the projection onto $\mathcal{X}$ under $\ell_2$-norm. Given any $\bm{x^\star}\in \mathcal{K}_{\bm{x}}$, we can pick a pair of linear subspaces $(\mathcal{X},\overline{\mathcal{X}}):=(\mathcal{X}_{\bm{x^\star}},\overline{\mathcal{X}}_{\bm{x^\star}})$   as in   \cref{assump4} (they depend on $\bm{x^\star}$, while we will proceed with $(\mathcal{X},\overline{\mathcal{X}})$ to avoid cumbersome notation) such that $f(\cdot)$ is decomposable over $(\mathcal{X},\overline{\mathcal{X}}^\bot)$.  
  Then we can proceed as   
      \begin{align} 
          f(\bm{\hat{x}}) =f(\bm{x^\star}+\bm{\Delta_x}) &=f(\bm{x^\star}+\mathcal{P}_{\overline{\mathcal{X}}}\bm{\Delta_x}+\mathcal{P}_{\overline{\mathcal{X}}^\bot}\bm{\Delta_x}) \\&\label{eq:use_tri1}\geq f(\bm{x^\star}+\mathcal{P}_{\overline{\mathcal{X}}^\bot}\bm{\Delta_x})-f(\mathcal{P}_{\overline{\mathcal{X}}}\bm{\Delta_x})\\\label{eq:use_decom1}&=f(\bm{x^\star})+f(\mathcal{P}_{\overline{\mathcal{X}}^\bot}\bm{\Delta_x})-f(\mathcal{P}_{\overline{\mathcal{X}}}\bm{\Delta_x}),
      \end{align}
   where in \cref{eq:use_tri1} we use   triangle inequality and in \cref{eq:use_decom1} we use the decomposibility of $f(\cdot)$ as per \cref{3.39}. Thus, we arrive at
\begin{align}\label{eq:flower}
    f(\bm{\hat{x}})-f(\bm{x^\star})\geq f(\mathcal{P}_{\overline{\mathcal{X}}^\bot}\bm{\Delta_x})-f(\mathcal{P}_{\overline{\mathcal{X}}}\bm{\Delta_x}).
\end{align}  Regarding the corruption, there exists $(\mathcal{V},\overline{\mathcal{V}}):=(\mathcal{V}_{\bm{v^\star}},\overline{\mathcal{V}}_{\bm{v^\star}})$ for a given $\bm{v^\star}\in \mathcal{K}_{\bm{v}}$ as in   \cref{assump4} such that $g(\cdot)$ is decomposable over $(\mathcal{V},\overline{\mathcal{V}}^\bot)$, and we can similarly show  
\begin{align}\label{eq:glower}
g(\bm{\hat{v}})-g(\bm{v^\star})\geq g(\mathcal{P}_{\overline{\mathcal{V}}^\bot}\bm{\Delta_v})-g(\mathcal{P}_{\overline{\mathcal{V}}}\bm{\Delta_v}).    
\end{align}
  Substituting \cref{eq:flower}--\cref{eq:glower} into the left-hand side of \cref{3.49}, and then use triangle inequality $f(\bm{\Delta_x})\leq f(\mathcal{P}_{\overline{\mathcal{X}}^\bot}\bm{\Delta_x})+f(\mathcal{P}_{\overline{\mathcal{X}}}\bm{\Delta_x})$ and $g(\bm{\Delta_v})\leq g(\mathcal{P}_{\overline{\mathcal{V}}^\bot}\bm{\Delta_v})+g(\mathcal{P}_{\overline{\mathcal{V}}}\bm{\Delta_v})$ in the right-hand side, we obtain 
  \begin{equation}
      \label{3.522}\lambda_1f(\mathcal{P}_{\overline{\mathcal{X}}^\bot}\bm{\Delta_x})+\lambda_2g(\mathcal{P}_{\overline{\mathcal{V}}^\bot}\bm{\Delta_v})\leq 3\lambda_1f(\mathcal{P}_{\overline{\mathcal{X}}}\bm{\Delta_x})+3\lambda_2g(\mathcal{P}_{\overline{\mathcal{V}}}\bm{\Delta_v}).
  \end{equation} Based on this, we can derive the following
      \begin{align}  &\lambda_1f(\bm{\Delta_x})+\lambda_2g(\bm{\Delta_v})\\\label{eq:use_tri2}\leq&\lambda_1f(\mathcal{P}_{\overline{\mathcal{X}}}\bm{\Delta_x})+ \lambda_1f(\mathcal{P}_{\overline{\mathcal{X}}^\bot}\bm{\Delta_x})+\lambda_2g(\mathcal{P}_{\overline{\mathcal{V}}}\bm{\Delta_v})+ \lambda_2g(\mathcal{P}_{\overline{\mathcal{V}}^\bot}\bm{\Delta_v})\\ \label{eq:usebound}
      \le &  4\lambda_1f(\mathcal{P}_{\overline{\mathcal{X}}}\bm{\Delta_x})+4\lambda_2g(\mathcal{P}_{\overline{\mathcal{V}}}\bm{\Delta_v})\\\label{eq:use_compatible}\le& 4\lambda_1\alpha_{\bm{x}}\|\bm{\Delta_x}\|_2+4\lambda_2\alpha_{\bm{v}}\|\bm{\Delta_v}\|_2,
      \end{align}
  where \cref{eq:use_tri2} follows from triangle inequality,   in \cref{eq:usebound} we substitute \cref{3.522}, and moreover, \cref{eq:use_compatible} is due to $\alpha_f(\overline{\mathcal{X}})\leq \alpha_{\bm{x}} $ and $\alpha_g(\overline{\mathcal{V}})\leq \alpha_{\bm{v}}$ that hold uniformly for all $(\bm{x^\star},\bm{v^\star})$   (\cref{assump4}); in more detail, e.g., $f(\mathcal{P}_{\overline{\mathcal{X}}}\bm{\Delta_x})\le \alpha_f(\overline{\mathcal{X}})\|\mathcal{P}_{\overline{\mathcal{X}}}\bm{\Delta_x}\|_2\le \alpha_{\bm{x}}\|\bm{\Delta_x}\|_2$.
  Therefore, we   arrive at \begin{align}
      \label{eq:constrain_err}(\bm{\Delta_x},\bm{\Delta_v})\in \mathcal{C}(\lambda_1,\lambda_2),~\forall (\bm{x^\star},\bm{v^\star})\in \mathcal{K}_{\bm{x}}\times \mathcal{K}_{\bm{v}},
  \end{align}  
  (it is easy to check that each step is uniform for all pairs of $(\bm{x^\star},\bm{v^\star})$,) where the constraint set is given by \begin{align}
      \mathcal{C}(\lambda_1,\lambda_2)=\big\{(\bm{c},\bm{d})\in\mathbb{R}^n\times \mathbb{R}^m:\lambda _1f(\bm{c})+\lambda_2g(\bm{d})\leq 4\lambda_1\alpha_{\bm{x}}\|\bm{c}\|_2+4\lambda_2\alpha_{\bm{v}}\|\bm{d}\|_2\big\}.\label{eq:Clam12}
  \end{align}

\textbf{Uniform Lower Bound on $\|\bm{\Phi\Delta_x}+\sqrt{m}\bm{\Delta_v}\|_2^2$:} Equipped with \cref{eq:constrain_err}, we are now able to establish a uniform lower bound on $\|\bm{\Phi\Delta_x}+\sqrt{m}\bm{\Delta_v}\|_2^2$. Note that $\mathcal{C}(\lambda_1,\lambda_2)$ is a cone, and we   localize it as  $\mathcal{C}^*=\mathcal{C}(\lambda_1,\lambda_2)\cap \mathbb{S}^{n+m-1}$. Then we invoke   \cref{pro1} with $\mathcal{T}=\mathcal{C}^*$ and $t= \omega(\mathbb{B}_f^n)$ to obtain that the event\footnote{By \cref{pro1} we can bound the left-hand side of \cref{3.53} by $O(\gamma(\mathcal{C}^*)+\omega(\mathbb{B}_f^n))$. Note that $\bm{v}\in\mathcal{C}^*$ implies  $-\bm{v}\in \mathcal{C}^*$, hence let $\bm{g}\sim \mathcal{N}(0,\bm{I}_{n+m})$ we have $\omega(\mathcal{C}^*)=\mathbbm{E}\sup_{\bm{v}\in \mathcal{C}^*}(\bm{g}^\top\bm{v})\ge \mathbbm{E}\sup_{\bm{v}\in \mathcal{C}^*}\max\{\bm{g}^\top \bm{v},\bm{g}^\top(-\bm{v})\}=\mathbbm{E}\sup_{\bm{v}\in \mathcal{C}^*}|\bm{g}^\top \bm{v}|=\gamma(\mathcal{C}^*)$. Thus, we arrive at the bound given in the right-hand side of \cref{3.53}.}
\begin{equation} \label{3.53}
\sup_{(\bm{c},\bm{d})\in\mathcal{C}^*}\Big|\|\bm{\Phi c}+\sqrt{m}\bm{d}\|_2-\sqrt{m}\Big|\leq C _1\big(\omega(\mathcal{C}^*)+\omega(\mathbb{B}_f^n)\big)
\end{equation} 
holds with  probability exceeding $1-\exp(-\omega^2(\mathbb{B}_f^n))$. Note that  \cref{pro9} provides a bound on $\omega(\mathcal{C}^*)$, which indicates that \cref{eq:thm2sample} implies $m\gtrsim \omega^2(\mathcal{C}^*)+\omega^2(\mathbb{B}_f^n)$. Thus, we can assume that the right-hand side of \cref{3.53} is bounded by $\frac{1}{2}\sqrt{m}$ and achieve 
\begin{align}
    \inf_{(\bm{c},\bm{d})\in\mathcal{C}^*}\|\bm{\Phi c}+\sqrt{m}\bm{d}\|_2 &\ge \sqrt{m} - \sup_{(\bm{c},\bm{d})\in \mathcal{C}^*}\big|\|\bm{\Phi c}+\sqrt{m}\bm{d}\|_2-\sqrt{m}\big|\\
    &\ge \sqrt{m}-\frac{\sqrt{m}}{2}= \frac{\sqrt{m}}{2}.
\end{align}
Combining with \cref{eq:constrain_err}, we obtain 
\begin{align}
    \|\bm{\Phi\Delta_x}+\sqrt{m}\bm{\Delta_v}\|_2^2 & \ge  \big(\|\bm{\Delta_x}\|_2^2+\|\bm{\Delta_v}\|_2^2\big) \inf_{(\bm{c},\bm{d})\in\mathcal{C}^*}\|\bm{\Phi c}+\sqrt{m}\bm{d}\|_2 ^2 \\
    &\ge \frac{m}{4}\big(\|\bm{\Delta_x}\|_2^2+\|\bm{\Delta_v}\|_2^2\big).\label{eq:un_lower} 
\end{align}

\subsubsection*{Step 3: Combining Everything} 
We derive the desired bound by bounding both sides of \cref{eq:final_use}: By \cref{eq:un_lower}, the left-hand side is uniformly lower bounded by $\frac{m}{4}(\|\bm{\Delta_x}\|_2^2+\|\bm{\Delta_v}\|_2^2)$; Using triangle inequality and \cref{eq:use_compatible} in \cref{eq:386}, along with simple algebra in \cref{eq:387}, the right-hand side  can be uniformly bounded by 
    \begin{align}\label{upperthm2}
        &\frac{\lambda_1}{2} f(\bm{\Delta_x})+\frac{\lambda_2}{2}g(\bm{\Delta_v})+\lambda_1\big(f(\bm{x^\star})-f(\bm{\hat{x}})\big)+\lambda_2\big(g(\bm{v^\star})-g(\bm{\hat{v}})\big)
        \\& \label{eq:386}\quad \leq \frac{3\lambda_1}{2}f(\bm{\Delta_x})+ \frac{3\lambda_2}{2}g(\bm{\Delta_v})\leq 6\big(\lambda_1\alpha_{\bm{x}}\|\bm{\Delta_x}\|_2+\lambda_2\alpha_{\bm{v}}\|\bm{\Delta_v}\|_2\big)\\&\label{eq:387} \quad\le 6(\lambda_1\alpha_{\bm{x}}+\lambda_2\alpha_{\bm{v}})\big(\|\bm{\Delta_x}\|_2^2+\|\bm{\Delta_v}\|_2^2\big)^{1/2}
    \end{align}
    Rearranging immediately yields the claim that holds with the promised probability as per \cref{eq:thm2_prob}.  
\end{proof}

Substituting \cref{eq:lam1choice}--\cref{eq:lam2choice} into \cref{thm2bound} yields the following more explicit uniform bound on $(\|\bm{\Delta_x}\|_2^2+\|\bm{\Delta_v}\|_2^2)^{1/2}$
\begin{align}\label{eq:explicit_bound}
     O\left( \frac{(E+\delta)\big(\alpha _{\bm{x}}\cdot\omega(\mathbb{B}^n_f)+\alpha_{\bm{v}}\cdot\omega(\mathbb{B}^m_g)\big)+\delta(\alpha_{\bm{x}}+\alpha_{\bm{v}})\big(\mathscr{H}(\mathcal{K}_{\bm{x}},\rho_1)+\mathscr{H}(\mathcal{K}_{\bm{v}},\rho_2)\big)^{1/2}}{\sqrt{m}}\right). 
\end{align}
To see the implication of   \cref{thm2} on structured priors (\cref{assump3}), we substitute  \cref{3.77} into \cref{eq:explicit_bound} and omit some logarithmic factors on $(\rho_1,\rho_2)$, then we obtain the uniform bound
\begin{equation}
\label{unconstrinbound}\tilde{O}\left(\frac{(E+\delta)\big(\alpha_{\bm{x}}\cdot\omega(\mathbb{B}_f^n)+\alpha_{\bm{v}}\cdot\omega(\mathbb{B}_g^m)\big)+\delta(\alpha_{\bm{x}}+\alpha_{\bm{v}})\big(\omega(\mathcal{K}_{\bm{x}})+\omega(\mathcal{K}_{\bm{v}})\big)}{\sqrt{m}}\right).
\end{equation}
\begin{rem}[The Cost of Uniformity: Unconstrained Lasso with Structured Priors]\label{rem4} 
The non-uniform error rate for \cref{3.5} in \cite[Thm. 2]{sun2022quantized} also reads as $O\big(\frac{\lambda_1\alpha_{\bm{x}}+\lambda_2\alpha_{\bm{v}}}{m}\big)$ (when adjusted to our notation), with the  parameters $\lambda_1\asymp (E+\delta)\sqrt{m}\cdot \omega(\mathbb{B}_f^n)$ and $\lambda_2\asymp (E+\delta)\sqrt{m}\cdot \omega(\mathbb{B}_g^m)$ to guarantee that 
\begin{equation}\label{unconstrainkey}
        \lambda_1\gtrsim \sup_{\bm{c}\in \mathbb{B}_f^n}\langle\bm{\epsilon}+\bm{\xi}_{\bm{x^\star},\bm{v^\star}},\bm{\Phi c}\rangle,~\lambda_2\gtrsim \sup_{\bm{d}\in \mathbb{B}_g^m}\langle\bm{\epsilon}+\bm{\xi}_{\bm{x^\star},\bm{v^\star}},\sqrt{m}\bm{d}\rangle
    \end{equation}
    holds for   a fixed $(\bm{x^\star},\bm{v^\star})$; see Remark 5 therein. Thus, their non-uniform error rate translates into \begin{align}\label{eq:sun_non_un}
        O\left(\frac{(E+\delta)\big(\alpha_{\bm{x}}\cdot\omega(\mathbb{B}_f^n)+\alpha_{\bm{v}}\cdot\omega(\mathbb{B}_g^m)\big)}{\sqrt{m}}\right).
    \end{align} 
    However, to achieve uniformity, we must use sufficiently large $\lambda_1$ and $\lambda_2$     to ensure that \cref{unconstrainkey} holds for all $(\bm{x^\star},\bm{v^\star})$, i.e., 
    \begin{align}
          \lambda_1\gtrsim \sup_{\bm{a}\in \mathcal{K}_{\bm{x}}}\sup_{\bm{b}\in \mathcal{K}_{\bm{v}}}\sup_{\bm{c}\in \mathbb{B}_f^n}\langle\bm{\epsilon}+\bm{\xi}_{\bm{a},\bm{b}},\bm{\Phi c}\rangle,~\lambda_2\gtrsim \sup_{\bm{a}\in \mathcal{K}_{\bm{x}}}\sup_{\bm{b}\in \mathcal{K}_{\bm{v}}}\sup_{\bm{d}\in \mathbb{B}_g^m}\langle\bm{\epsilon}+\bm{\xi}_{\bm{a},\bm{b}},\sqrt{m}\bm{d}\rangle.
    \end{align}
    For this purpose, the second term of $\lambda_1$  in \cref{eq:lam1choice} (or $\lambda_2$ in \cref{eq:lam2choice}) is additional compared to the $\lambda_1$ and $\lambda_2$ in \cite{sun2022quantized}. This leads to the additional term   \begin{align}
        \tilde{O}\left(\frac{\delta(\alpha_{\bm{x}}+\alpha_{\bm{v}})(\omega(\mathcal{K}_{\bm{x}})+\omega(\mathcal{K}_{\bm{v}}))}{\sqrt{m}}\right)\label{eq:un_additional}
    \end{align} in our uniform error rate \cref{unconstrinbound}, which we further discuss as follows: 
    \begin{itemize}
     [leftmargin=5ex,topsep=0.25ex]
        \item Under the regular scaling of $E\asymp \delta$, the additional term \cref{eq:un_additional} is typically dominant, and as a result our uniform rate \cref{unconstrinbound} often exhibits worse dependence on structured parameters than its non-uniform counterpart \cref{eq:sun_non_un}. For instance, in recovery of $s$-sparse $\bm{x^\star}$ we have $\mathcal{K}_{\bm{x}}=\Sigma^n_s\cap \mathbb{B}^n_2$ and $f(\bm{x})=\|\bm{x}\|_1$, and thus  $\omega(\mathcal{K}_{\bm{x}})\asymp \sqrt{s\log\frac{en}{s}}$ and $\omega(\mathbb{B}_f^n)\asymp \sqrt{\log n}$, then the term $\frac{\alpha_{\bm{x}}\cdot \omega(\mathcal{K}_{\bm{x}})}{\sqrt{m}}$ in \cref{eq:un_additional} already loses a factor of $\sqrt{s}$ compared to $\frac{\alpha_{\bm{x}}\cdot \omega(\mathbb{B}_f^n)}{\sqrt{m}}$ in \cref{eq:sun_non_un}. The readers can   clearly see such degradation by comparing   \cite[Coros 3--4]{sun2022quantized} and our \cref{coro3}--\cref{coro4} below.  
        
        \item Nevertheless, to the best of our knowledge,  there exists no   uniform result for unconstrained Lasso   in  quantized compressed/corrupted sensing (see   \cref{rem3} for a review), thereby pointing to the open question on the possibility of improvement.
        \item  Moreover, we observe that the additional term \cref{eq:un_additional} presents a multiplicative factor $\delta$, reflecting that the gap between uniform recovery and non-uniform recovery closes when $\delta$ is extremely small or in an unquantized setting where $\delta=0$. 
    \end{itemize}
\end{rem}

We present some concrete outcomes of   \cref{thm2}, with proofs   deferred to \cref{appendixc}.

\begin{corollary}[Sparse Signal and Sparse Corruption]
    \label{coro3}
 We consider the same settings as   in  \cref{coro1} (i.e., $\bm{x^\star}\in\mathcal{K}_{\bm{x}}=\Sigma^n_s\cap \mathbb{B}_2^n$ with $f(\bm{x})=\|\bm{x}\|_1$, $\bm{v^\star}\in\mathcal{K}_{\bm{v}}=\Sigma^m_k\cap \mathbb{B}_2^m$ with $g(\bm{v})=\|\bm{v}\|_1$), but use \cref{3.5}  as recovery program. We set \begin{equation}
        \begin{aligned}\nonumber
            &\lambda_1=C_1 \delta \sqrt{ms\log\Big(\frac{nm^{3/2}}{s^{5/2}\delta}\Big)+mk\log\Big(\frac{m^{5/2}}{k^{5/2}\delta}\Big)}+C_1(E+\delta)\sqrt{m\log n} ,\\
            &\lambda_2=C_2  \delta \sqrt{ms\log\Big(\frac{nm^{3/2}}{s^{5/2}\delta}\Big)+mk\log\Big(\frac{m^{5/2}}{k^{5/2}\delta}\Big)}+C_2(E+\delta)\sqrt{m\log m} 
        \end{aligned}
    \end{equation}   
    with sufficiently large $C_1,C_2$. If $m\ge C_3(s+k)\log(mn)+C_3 s\log(\frac{nm^{3/2}}{s^{5/2}\delta})+C_3k\log(\frac{m^{5/2}}{k^{5/2}\delta})$ for large enough $C_3$, then with probability exceeding $1-n^{-c_4}-m^{-c_5}$ on a single draw of $(\bm{\Phi},\bm{\epsilon},\bm{\tau})$,  the following uniform error bound holds true for all $(\bm{x^\star},\bm{v^\star})\in \mathcal{K}_{\bm{x}}\times \mathcal{K}_{\bm{v}}$: 
    \begin{equation}\nonumber
        \sqrt{\|\bm{\Delta_{x}} \|_2^2+\|\bm{\Delta_{v}}\|_2^2}\lesssim \frac{(E+\delta)\sqrt{s\log n+k\log m}+\delta \sqrt{(s+k)\big(s\log(\frac{nm^{3/2}}{s^{5/2}\delta})+k\log(\frac{m^{5/2}}{k^{5/2}\delta})\big)}}{\sqrt{m}},
    \end{equation}
    where   $\bm{\Delta_x}=\bm{\hat{x}}-\bm{x^\star},\bm{\Delta_v}=\bm{\hat{v}}-\bm{v^\star}$, $(\bm{\hat{x}},\bm{\hat{v}})$ is the solution to \cref{3.5}.
\end{corollary}
\begin{corollary}[Low-Rank Signal and Sparse Corruption]\label{coro4}
 We consider the same settings as   in  \cref{coro2} (i.e., $\bm{x^\star}\in \mathcal{K}_{\bm{x}}=M_r^{p,q}\cap\mathbb{B}_{\rm F}^{p,q}$ with $f(\bm{x})=\|\bm{x}\|_{\rm nu}$, $\bm{x^\star}$ is vectorized as a $(pq)$-dimensional vector when substituted into \cref{3.1}, $\bm{v^\star}\in\mathcal{K}_{\bm{v}}=\Sigma^m_k\cap \mathbb{B}_2^m$ with $g(\bm{v})=\|\bm{v}\|_1$), but use \cref{3.5}   as recovery program. We set \begin{equation}
        \begin{aligned}\nonumber
            &\lambda_1=C_1\delta\cdot \sqrt{mr(p+q)\log\Big(\frac{m^{3/2}}{\delta(r(p+q))^{3/2}}\Big)+mk\log\Big(\frac{m^{5/2}}{k^{5/2}\delta}\Big)}+C_1(E+\delta)\sqrt{m(p+q)},
            \\&\lambda_2=C_2\delta \cdot \sqrt{mr(p+q)\log\Big(\frac{m^{3/2}}{\delta(r(p+q))^{3/2}}\Big)+mk\log\Big(\frac{m^{5/2}}{k^{5/2}\delta}\Big)}+C_2(E+\delta)\sqrt{m\log m}
        \end{aligned}
    \end{equation}
    with sufficiently large $C_1,C_2$. If $m\ge C_3(r+k)(p+q+\log m)+C_3 r (p+q)\log(\frac{m^{3/2}}{\delta(r(p+q))^{3/2}})+C_3k\log \frac{m^{5/2}}{k^{5/2}\delta}$ for large enough $C_3$, then with probability exceeding $1- C_4\exp(-c_5(p+q))-C_6m^{-c_7}$ on a single draw of $(\bm{\Phi},\bm{\epsilon},\bm{\tau})$,  the following uniform error bound holds true for all $(\bm{x^\star},\bm{v^\star})\in \mathcal{K}_{\bm{x}}\times \mathcal{K}_{\bm{v}}$: 
    \begin{align*}
        \sqrt{\|\bm{\Delta_{x}} \|_\mathrm{F}^2+\|\bm{\Delta_{v}}\|_2^2}&\lesssim \frac{(E+\delta)\sqrt{r(p+q)+k\log m}}{\sqrt{m}}\\&\quad\quad+\frac{\delta\sqrt{(r+k)\big(r(p+q)\log(\frac{m^{3/2}}{\delta(r(p+q))^{3/2}})+k\log(\frac{m^{5/2}}{\delta k^{5/2}})\big)}}{\sqrt{m}},
    \end{align*}
    where $\bm{\Delta_x}=\bm{\hat{x}}-\bm{x^\star},\bm{\Delta_v}=\bm{\hat{v}}-\bm{v^\star}$, $(\bm{\hat{x}},\bm{\hat{v}})$ is the solution to \cref{3.5}. 
\end{corollary}

\subsection{Generative Priors}
 To handle the case where traditional structured priors fail to precisely characterize the underlying signal,  it was recently proposed to use a generative prior for compressed sensing, i.e., assuming that the desired signal lies in the range of a generative model \cite{bora2017compressed}. 
This new perspective for compressed sensing  has led to numerical success and attracted
much research interest.
The goal of this subsection is to establish uniform recovery guarantee for quantized corrupted sensing using generative priors. Following the long list of prior works on this field (e.g., \cite{bora2017compressed,liu2020generalized,chen2023unified,berk2021deep}), we formulate the generative priors on signal and corruption as follows. 
\begin{assumption}[Generative Priors]
    \label{assump5}
    For some $r, r' > 0$, let  $G\,:\, \mathbb{B}_2^k(r) \to \mathbb{R}^n$, $H\,:\, \mathbb{B}_2^{k'}(r') \to \mathbb{R}^m$ be some   generative models. We assume that $G(\cdot)$ is $L$-Lipschitz continuous, $H(\cdot)$ is $L'$-Lipschitz continuous:
    \begin{align}
        &\|G(\bm{a}_1)-G(\bm{a}_2)\|_2\leq L \|\bm{a}_1-\bm{a}_2\|_2,~\forall \bm{a}_1,\bm{a}_2\in \mathbb{R}^k,\\
        &\|H(\bm{b}_1)-H(\bm{b}_2)\|_2\leq L'\|\bm{b}_1-\bm{b}_2\|_2,~\forall \bm{b}_1,\bm{b}_2\in \mathbb{R}^{k'}, 
    \end{align}
  and assume that the signal and corruption lie in the range of the two generative models:  
  \begin{align}
      \bm{x^\star}\in \mathcal{K}_{\bm{x}}:=G\big(\mathbb{B}_2^k(r)\big), \bm{v^\star}\in\mathcal{K}_{\bm{v}}:= H\big(\mathbb{B}_2^{k'}(r')\big). \label{eq:gene_prior}
  \end{align}
\end{assumption}

Recall that $\bm{\dot{y}}$ are the quantized measurements as per \cref{3.2}. We naturally extend the scope of constrained Lasso \cref{3.4} by substituting the norm constraints with the generative priors on $(\bm{x^\star},\bm{v^\star})$: \begin{equation}
    \label{geneprogram}
    (\bm{\hat{x}},\bm{\hat{v}})=\mathrm{arg}\min_{\substack{\bm{x}\in \mathbb{R}^n\\\bm{v}\in \mathbb{R}^m}}\|\bm{\dot{y}}-\bm{\Phi x}-\sqrt{m}\bm{v}\|_2,~~\mathrm{s.t.~}\bm{x}\in G\big(\mathbb{B}_2^k(r)\big),~\bm{v}\in H\big(\mathbb{B}_2^{k'}(r')\big).
\end{equation}
We note that it is in general hard to exactly optimize \cref{geneprogram}   due to the highly non-convex constraint,  while fortunately there have been some practical approaches to {\it approximately} solve this program~\cite{shah2018solving, raj2019gan, liu2022generative}. Note that all prior recovery methods in this area exhibit optimization issue of this type but have proven effective in practice.

We present a uniform recovery guarantee of \cref{geneprogram}.
\begin{theorem}[Quantized Corrupted Sensing with Generative Priors]
    \label{thm3}
    Under   \cref{assump1} and  \cref{assump5}, we let $\mu\in (0,1)$ be some given recovery accuracy and suppose that $m=O(n)$. If  for some sufficiently large $C_1$ we have
    \begin{equation}\label{genesample}
        m\ge  \frac{C_1E^2}{\mu^2}\Big(k\log\frac{Lr}{\mu}+k'\log\frac{L'r'}{\mu}\Big) + C_1\Big(1+\frac{\delta^2}{\mu^2}\Big)\Big(k\log \Big(\frac{Lrn^{3/2}}{\mu\delta k^{3/2}}\Big)+k'\log\Big(\frac{L'r'm^{3/2}}{\mu\delta (k')^{3/2}}\Big)\Big),
    \end{equation}
    then with probability exceeding  $1-C_2\exp(-\Omega(k\log(Lr)+k'\log(L'r')))$ on a single draw of $(\bm{\Phi},\bm{\epsilon},\bm{\tau})$,   the uniform error bound $\sqrt{\|\bm{\Delta_x} \|_2^2+\|\bm{\Delta_v}\|_2^2}\leq \mu$  holds  for all $(\bm{x^\star},\bm{v^\star})\in \mathcal{K}_{\bm{x}}\times \mathcal{K}_{\bm{v}}$, where $\bm{\Delta_x}=\bm{\hat{x}}-\bm{x^\star}$ and $\bm{\Delta_v}=\bm{\hat{v}}-\bm{v^\star}$, with  $(\bm{\hat{x}},\bm{\hat{v}})$ being the solution to \cref{geneprogram}. 
\end{theorem}
\begin{proof}[A Sketch of the Proof]
The proof of  \cref{thm3} is analogous to that of \cref{thm1}, with the major differences lying in constraining the range of the estimation error $(\bm{\Delta_x},\bm{\Delta_v})$, which also appeals to separate treatments in   estimating the Gaussian width and Kolmogorov entropy (see   \cref{proadd1}). The full proof can be found in   \cref{appendixc}.
\end{proof}
\begin{rem}
    [Uniform Error Bound for \cref{thm3}]\label{rem:gene_rate} Provided that $m\gtrsim k\log (\frac{Lrn^{3/2}}{\mu\delta k^{3/2}})+k'\log(\frac{L'r'm^{3/2}}{\mu\delta (k')^{3/2}})$ for large enough implied constant, the sample complexity \cref{genesample} sufficient for achieving a uniform $\ell_2$-error of $\mu$ implies a uniform error bound 
    \begin{align}
        \nonumber\sqrt{\|\bm{\Delta_x}\|_2^2+\|\bm{\Delta_v}\|_2^2}&\lesssim E\sqrt{\frac{k\log(\frac{Lr\sqrt{m}}{E\sqrt{k}})+k'\log(\frac{L'r'\sqrt{m}}{E\sqrt{k'}})}{m}}  \\&\quad\quad+ \delta\sqrt{\frac{k\log(\frac{Lrn^2}{\delta^2k^2})+k'\log(\frac{L'r'm^2}{\delta^2(k')^2})}{m}}, 
    \end{align}
    which can be verified by some algebra. 
\end{rem}
\begin{rem}[Technical Comparison with \cite{chen2023unified}]\label{rem:gene_related}  
Restricted to compressed sensing with generative prior, most recovery guarantees for non-linear models are non-uniform. As a follow-up of \cite{genzel2022unified}, Chen et al. \cite{chen2023unified} built a unified framework for proving uniform recovery guarantee in non-linear compressed sensing with generative prior. Specifically, they handled potential discontinuity of $f_i(\cdot)$ by constructing Lipschitz approximation 
    as in \cite{genzel2022unified}, but used \cite[Thm. 2]{chen2023unified} (rather than \cite[Thm. 8]{genzel2022unified}) to bound the product processes.  
    Their key observation is that the replacement of concentration inequality yields tighter bound for the generative case \cite[Remark 8]{chen2023unified}, thus they proved a uniform decaying rate of $O(m^{1/2})$ even when $f_i(\cdot)$ contains some discontinuity. (In contrast, \cite{genzel2022unified} only achieves a uniform decaying rate of $O(m^{-1/4})$ under discontinuous $f_i(\cdot)$, see \cref{rem:insensible}.) However, as with \cref{rem:insensible}, it is unclear whether their approach applies to analyzing corrupted sensing, since the corruption leads to random process out of the scope of \cite[Thm. 2]{chen2023unified}.\footnote{As we reviewed in \cref{sec:introduction}, under generative prior, the only existing result that accommodates a generative corruption $\bm{v^\star}$ was presented in  \cite{berk2020deep,berk2021deep} and restricted to the linear case.} Moreover, while the present paper provides unified analysis of two priors, it is unclear whether the techniques in \cite{chen2023unified} can be adapted to structured prior.


\end{rem}
\section{Experimental Results}\label{sec4:expresult}
In this section we provide experimental results to corroborate and demonstrate our uniform recovery guarantees. Due to the theoretical nature of our work, an extensive set of experiments will not pursued.

\subsection{Structured Priors}
First, we consider using structured priors for corruption sensing, for which our main theoretical results are   \cref{thm1} and \cref{thm2}. The aim of our first set of experiments is to show the recovery performance of   (un)constrained Lasso  under two settings, namely sparse signal recovery from sparse corruption (\cref{coro1} and \cref{coro3}) and low-rank matrix recovery from sparse corruption (\cref{coro2} and \cref{coro4}). All simulations in this subsection are performed using MATLAB R2018b on a desktop with a 3.70 GHz Intel Core i7-8700M CPU and 32 GB RAM.

 We use a realization of the ensemble  $(\bm{\Phi},\bm{\tau},\bm{\epsilon})$ to recover a fixed $(\bm{x}^\star,\bm{v}^\star)$ for simulating non-uniform recovery. In contrast, with a single realization of the sensing ensemble, the error rates in our theorems holds uniformly for {\it all} $(\bm{x}^\star,\bm{v}^\star)\in \mathcal{K}_{\bm{x}}\times \mathcal{K}_{\bm{v}}$,  or equivalently interpreted, they are upper bounds on the following quantity:  
 \begin{align}
     \label{eq:uniform_error}
     \sup_{\bm{x^\star}\in \mathcal{K}_{\bm{x}}}\sup_{\bm{v^\star}\in \mathcal{K}_{\bm{v}}}\big(\|\bm{\hat{x}}-\bm{x^\star}\|_2^2+\|\bm{\hat{v}}-\bm{v^\star}\|_2^2\big)^{1/2}.
 \end{align}
 Nonetheless, under a fixed $(\bm{\Phi},\bm{\tau},\bm{\epsilon})$, it is in general impossible to track \cref{eq:uniform_error} since $\mathcal{K}_{\bm{x}}\times \mathcal{K}_{\bm{v}}$ is typically infinite set (e.g., $\mathcal{K}_{\bm{x}}=\Sigma^n_s\cap \mathbb{B}_2^n$). To   provide some clues to demonstrate our theories, we instead utilize a fixed $(\bm{\Phi},\bm{\tau},\bm{\epsilon})$ to recover multiple signal-and-corruption pairs in a testing set $\mathcal{X}_{test}$, and we will track the maximum recovery error 
\begin{equation}\label{track}
    \sup_{(\bm{x^\star},\bm{v^\star})\in \mathcal{X}_{test}}\big(\|\bm{\hat{x}}-\bm{x^\star}\|_2^2+\|\bm{\hat{v}}-\bm{v^\star}\|_2^2\big)^{1/2}
\end{equation}
as an approximation of \cref{eq:uniform_error}. We adopt the following general principles in our simulations: \begin{itemize}
[leftmargin=5ex,topsep=0.25ex]
    \item \textbf{Construction of $\mathcal{X}_{test}$:} Given the cardinality $|\mathcal{X}_{test}|$, the $(\bm{x}^\star,\bm{v}^\star)$ in $\mathcal{X}_{test}$ are independently, randomly created by the   construction of sparse vector and low-rank matrix below. 

    \item \textbf{Calculation of \cref{track}:} In a single trial that simulates the uniform recovery over some $\mathcal{X}_{test}$, we calculate \cref{track} under a single draw of $(\bm{\Phi},\bm{\tau},\bm{\epsilon})$. We will report \cref{track} as its mean value in $10$ independent trials.

    \item   
    \textbf{Construction of a Vector in $\Sigma^N_s$:} We let the support of the vector uniformly distributed over all $\binom{N}{s}$ possibilities, and then draw the non-zero entries from $\mathcal{N}(0,1)$. 

    \item \textbf{Construction of a Matrix in $M^{p,q}_r$:} We use $\bm{U}\bm{V}^\top$ where  
where $\bm{U}\in \mathbb{R}^{p\times r}$ and $\bm{V}\in \mathbb{R}^{q\times r}$ are 
independent random matrices with orthonormal columns, generated by the Matlab code ``\texttt{[U,S,U1]=svd(randn(p,r),r);~[V,S,V1]=svd(randn(q,r),r)}''.  

    \item \textbf{Tuning Parameters:} In unconstrained Lasso we provide the best possible constraint for each pair  $(\bm{x}^\star,\bm{v}^\star)\in \mathcal{X}_{test}$, namely $f(\bm{x})\le f(\bm{x}^\star)$ and $g(\bm{v})\le g(\bm{v^\star})$. In unconstrained Lasso, we properly choose a large enough $(\lambda_1,\lambda_2)$ and then use it for all   $(\bm{x}^\star,\bm{v}^\star)\in \mathcal{K}_{test}$. 
\end{itemize}

\subsubsection{Non-Uniformity v.s. Uniformity} We first compare non-uniform recovery and uniform recovery (over $\mathcal{X}_{test}$) to {\it demonstrate the major theoretical achievements of this work}.

\textbf{Constrained Lasso:} we simulate \cref{coro1} with $(n,s,k,\delta,E)=(256,2,2,0.1,0)$. 
 We construct four test sets with different cardinalities: $|\mathcal{X}_{test}|=1,10,100,300$. Then we test these four cases under a properly chosen range of the measurement number $m$. 
 Note that $|\mathcal{X}_{test}|=1$ reduces to non-uniform recovery, while $|\mathcal{X}_{test}|=300$ represents the highest level of uniformity {\it that we simulate}. 
 Over these four test sets, we report \cref{track} as its mean value in 10 independent trials under properly chosen measurement number, and then plot the log-log curves in \cref{fig:constra_demo}(left). 
 We note the following observations from \cref{fig:constra_demo}(left) that are consistent with our \cref{coro1}: 
 \begin{itemize}
 [leftmargin=5ex,topsep=0.25ex]
     \item All curves decay  in a rate of $O(m^{-1/2})$;
     \item To achieve the same recovery error over a ``larger'' $\mathcal{X}_{test}$ (which corresponds to a higher level of uniformity) requires more measurements. 
 \end{itemize}
   Since \cref{eq:uniform_error} is always larger than \cref{track}, the log-log curve of   \cref{eq:uniform_error}   will further shift to the right compared to the green curve (corresponding to $|\mathcal{X}_{test}|=300$) in \cref{fig:constra_demo}(left), whereas regarding this curve corresponding to the ``actual'' uniform recovery (that we cannot experimentally track), our \cref{coro1}   guarantees the following: (i) Fixing other parameters, this curve still decays in a rate of $O(m^{-1/2})$; (ii) The  measurement number   needed to achieve the same recovery error as the curves in \cref{fig:constra_demo}(left) is still of the   scaling law  $\tilde{O}(s+k)$  (logarithmic factors omitted).

\textbf{Unconstrained Lasso:} We provide   similar experiment results for unconstrained Lasso to demonstrate our theory. We simulate \cref{coro3} with $(n,s,k,\delta,E)=(256,2,2,0.1,0)$ and track the quantity \cref{track} with $|\mathcal{X}_{test}|=1,10,100,300$. The results are displayed in \cref{fig:constra_demo}(right), and we note the similar observations that  all curves are decaying roughly in a rate of $O(m^{-1/2})$, and that  larger $\mathcal{X}_{test}$ requires more measurements to  ensure that \cref{track} is smaller than some desired accuracy.

Naturally, the log-log curve of \cref{eq:uniform_error} will further shift to the right compared to the green curve (under $|\mathcal{X}_{test}|=300$) in \cref{fig:constra_demo}(right), while regarding this curve corresponding to the ``actual'' uniform recovery (that we cannot experimentally track), our \cref{coro3} promises the following: (i) This curve still decays i a rate of $O(m^{-1/2})$ when other parameters are fixed; (ii) The measurement number needed to achieve the same recovery error as the curves in \cref{fig:constra_demo}(right) is of scaling $\tilde{O}(s^2+k^2)$ (logarithmic factors omitted, and note that we do not know whether this is sharp).

\begin{figure*}[!ht]\label{fig:constra_demo}
\scriptsize\setlength{\tabcolsep}{0.3pt}
\begin{center}
\begin{tabular}{cc}
\includegraphics[width=0.45\textwidth]{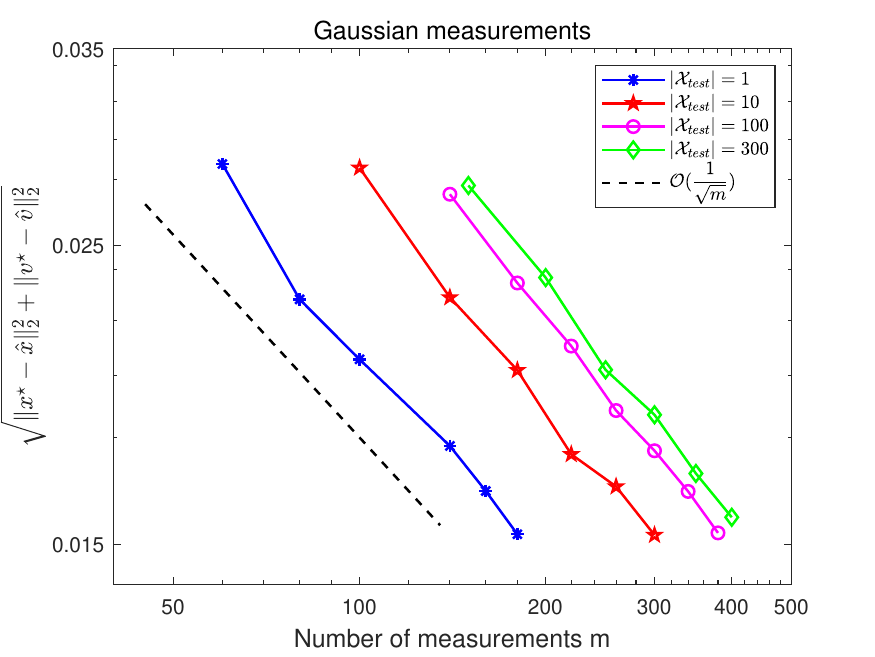}&
\includegraphics[width=0.45\textwidth]{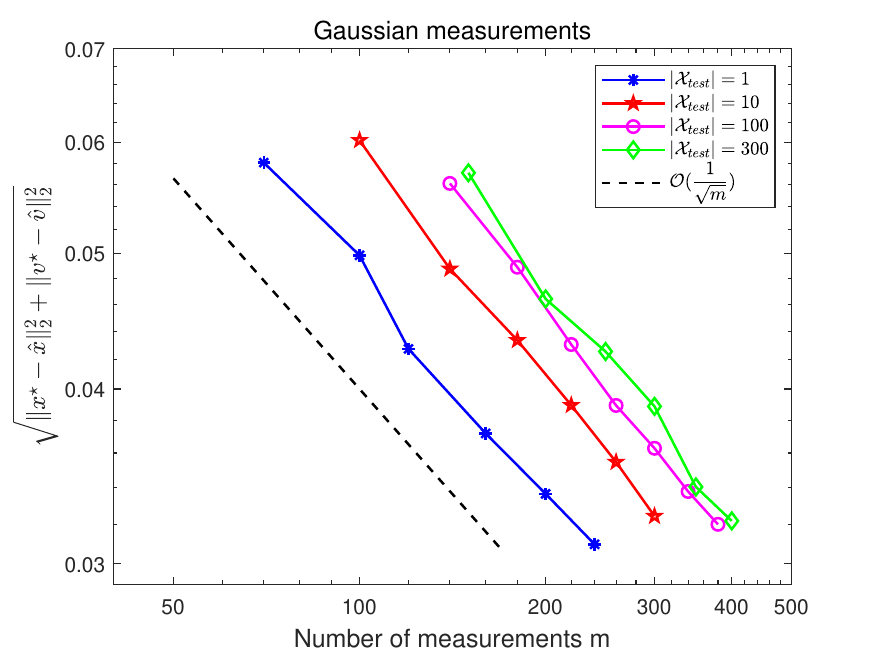}
\end{tabular}
\caption{From Non-Uniformity to Uniformity: (Left) Sparse Recovery  via Constrained Lasso (\cref{coro1}); (Right) Sparse Recovery via Unconstrained Lasso (\cref{coro3})} 
  \end{center}
\end{figure*}

\begin{figure*}[!ht]
\scriptsize\setlength{\tabcolsep}{0.3pt}
\begin{center}
\begin{tabular}{ccc}
\includegraphics[width=0.33\textwidth]{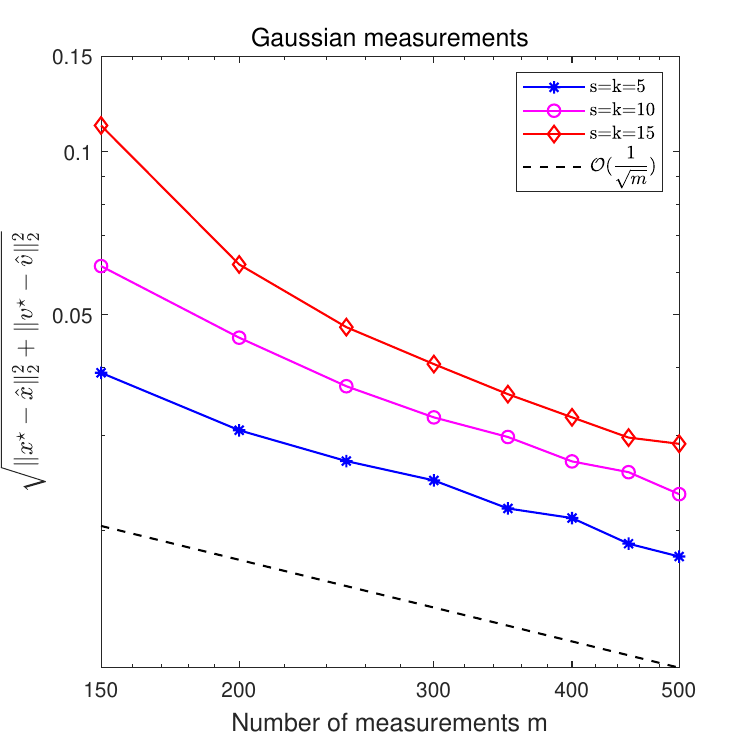}&
\includegraphics[width=0.33\textwidth]{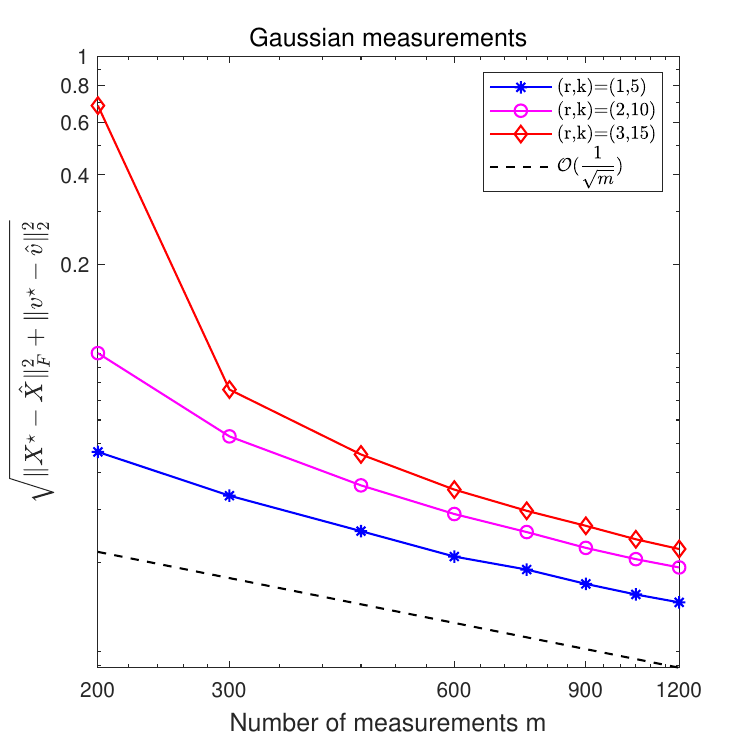}&
\includegraphics[width=0.33\textwidth]{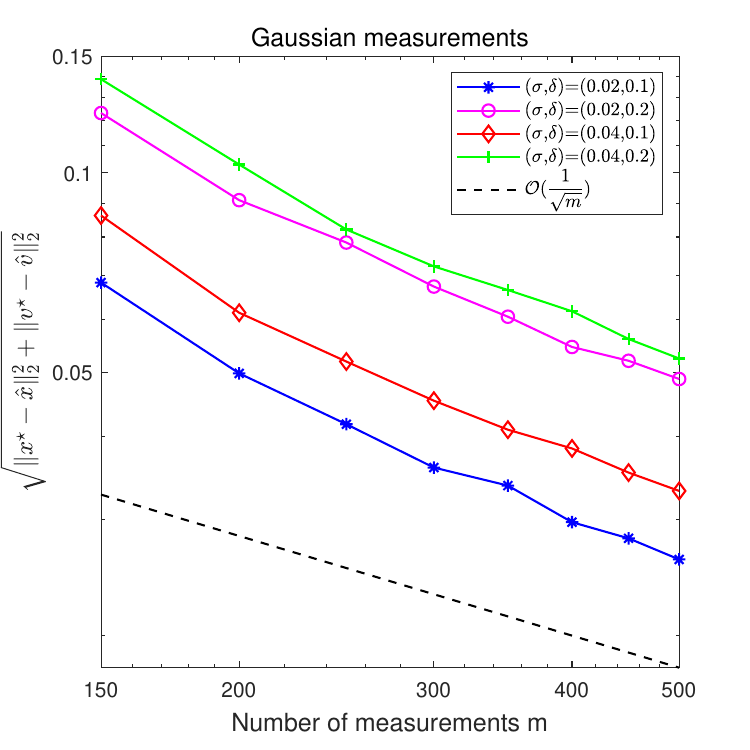}\\
 (a) Sparse signal recovery 
 & (b) Low-rank matrix recovery 
 & (c) Robustness to noise\\
 from sparse corruption 
 & from sparse corruption 
 & \\
\includegraphics[width=0.33\textwidth]{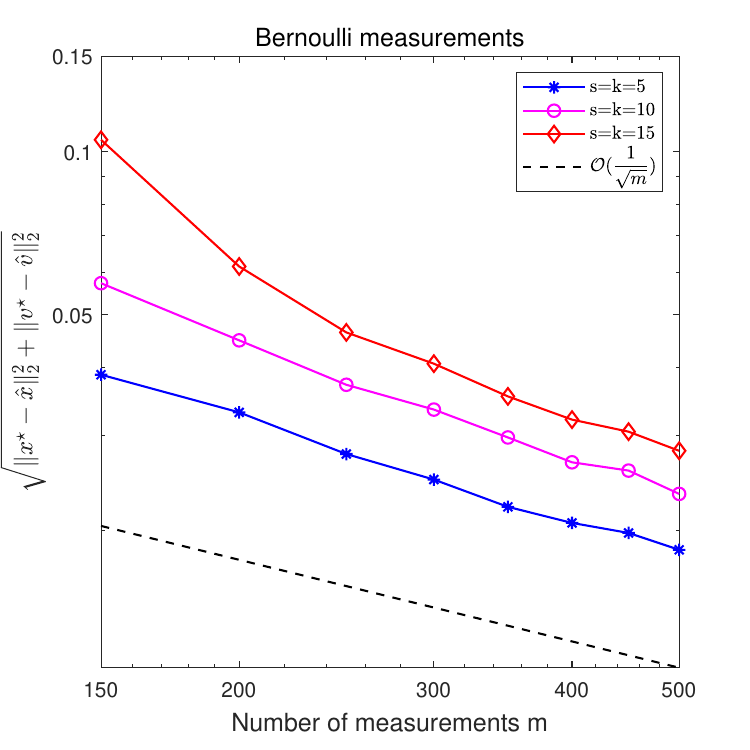}&
\includegraphics[width=0.33\textwidth]{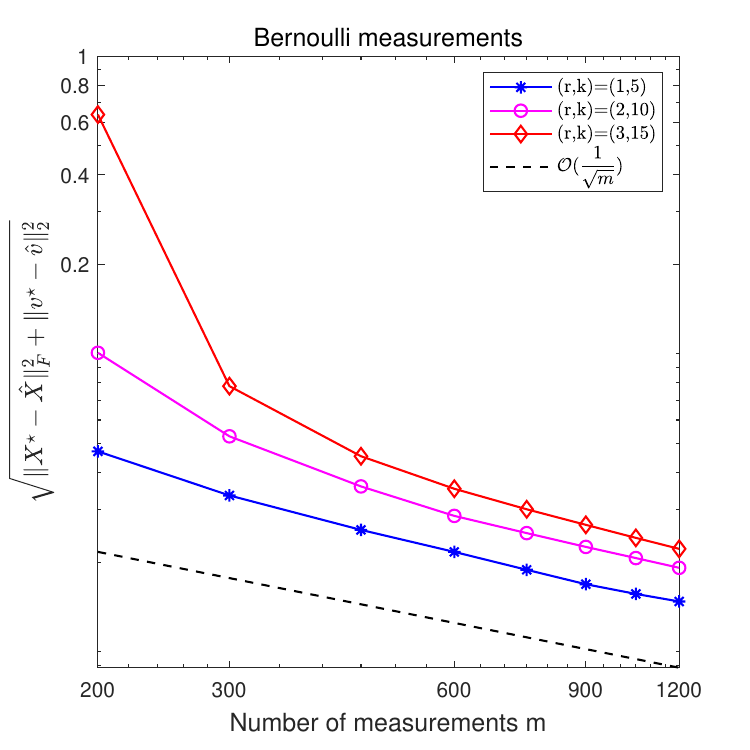}&
\includegraphics[width=0.33\textwidth]{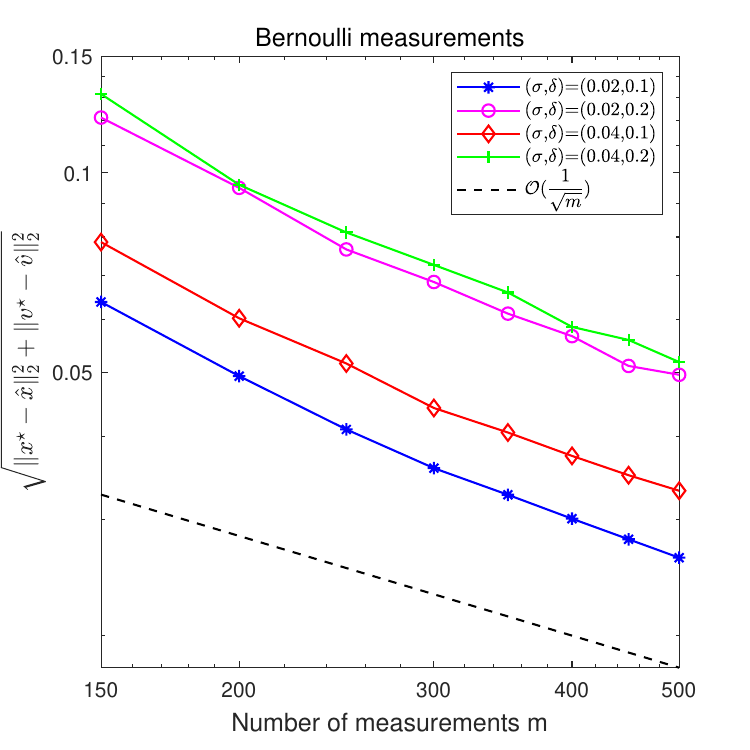}\\
 (d) Sparse signal recovery 
 & (e) Low-rank matrix recovery 
 & (f) Robustness to noise\\
 from sparse corruption 
 & from sparse corruption 
 & \\
\end{tabular}
\caption{Log-log error curves for the constrained Lasso under Gaussian or Bernoulli measurements.}
  \label{fig:ConstrainedLasso}
  \end{center}\vspace{-0.3cm}
\end{figure*}


\subsubsection{Constrained Lasso} 
In order to demonstrate the role of different parameters in our uniform bound, we proceed to more simulations for constrained Lasso. 

\textbf{Sparse Recovery from Sparse Corruption:} 
We simulate \cref{coro1} with  $(n,\delta,E)=(256,0.1,0)$ and   vary the measurement number $m$ between 150 and 500.  We consider both Gaussian design $\bm{\Phi}\sim \mathcal{N}^{m\times n}(0,1)$ and Bernoulli design that has i.i.d. zero-mean $\{-1,1\}$-valued entries.\footnote{Bernoulli design is an example that demonstrates the benefit of using dithering. Without the random dither $\bm{\tau}$, the identifiability issue   arises under Bernoulli design even  in compressed sensing without the corruption $\bm{v^\star}$ (e.g., \cite{chen2022quantizing,sun2022quantized}). On the other hand, by using dithering, recovery can be ensured   under general sub-Gaussian $\bm{\Phi}$.} We simulate the uniform recovery with  size-$100$ $\mathcal{X}_{test}$ in the cases of ``$s=k=5$'', ``$s=k=10$'' and ``$s=k=15$'', and we report  \cref{track} as log-log curves in 
 \cref{fig:ConstrainedLasso}(a) and \cref{fig:ConstrainedLasso}(d) for Gaussian design and Bernoulli design, respectively. 
Clearly, the results under two designs are similar, and note the   two observations that are consistent with our theory. First, all curves decrease with $m$ in the theoretical rate $O(m^{-1/2})$. Second, the errors increase under larger $(s,k)$, for which the intuition is that  weaker sparse priors on $(\bm{x^\star},\bm{v^\star})$ correspond to a harder high-dimensional estimation problem.

\textbf{Low-Rank Recovery from Sparse Corruption:} 
We simulate \cref{coro2} with  $(p,q,\delta,E)=(16,16,0.1,0)$ and vary measurements $m$ between 200 and 1200. Under a realization of $(\bm{\Phi},\bm{\tau})$, we track \cref{track} with size-$100$ $\mathcal{X}_{test}$ in the cases of ``$(r,k)=(1,5)$'', ``$(r,k)=(2,10)$'' and ``$(r,k)=(3,15)$''. The log-log curves corresponding to Gaussian and Bernoulli designs are shown in   \cref{fig:ConstrainedLasso}(b) and \cref{fig:ConstrainedLasso}(e)  respectively, which are consistent with our theoretical uniform bound.  

\begin{figure*}[!ht]
\scriptsize\setlength{\tabcolsep}{0.3pt}
\begin{center}
\begin{tabular}{ccc}
\includegraphics[width=0.33\textwidth]{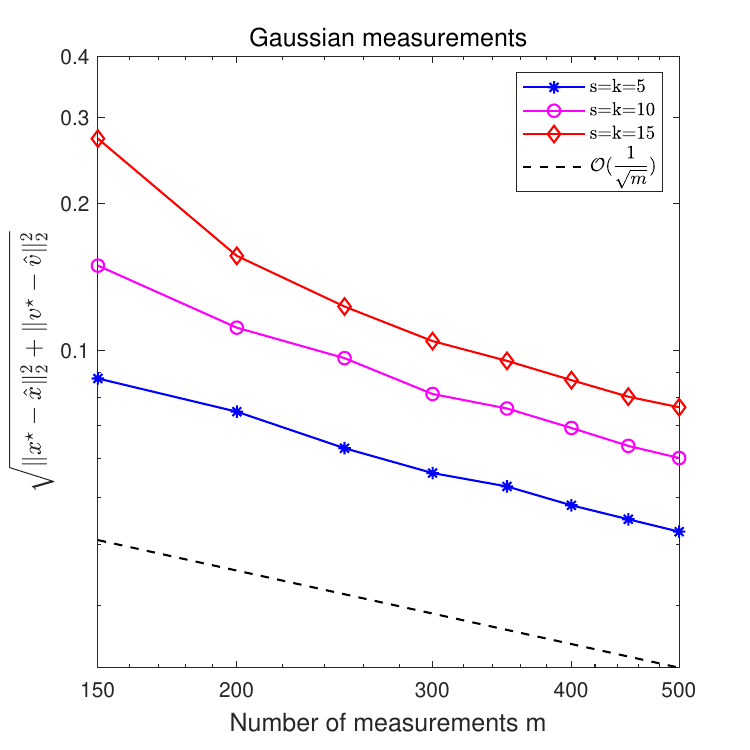}&
\includegraphics[width=0.33\textwidth]{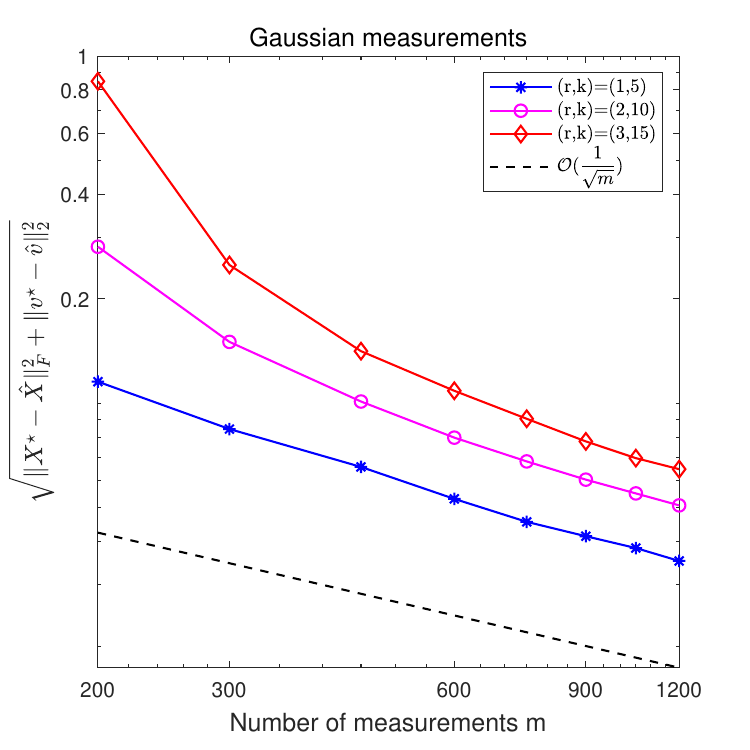}&
\includegraphics[width=0.33\textwidth]{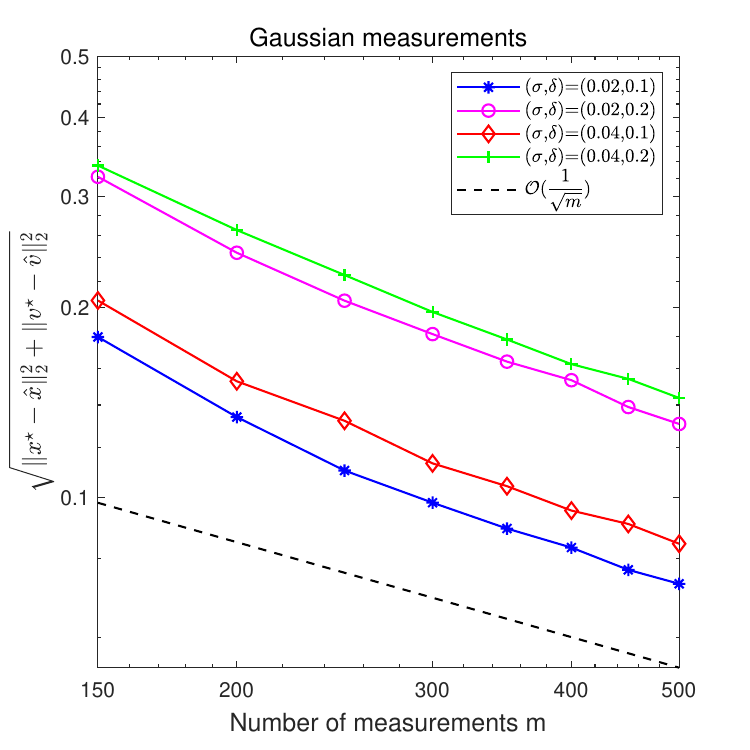}\\
 (a) Sparse signal recovery
 & (b) Low-rank matrix recovery 
 & (c) Robustness to noise\\
 from sparse corruption 
 & from sparse corruption 
 & \\
\includegraphics[width=0.33\textwidth]{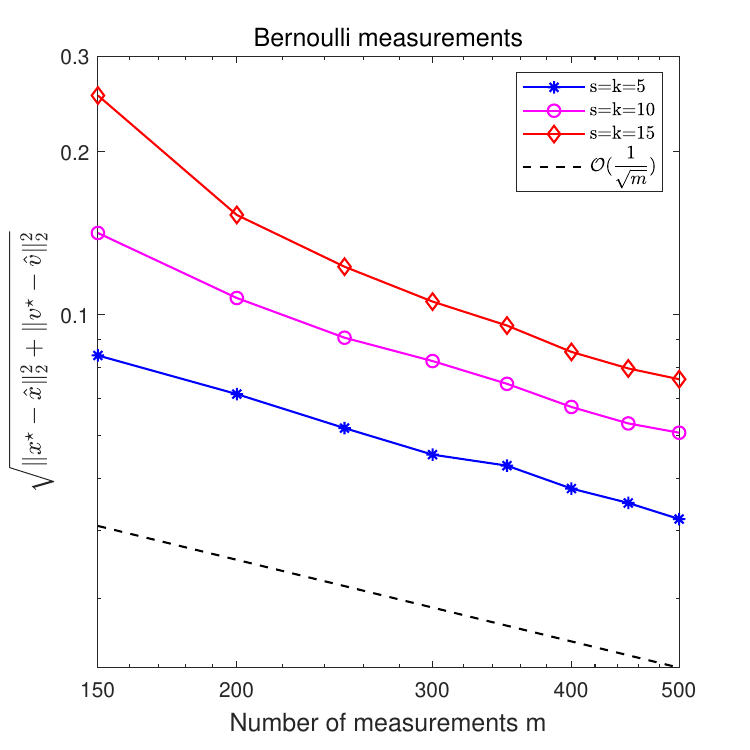}&
\includegraphics[width=0.33\textwidth]{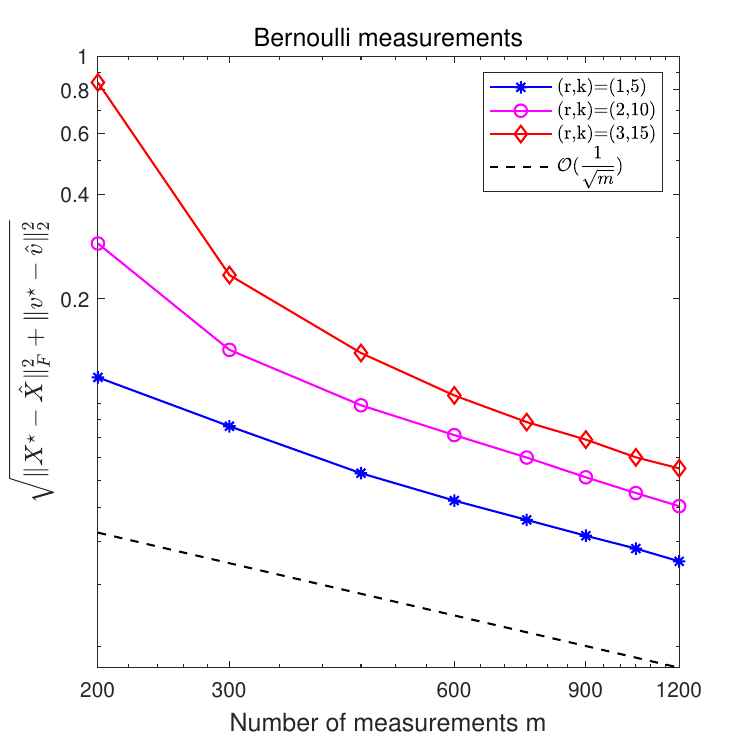}&
\includegraphics[width=0.33\textwidth]{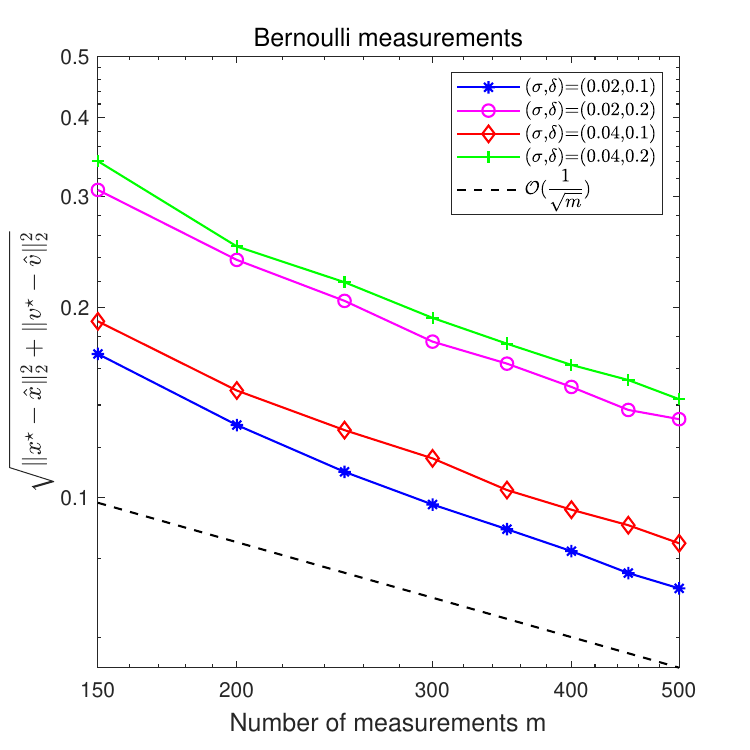}\\
 (d) Sparse signal recovery
 & (e) Low-rank matrix recovery
 & (f) Robustness to noise\\
 from sparse corruption 
 & from sparse corruption 
 & \\
\end{tabular}
\caption{Log-log error curves for the unconstrained Lasso under Gaussian or Bernoulli measurements.}
  \label{fig:UnconstrainedLasso}
  \end{center}\vspace{-0.3cm}
\end{figure*}

\textbf{The Role of $\delta$ and $E$:} We also use the setting of   \cref{coro1} to illustrate   the role played by the quantization resolution $\delta$ and the robustness to   noise $\bm{\epsilon}$. We test  Gaussian noise  $\bm{\epsilon}\sim \mathcal{N}(0,\sigma^2\bm{I}_m)$ (note that the noise level $\sigma$ can be simply understood as $E$ in   \cref{assump1}) in the cases of ``$(\sigma,\delta)=(0.02,0.1)$'', ``$(\sigma,\delta)=(0.02,0.2)$'', ``$(\sigma,\delta)=(0.04,0.1)$'' and ``$(\sigma,\delta)=(0.04,0.2)$''. The log-log curves  are displayed in   \cref{fig:ConstrainedLasso}(c) and \cref{fig:ConstrainedLasso}(f). As predicted by our uniform bound \cref{boundcoro1},    larger $\delta$ and severer sub-Gaussian noise (i.e., larger $E$) lift the curves higher but do not affect the decaying rate of $O(m^{-1/2})$.



\subsubsection{Unconstrained Lasso}
We conduct parallel experiments using   unconstrained Lasso \cref{3.5}. Note that unconstrained Lasso \cref{3.5} might be more practical than its constrained counterpart, in the sense that we can   use a program with   fixed large enough $(\lambda_1,\lambda_2)$ for all 100 pairs of $(\bm{x^\star},\bm{v^\star})$ in $ \mathcal{X}_{test}$. We show the results in   \cref{fig:UnconstrainedLasso}, which are consistent with the uniform error bounds in   \cref{coro3}--\cref{coro4} in terms of decaying rate, qualitative dependence on structured parameters, $\delta$ and $E$. 

\subsection{Generative Priors}

In this subsection, we present proof-of-concept experimental results for the case of using generative priors. In particular, we consider the case that the signal is close to the range of a generative model, and the corruption vector is also close to the range of another generative model. All the experiments were conducted using the Python 3.10.6 and PyTorch 2.0.0 framework on an NVIDIA RTX 3060 Laptop 6GB GPU. We modify \cref{track} and track the following two quantities for   signal and corruption:
\begin{equation}\label{track1}
    \sup_{(\bm{x^\star},\bm{v^\star})\in \mathcal{X}_{test}}\frac{\|\bm{\hat{x}}-\bm{x^\star}\|_2}{\|\bm{x^\star}\|_2}~,~~\sup_{(\bm{x^\star},\bm{v^\star})\in \mathcal{X}_{test}}\frac{\|\bm{\hat{v}}-\bm{v^\star}\|_2}{\|\bm{v^\star}\|_2},
\end{equation}
which are just the maximum relative error over a test set of $(\bm{x^\star},\bm{v^\star})$ denoted by $\mathcal{X}_{test}$.

\textbf{Demixing ``8'' from ``1'' in MNIST:}  
First, we follow~\cite{berk2020deep,berk2021deep} to train two variational auto-encoders (VAEs) for the training images of digits $8$ and $1$ in the MNIST dataset~\cite{lecun1998gradient} respectively. The decoders of these two VAEs were composed of a fully connected neural network with ReLU activation functions. The VAEs had an input dimension $k=k'=20$ and an output dimension of $m = n=28\times 28 = 784$, with two hidden layers consisting of 500 neurons each. We used the Adam optimizer with a mini-batch size of 100 and a learning rate of $0.001$ to train these VAEs.

We take the images of digit 8 as the   signal and those of digit 1 as the   corruption vectors. We use $\bm{\epsilon}\sim \mathcal{N}(0,\sigma^2\bm{I}_m)$ to simulate sub-Gaussian noise. 
 To demonstrate uniform recovery, we use a single realization of $(\bm{\Phi},\bm{\epsilon},\bm{\tau})$ to track the maximum relative error in \cref{track1}, where
$\mathcal{X}_{test}$ contains 20 test images of digits 8 and 1 from the testing set of MNIST.
As before, we report the quantities  in \cref{track1}  as its mean value in   10 independent random trials. 

We use the constrained Lasso \cref{geneprogram} to reconstruct $(\bm{x^\star},\bm{v^\star})$. Similarly to the algorithm proposed in~\cite{bora2017compressed}, we employ the gradient descent algorithm to minimize the following objective function over $\mathbb{R}^{k} \times  \mathbb{R}^{k'}$:
 \begin{equation}\nonumber
     \mathcal{L}(\bm{z}, \bm{z}') := \|\dot{\bm{y}} - \bm{\Phi} G(\bm{z}) - \sqrt{m} H(\bm{z'})\|_2, 
 \end{equation}
 {\color{black}where $G:\mathbb{B}_2^k(r)\to \mathbb{R}^n$, $H:\mathbb{B}_2^{k'}(r')\to \mathbb{R}^m$ are the generative models as per \cref{assump5}.}\footnote{Since $r$ and $r'$ can typically scale as large as $n^{\Theta(d)}$ with $d$ being the number of layers~\cite{bora2017compressed}, we do not impose constraints of $\|\bm{z}\|_2 \le r$ and $\|\bm{z}'\|_2 \le r'$ in our experiments.} 
 Our algorithm is referred to as~\texttt{QCS\_Gen}. We follow the settings in~\cite{bora2017compressed} and perform $10$ random restarts with $1000$ gradient descent steps per restart. The optimal reconstruction is chosen based on the lowest measurement error. 

We test the sensing matrix $\bm{\Phi}\in \mathbb{R}^{n\times n}$ under both Gaussian and Bernoulli designs. Examples of reconstructed signals and corruptions are presented in \cref{fig:images_vaes_gaussian} and~\cref{fig:images_vaes_bernoulli}, with the quantitative results showcased in  ~\cref{fig:quant_vaes_gaussian} and~\cref{fig:quant_vaes_bernoulli}. The $\ell_2$-norm of the test images is about $10$, and the reconstructed images have impressive accuracy, even in a  coarsely quantized and highly noisy setting with  $\delta = 20$ and $\sigma = 10$. 

   \begin{figure}
     \includegraphics[height=0.4\textwidth]{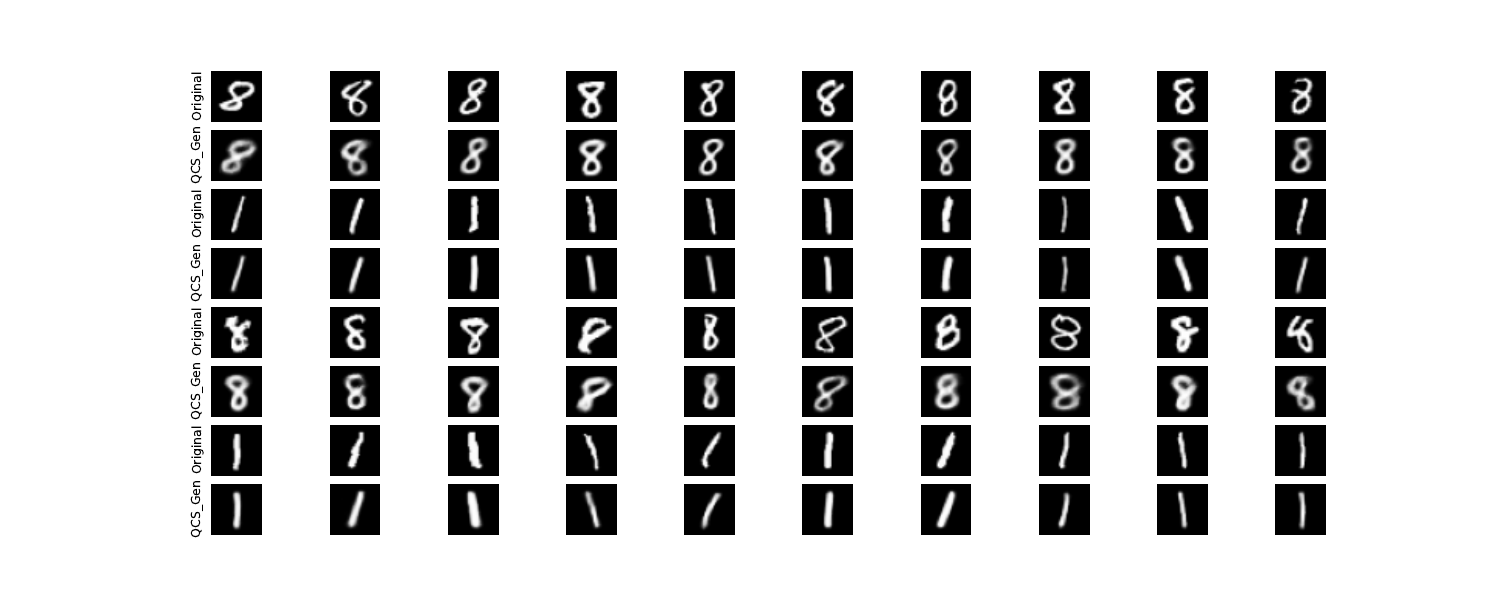}
     \vspace{-1.cm}
     \caption{Reconstructed images for digits 8 and 1 of MNIST under Gaussian measurements with $\sigma =2$ and $\delta = 10$.}
     \label{fig:images_vaes_gaussian}
 \end{figure}

    \begin{figure}
     \includegraphics[height=0.4\textwidth]{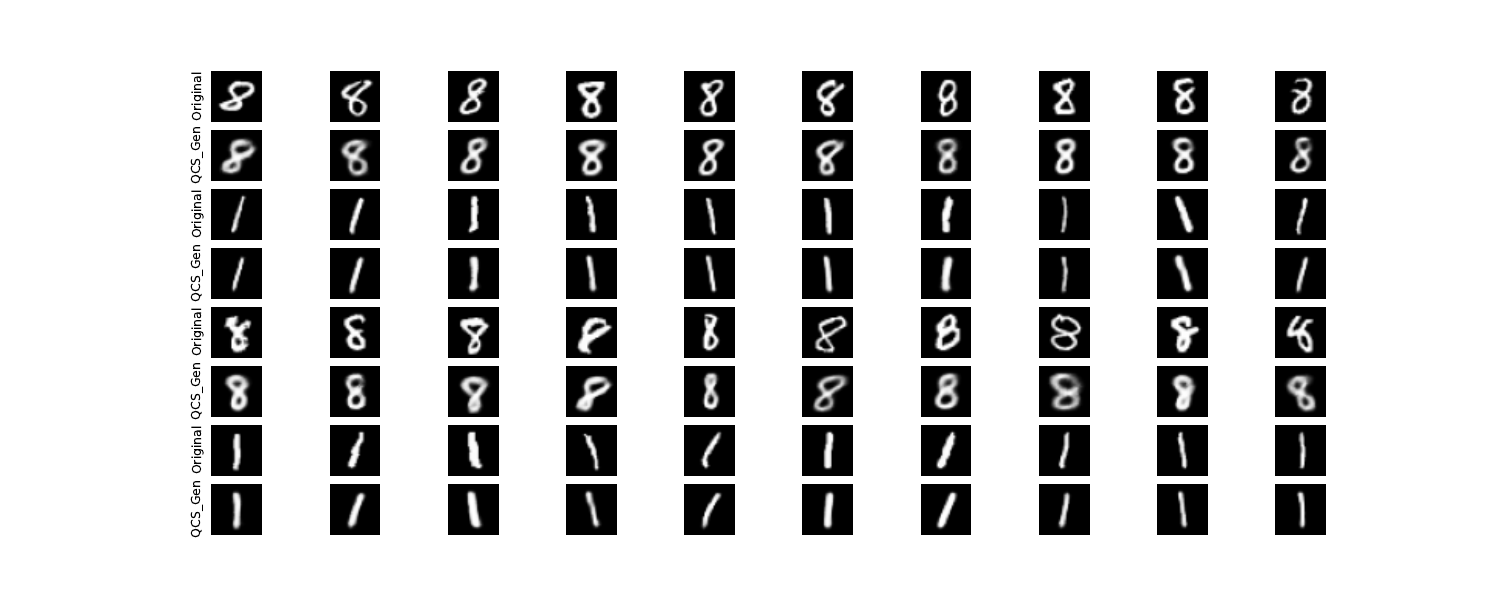}
     \vspace{-1.cm}
     \caption{Reconstructed images for digits 8 and 1 of MNIST under Bernoulli measurements with $\sigma =10$ and $\delta = 20$.}
     \label{fig:images_vaes_bernoulli}
 \end{figure}

   \begin{figure}
\begin{center}
\begin{tabular}{cc}
\includegraphics[height=0.37\textwidth]{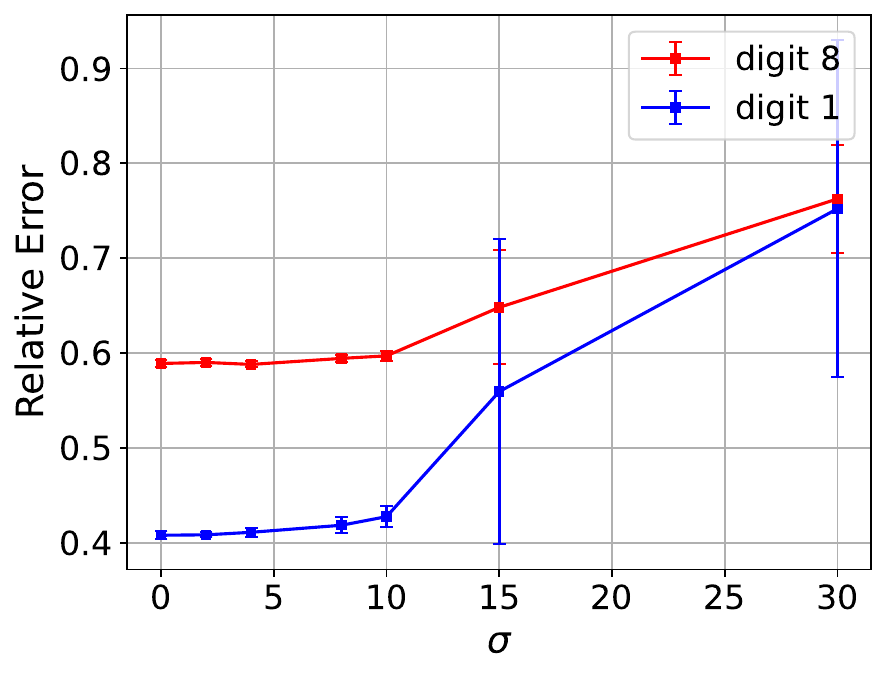} & \hspace{-0.5cm}
\includegraphics[height=0.37\textwidth]{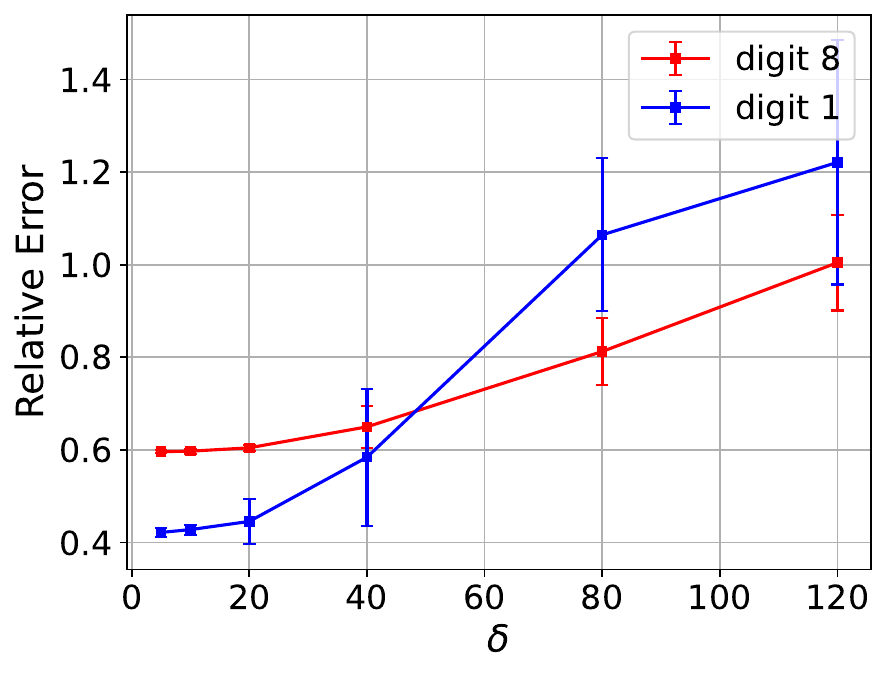} \\
{\small (a) Varying $\sigma$ with fixed $\delta = 10$} & \hspace{-0.8cm} {\small (b) Varying $\delta$ with fixed $\sigma = 10$} 
\end{tabular}
\caption{Quantitative results of the performance of~\texttt{QCS\_Gen} under Gaussian measurements for digits 8 and 1 of MNIST.} \label{fig:quant_vaes_gaussian} \vspace*{-2ex}
\end{center}
\end{figure}

   \begin{figure}
\begin{center}
\begin{tabular}{cc}
\includegraphics[height=0.37\textwidth]{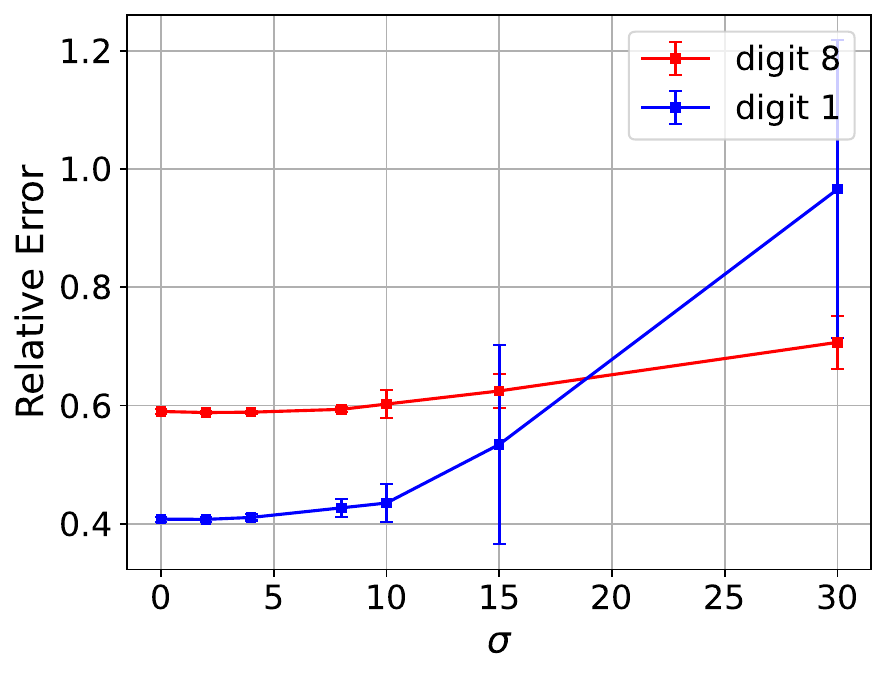} & \hspace{-0.5cm}
\includegraphics[height=0.37\textwidth]{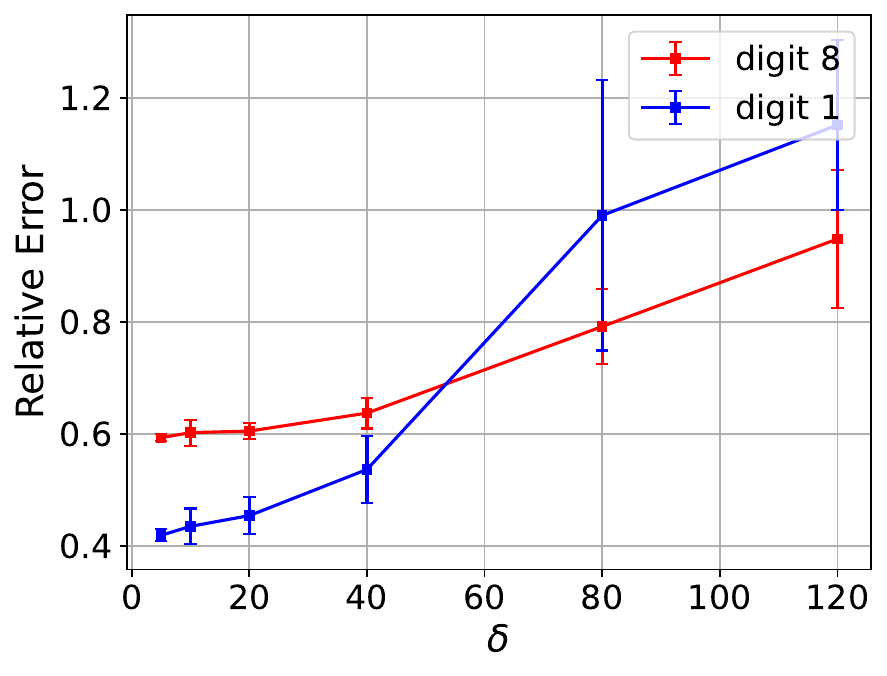} \\
{\small (a) Varying $\sigma$ with fixed $\delta = 10$} & \hspace{-0.8cm} {\small (b) Varying $\delta$ with fixed $\sigma = 10$} 
\end{tabular}
\caption{Quantitative results of the performance of~\texttt{QCS\_Gen} under Bernoulli measurements for digits 8 and 1 of MNIST.} \label{fig:quant_vaes_bernoulli} \vspace*{-2ex}
\end{center}
\end{figure}


\textbf{Demixing Images in CelebA from Digits in MNIST:} Since the image vectors of digits 8 and 1 have the same data dimension, the above experiment with $\bm{x^\star}$ and $\bm{v^\star}$ both from MNIST dataset is a demixing task without compression. In order to achieve compression, we additionally performed experiments for the case in which the signals are selected from the test images of the CelebA dataset~\cite{liu2015deep} with data dimension $n = 3\times 64 \times 64 = 12288$ and the corruptions correspond to the test images of the MNIST dataset   with data dimension $m = 784$. Since the number of measurements $m$ is much smaller than $n$, the measurement matrix $\bm{\Phi}$ of dimension $m \times n = 784\times12288$ is used  for our simulations of   corrupted sensing with compression (as $m\ll n$). 

We train a VAE on the training set of the MNIST dataset, which comprises 60,000 images for digits 0 to 9. The decoder of the VAE is a fully connected neural network with two hidden layers and 500 neurons each, ReLU activations in the layers, and an input dimension of $k=20$ and output dimension of $n=784$. We also employ the Adam optimizer with a mini-batch size of 100 and a learning rate of $0.001$ for training.

The CelebA database contains more than 200,000 face images of celebrities, on which we train a deep convolutional generative adversarial network (DCGAN) following the settings in~\url{https://pytorch.org/tutorials/beginner/dcgan_faces_tutorial.html}. The latent dimension of the generator for this model is $100$ and the number of epochs for training is $20$. We select $20$ images from the test set of CelebA as our signals and $20$ test images of MNIST as our corruptions, and we conduct $5$ random trials. As the images of CelebA and MNIST differ significantly in their $\ell_2$-norm, we normalize both of them to have unit $\ell_2$-norm prior to generating the quantized observations. All other settings remain the same as those applied for the case of using two VAEs for digits 8 and 1 of MNIST.

Examples of reconstructed signals and corruptions can be seen in  ~\cref{fig:images_vaegan_gaussian} and~\cref{fig:images_vaegan_bernoulli}, and the relative error for each is quantified in  ~\cref{fig:quant_vaegan_gaussian} and~\cref{fig:quant_vaegan_bernoulli}. We observe that when $m$ is much smaller than $n$, accurate reconstructions of the multiple test images can be achieved using a single draw of $(\bm{\Phi},\bm{\epsilon},\bm{\tau})$, as theoretically supported by our uniform recovery guarantee. Consistent with \cref{genesample} in   \cref{thm3}, under fixed $\delta$, the error increases under larger $\sigma$. Also, while fixing the noise level $\sigma$,  larger $\delta$ (that represents coarser quantization) corresponds to larger error, indicating that a trade-off between quantization resolution and 
recovery accuracy is important in practice.

   \begin{figure}
     \includegraphics[height=0.40\textwidth]{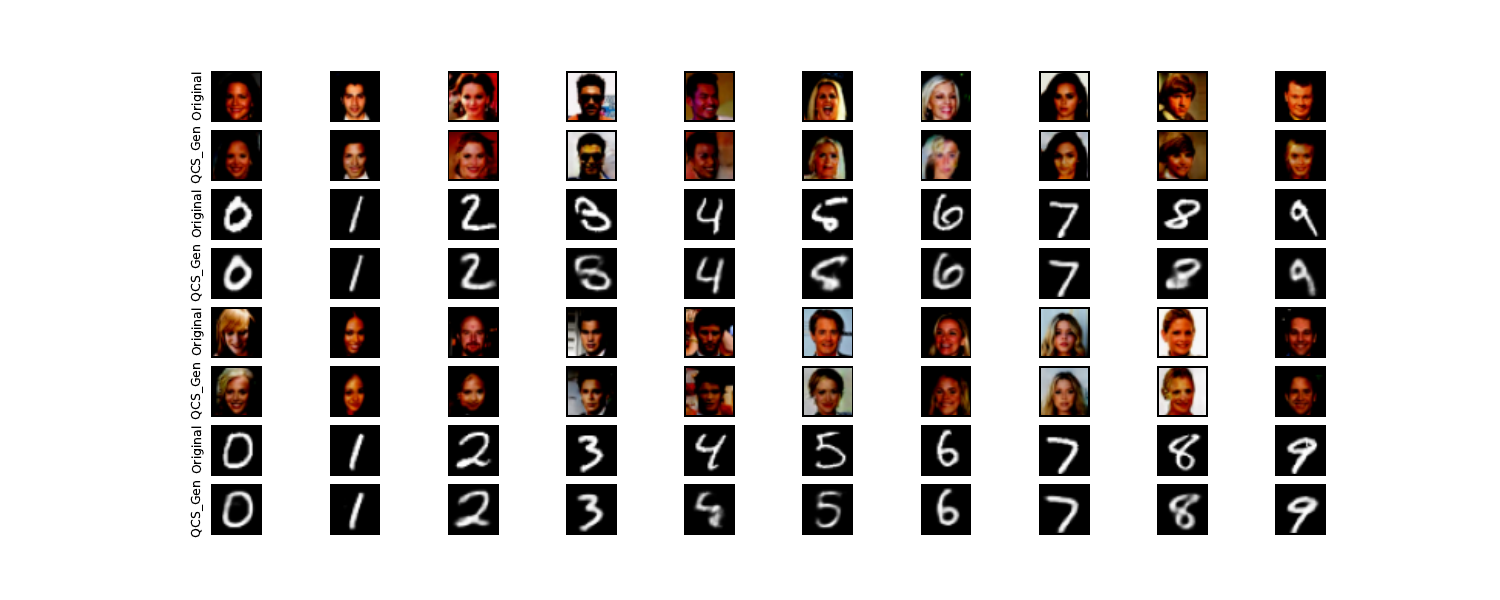}
     \vspace{-1.cm}
     \caption{Reconstructed images for CelebA and MNIST under Gaussian measurements and $\sigma =0.1$, $\delta = 1.0$.}
     \label{fig:images_vaegan_gaussian}
 \end{figure}

    \begin{figure}
     \includegraphics[height=0.40\textwidth]{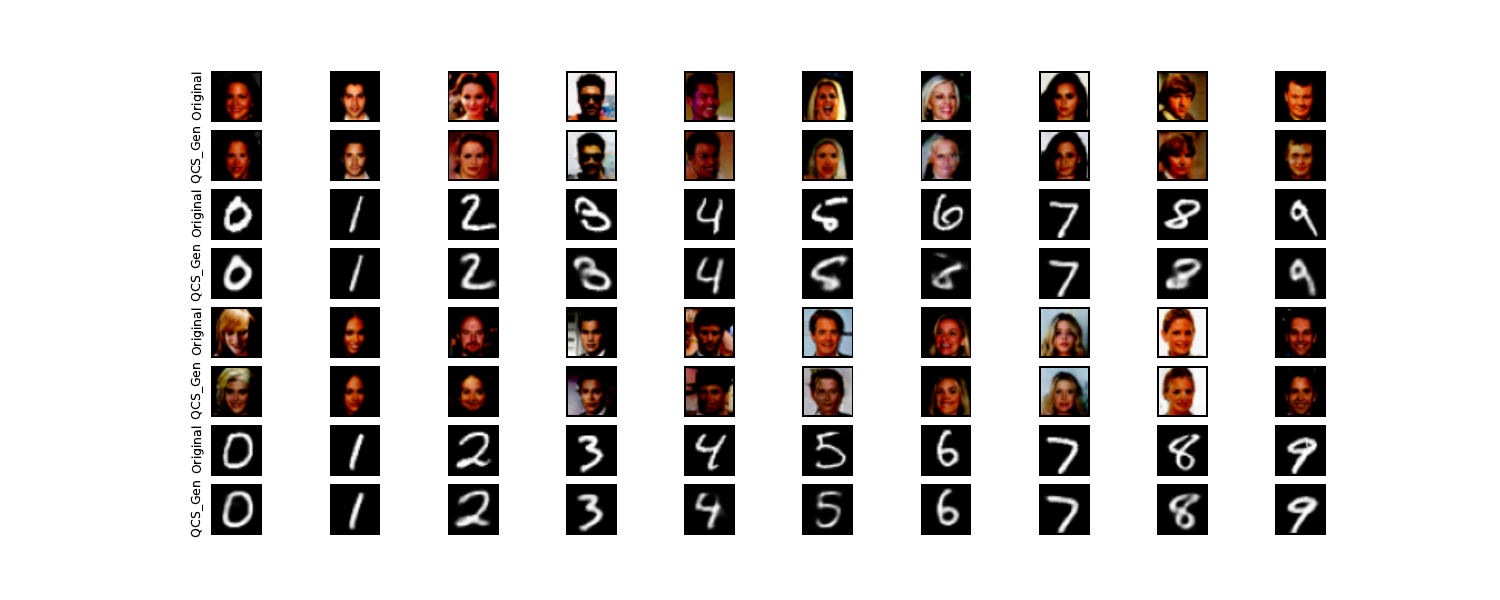}
     \vspace{-1.cm}
     \caption{Reconstructed images for CelebA and MNIST under Bernoulli measurements and $\sigma =0.2$, $\delta = 0.2$.}
     \label{fig:images_vaegan_bernoulli}
 \end{figure}

   \begin{figure}
\begin{center}
\begin{tabular}{cc}
\includegraphics[height=0.36\textwidth]{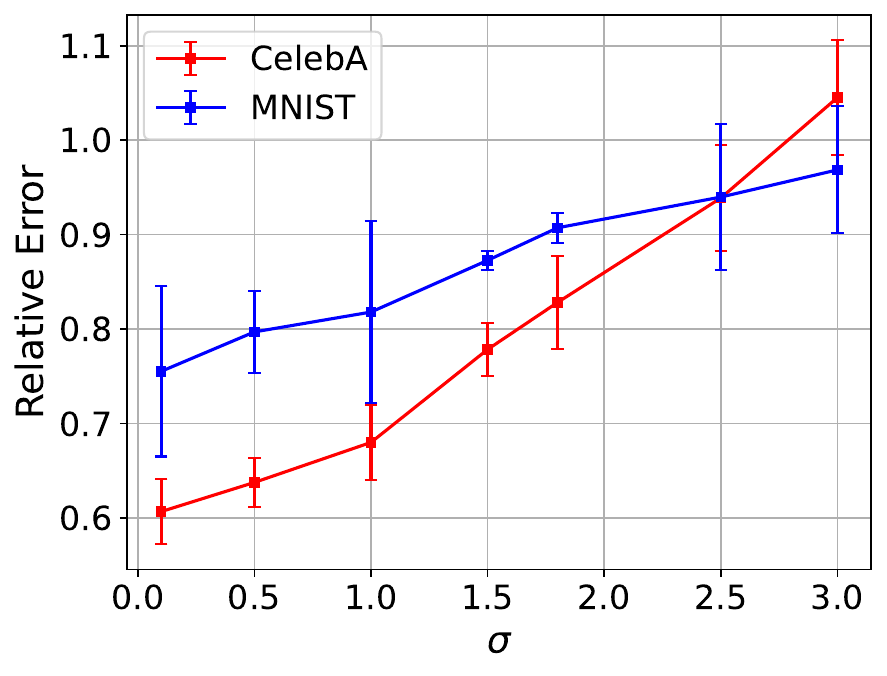} & \hspace{-0.5cm}
\includegraphics[height=0.36\textwidth]{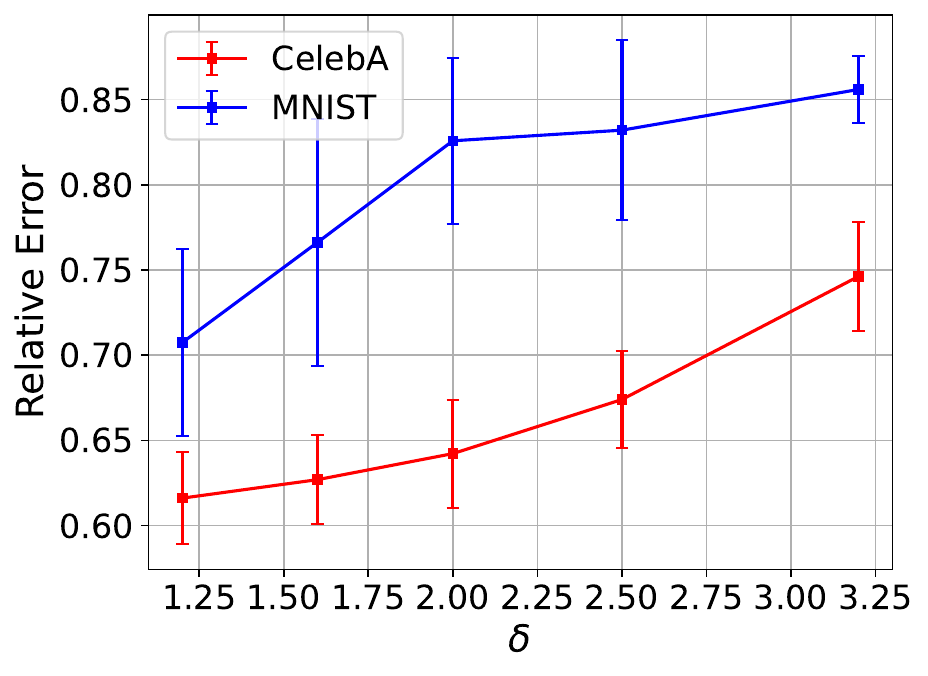} \\
{\small (a) Varying $\sigma$ with fixed $\delta = 1.0$} & \hspace{-0.8cm} {\small (b) Varying $\delta$ with fixed $\sigma = 0.2$} 
\end{tabular}
\caption{Quantitative results of the performance of~\texttt{QCS\_Gen} under Gaussian measurements for CelebA and MNIST.} \label{fig:quant_vaegan_gaussian} \vspace*{-2ex}
\end{center}
\end{figure}

   \begin{figure}
\begin{center}
\begin{tabular}{cc}
\includegraphics[height=0.36\textwidth]{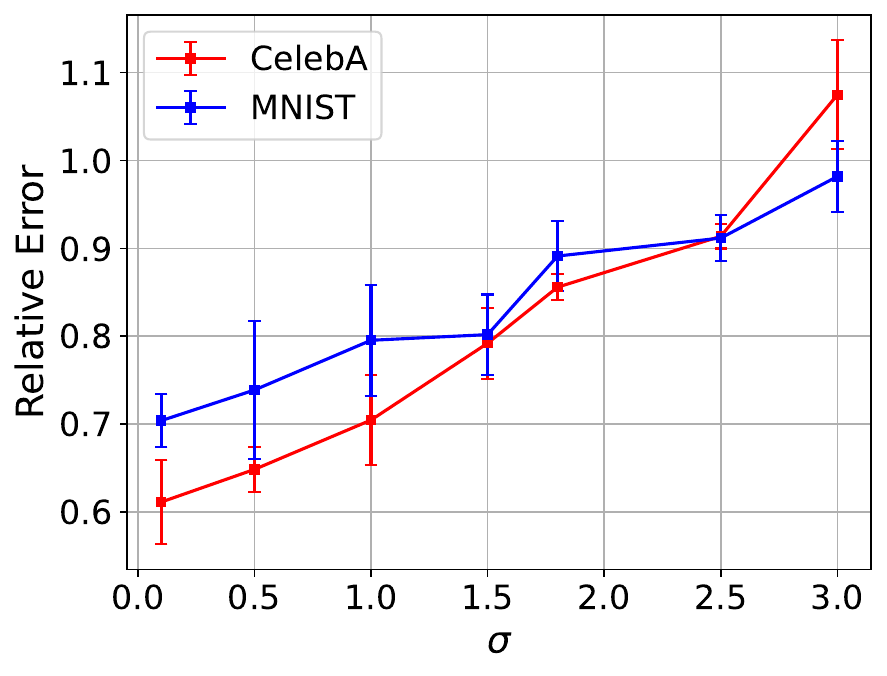} & \hspace{-0.5cm}
\includegraphics[height=0.36\textwidth]{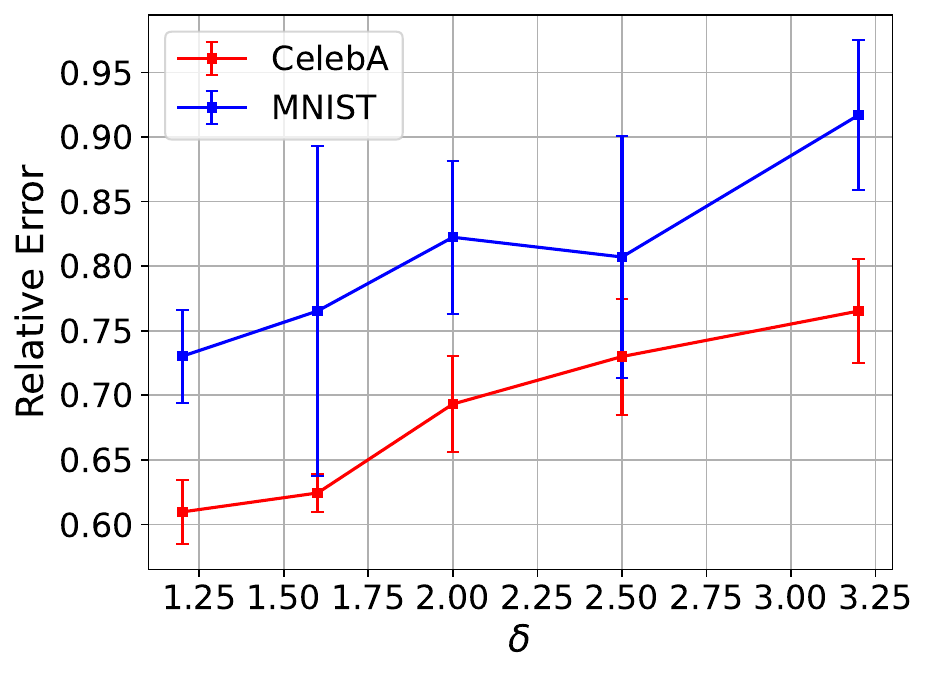} \\
{\small (a) Varying $\sigma$ with fixed $\delta = 1.0$} & \hspace{-0.8cm} {\small (b) Varying $\delta$ with fixed $\sigma = 0.2$} 
\end{tabular}
\caption{Quantitative results of the performance of~\texttt{QCS\_Gen} under Bernoulli measurements for CelebA and MNIST.} \label{fig:quant_vaegan_bernoulli} \vspace*{-2ex}
\end{center}
\end{figure}

\section{Concluding Remarks}\label{sec5:conclu}
In this work, we established uniform recovery guarantees for the problem of quantized corrupted sensing,  using a dithered uniform quantizer, as well as structured or generative priors on the signal and corruption. Unlike the non-uniform results in \cite{sun2022quantized}, our results allow one to use a fixed realization of the sensing ensemble to accurately recover all signals and corruptions of interest.  Our main techniques to prove uniformity are certain quantized embedding properties obtained from   covering arguments; based on them, interestingly, our treatments for structured priors and generative priors are nearly unified. All the uniform error bounds exhibit decaying rates of $O\big(\frac{1}{\sqrt{m}}\big)$. Specifically, the bounds for constrained Lasso typically coincide with the corresponding non-uniform ones up to logarithmic factors, while those for unconstrained Lasso usually display worse dependence on the structured parameters like sparsity level or matrix rank, creating a gap between  constrained Lasso and unconstrained Lasso whose closing presents an open question.  To demonstrate uniform recovery, in the experiments we use a fixed realization of $(\bm{\Phi},\bm{\epsilon},\bm{\tau})$ to accurately recover all $(\bm{x^\star},\bm{v^\star})$   in some testing sets, and the numerical results are consistent with our uniform bounds. For future research, besides the aforementioned open question concerning using unconstrained Lasso for uniform recovery, corrupted sensing under 1-bit quantization is also worth investigation.

\begin{appendix}
\section{Auxiliary Facts}\label{appenA}
We collect  auxiliary results in this Appendix, including some known concentration inequalities and technical lemmas that support our proofs. 
\subsection{Concentration Inequalities}\label{appendixA1}
First of all, we provide the extended matrix deviation inequality that is well tailored to suit the analysis of corrupted sensing. We comment that setting $\mathcal{T}=\mathcal{T}_0\times \{0\}$ for some $\mathcal{T}_0\subset \mathbb{R}^n$ returns the regular matrix deviation inequality  in \cite{jeong2022sub,liaw2017simple,vershynin2018high}  
\begin{proposition}[Extended matrix deviation inequality, Theorem 1 in \cite{chen2018stable}]
\label{pro1} Let $\bm{\Phi}\in \mathbb{R}^{m\times n}$ be the sub-Gaussian  sensing matrix described in   \cref{assump1}, we let $\mathcal{T}\subset \mathbb{R}^n\times \mathbb{R}^m$ be a bounded subset, then for any $t\geq 0$, the event \begin{equation}\nonumber
    \sup_{(\bm{a},\bm{b})\in\mathcal{T}}\Big|\|\bm{\Phi a}+\sqrt{m}\bm{b}\|_2-\sqrt{m}\big(\|\bm{a}\|_2^2+\|\bm{b}\|_2^2\big)^{1/2}\Big|\leq C \big(\gamma(\mathcal{T})+t\cdot\mathrm{rad}(\mathcal{T})\big)
\end{equation} 
holds with probability at least $1-\exp(-t^2)$.
\end{proposition}

We will utilize the following result   to prove a local version of the quantized product embedding property (see  \cref{lem:localqpe}).
\begin{proposition}[Exercise 8.6.5 in \cite{vershynin2018high}]\label{pro2}
     Consider a random process $(Z_{\bm{a}})_{\bm{a}\in \mathcal{T}}$ indexed by points $\bm{a}$ in a bounded subset $\mathcal{T}\subset \mathbb{R}^n$. Assuming that $Z_0=0$, and for all $\bm{a},\bm{b}\in \mathcal{T}\cup\{0\}$ we have 
     \begin{align}\label{eq:increment}
        \|Z_{\bm{a}}-Z_{\bm{b}}\|_{\psi_2}\leq M\|\bm{a}-\bm{b}\|_2. 
    \end{align}
    Then  for any $t\geq 0$, the event  \begin{align} 
        \sup_{\bm{a}\in\mathcal{T}}\big|Z_{\bm{a}}\big|\leq CM\cdot \big[\omega(\mathcal{T})+t\cdot \mathrm{rad}(\mathcal{T})\big]
    \end{align}
    holds with probability exceeding $1-2\exp(-t^2)$.
\end{proposition}

Next, we present a result that precisely characterizes the range of a low-complexity set under the sub-Gaussian map $\bm{\Phi}$. In particular, it provides uniform bound on the $l$-th largest measurement (since this is evidently no larger than the left-hand side of \cref{eq:bound_l_large} below), which   proves an effective tool in bounding the number of large perturbations (see \cref{eq:bound_ell_large} in the proof of \cref{thm:globalqpe}). We note that \cref{prop:bound_l_largest} is the most crucial ingredient for getting our  improvement on \cite{xu2020quantized}, which is to be presented in \cref{app:improve}.  

\begin{proposition}[Adapted from Theorem 2.10 in  \cite{dirksen2021non}]\label{prop:bound_l_largest}
   Let $\bm{\Phi}_1,...,\bm{\Phi}_m$ be independent, isotropic sub-Gaussian sensing vectors satisfying $\max_i\|\bm{\Phi}_i\|_{\psi_2}=O(1)$, we consider some $\mathcal{T}\subset \mathbb{R}^n$. If $1\le l\le m$, then for some absolute constants $C_1,C_2$,  
    the event\footnote{In the original statement of \cite[Thm. 2.10]{dirksen2021non}, $\omega(\mathcal{T})$ in the right-hand side of \cref{eq:bound_l_large} should be $\gamma(\mathcal{T})$, while we can safely use $\omega(\mathcal{T})$ here because \cref{widthcomple} gives  $\gamma(\mathcal{K})\le 2\omega(\mathcal{K})+2\rad(\mathcal{T})$, and we observe that $\frac{\rad(\mathcal{T})}{l^{1/2}}\le \rad(\mathcal{T})(\log\frac{em}{l})^{1/2}$ holds for any $l\in[1,m]$.}
    \begin{align}\label{eq:bound_l_large}
        \sup_{\bm{a}\in\mathcal{T}}\max_{\substack{I\subset [m]\\|I|\le l}}\Big(\frac{1}{l}\sum_{i\in I}|\langle\bm{\Phi}_i,\bm{a}\rangle|^2\Big)^{1/2} \le C_1 \Big(\frac{\omega(\mathcal{T})}{\sqrt{l}}+\rad(\mathcal{T})\sqrt{\log\frac{em}{l}}\Big)
    \end{align}
    holds with probability at least $1-2\exp(-C_2l\log\frac{em}{l})$. 
\end{proposition}

Finally, we present a Chevet-type inequality   that is indeed a simple outcome of \cref{pro2}.


\begin{proposition}[Chevet-type inequality, Lemma 4 in \cite{chen2018stable}] \label{pro4}
    Let $\bm{\Phi}$ be the sub-Gaussian sensing matrix described in   \cref{assump1}, $\bm{w}\in \mathbb{R}^m$ be fixed, and $\mathcal{T}$ be a bounded subset of $\mathbb{R}^n$. Then for any $t\geq 0$, the event  \begin{equation}
        \sup_{\bm{a}\in \mathcal{T}}\big|\langle\bm{w},\bm{\Phi a}\rangle\big| \leq C\|\bm{w}\|_2\big(\omega(\mathcal{T})+t\cdot\mathrm{rad}(\mathcal{T})\big),
    \end{equation}
    holds with probability exceeding $1-2\exp(-t^2)$. 
\end{proposition}

\subsection{Estimates of Geometric Quantities}\label{appendixA.2}
We collect some useful estimates on geometric quantities, specifically on Gaussian width, Gaussian complexity or Kolmogorov entropy (a notion equivalent to covering number).
We begin with the well-known bounds on Gaussian width and Kolmogorov entropy for the structured priors of sparsity and low-rankness. 
\begin{proposition}[See, e.g., \cite{plan2012robust}]
    \label{pro5}
    We can estimate the Gaussian width of $\Sigma^n_s\cap \mathbb{S}^{n-1}$, $\mathbb{B}_2^n\cap \sqrt{s}\mathbb{B}^n_1$   as follows:\begin{equation}\nonumber
        \max\big\{\omega(\Sigma^n_s\cap \mathbb{S}^{n-1}),\omega(\mathbb{B}_2^n\cap \sqrt{s}\mathbb{B}_1^n)\big\}\asymp \sqrt{s\log\frac{en}{s}}.
    \end{equation}
\end{proposition}

\begin{proposition}[See, e.g.,  \cite{plan2013one,candes2011tight,baraniuk2007simple}]\label{pro6}
  Given some $\epsilon>0$, recall that $\Sigma_s^n$ is the set of all $s$-sparse signals in $\mathbb{R}^n$, and $M_{r}^{p,q}$ is the set of $p\times q$  matrices   with rank  not exceeding $r$, 
    then we have    \begin{equation}\nonumber
        \mathscr{H}(\Sigma_{s}^n\cap \mathbb{B}_2^{n},\epsilon)\leq s\log\Big( \frac{9n}{\epsilon s} \Big)~~\mathrm{and}~~\mathscr{H}(M_r^{p,q}\cap \mathbb{B}_{\rm F}^{p,q},\epsilon)\leq 2r(p+q)\log \Big(\frac{9}{\epsilon}\Big).
    \end{equation}   
\end{proposition}


Then, we give two results on relations between complexity quantities. With structured priors,  \cref{pro7} will be used in the analysis of   constrained Lasso, while \cref{pro9} is for analyzing unconstrained Lasso. 
\begin{proposition}
    \label{pro7}
    Given two cones $\mathcal{A}\subset \mathbb{R}^n$ and  $\mathcal{B}\subset \mathbb{R}^m$,   let $\mathcal{C}=\mathcal{A}\times \mathcal{B}$  and   write $\mathcal{A}^*=\mathcal{A}\cap \mathbb{S}^{n-1}$, $\mathcal{B}^*=\mathcal{B}\cap \mathbb{S}^{m-1}$, $\mathcal{C}^*=\mathcal{C}\cap \mathbb{S}^{m+n-1}$.  Then we have \begin{equation}\nonumber
        \max\{\gamma(\mathcal{A}^*),\gamma(\mathcal{B}^*)\}\leq \gamma(\mathcal{C}^*) \leq \gamma(\mathcal{A}^*)+\gamma(\mathcal{B^*}).
    \end{equation}
\end{proposition}
\begin{proof}
    Because $0\in \mathcal{A}$ and $0\in \mathcal{B}$, it is easy to see that $\mathcal{A}^*\times \{0\}\subset \mathcal{C}^*$ and $\{0\}\times \mathcal{B}^*\subset \mathcal{C}^*$, which implies $\gamma(\mathcal{A}^*)=\gamma(\mathcal{A}^*\times \{0\})\leq \gamma(\mathcal{C}^*)$, and similarly $\gamma(\mathcal{B}^*)\leq \gamma(\mathcal{C}^*)$. To prove the second inequality, we first observe that $\mathcal{C}^*\subset\cup_{c\in [0,1]} \big(c\mathcal{A}^*\times \sqrt{1-c^2}\mathcal{B}^*\big)$, and we let $\bm{g}_i$ be standard Gaussian vector with self-evident dimension. Then   following the definition of Gaussian complexity we have \begin{equation}
        \begin{aligned}\nonumber
        \gamma(\mathcal{C}^*)&\leq \mathbbm{E}\sup_{\substack{c\in[0,1]\\\bm{a}\in \mathcal{A}^*,\bm{b}\in \mathcal{B}^*}}\Big|\langle\bm{g}_1,c\cdot \bm{a}\rangle+\langle\bm{g}_2,\sqrt{1-c^2}\cdot \bm{b}\rangle\Big|\\
           &\leq \mathbbm{E}\sup_{\bm{a}\in\mathcal{A}^*}\big|\langle\bm{g}_1,\bm{a}\rangle\big|+\mathbbm{E}\sup_{\bm{b}\in \mathcal{B}^*}\big|\langle\bm{g}_2,\bm{b}\rangle\big|=\gamma(\mathcal{A}^*)+\gamma(\mathcal{B}^*),
        \end{aligned}
    \end{equation}
   which completes the proof.
\end{proof}

\begin{proposition}
    \label{pro9}
    Let $f(\cdot)$ be a norm in $\mathbb{R}^n$, $g(\cdot)$ be a norm in $\mathbb{R}^m$, $\lambda_1,\lambda_2,\alpha_{\bm{x}},\alpha_{\bm{v}}$ be some positive parameters, $C$ be some absolute constant. Define the cone \begin{align}\label{eq:defi_Clam12}
        \mathcal{C}(\lambda_1,\lambda_2)=\big\{(\bm{a},\bm{b})\in\mathbb{R}^n\times \mathbb{R}^m:\lambda_1f(\bm{a})+ \lambda_2g(\bm{b})\leq C\lambda_1\alpha_{\bm{x}}\|\bm{a}\|_2+C\lambda_2\alpha_{\bm{v}}\|\bm{b}\|_2\big\} 
    \end{align} and  let $\mathcal{C}^*=\mathcal{C}(\lambda_1,\lambda_2)\cap \mathbb{S}^{n+m-1},$ 
    then we have 
    \begin{equation}\label{A.17}
        \omega(\mathcal{C}^*)\lesssim \Big(\alpha_{\bm{x}}+\frac{\lambda_2\alpha_{\bm{v}}}{\lambda_1}\Big)\cdot\omega(\mathbb{B}_f^n)+\Big(\alpha_{\bm{v}}+\frac{\lambda_1\alpha_{\bm{x}}}{\lambda_2}\Big)\cdot\omega(\mathbb{B}_g^m).
    \end{equation}
\end{proposition}
\begin{proof}
    Let $\bm{h}_1\sim \mathcal{N}(0,\bm{I}_n)$, $\bm{h}_2\sim \mathcal{N}(0,\bm{I}_m)$, then we proceed as 
     
    \begin{align} 
\omega(\mathcal{C}^*)&=\mathbbm{E}\sup_{(\bm{a},\bm{b})\in \mathcal{C}^*}\langle\bm{h}_1,\bm{a}\rangle+\langle\bm{h}_2,\bm{b}\rangle \\&\label{eq:holder}\leq \mathbbm{E}\sup_{(\bm{a},\bm{b})\in\mathcal{C}^*}\Big[\lambda_1f(\bm{a})\cdot f^*\Big(\frac{\bm{h}_1}{\lambda_1}\Big)+\lambda_2g(\bm{b})\cdot g^*\Big(\frac{\bm{h}_2}{\lambda_2}\Big)\Big] 
        \\& \label{eq:dual_norm}
        \leq \mathbbm{E}\sup_{(\bm{a},\bm{b})\in\mathcal{C}^*} \left[\big(\lambda_1f(\bm{a})+\lambda_2g(\bm{b})\big)\cdot\Big(\sup_{\bm{a}\in\mathbb{B}_f^n}\Big\langle\frac{\bm{h}_1}{\lambda_1},\bm{a}\Big\rangle+\sup_{\bm{b}\in\mathbb{B}_g^m}\Big\langle\frac{\bm{h}_2}{\lambda_2},\bm{b}\Big\rangle\Big)\right]
        \\&\label{eq:compatible_cons}\lesssim \mathbbm{E}\left[(\lambda_1\alpha_{\bm{x}}+\lambda_2\alpha_{\bm{v}})\cdot\Big(\sup_{\bm{a}\in\mathbb{B}_f^n}\Big\langle\frac{\bm{h}_1}{\lambda_1},\bm{a}\Big\rangle+\sup_{\bm{b}\in\mathbb{B}_g^m}\Big\langle\frac{\bm{h}_2}{\lambda_2},\bm{b}\Big\rangle\Big)\right] \\&\lesssim   \Big(\alpha_{\bm{x}}+\frac{\lambda_2\alpha_{\bm{v}}}{\lambda_1}\Big)\cdot\omega(\mathbb{B}_f^n)+\Big(\alpha_{\bm{v}}+\frac{\lambda_1\alpha_{\bm{x}}}{\lambda_2}\Big)\cdot\omega(\mathbb{B}_g^m),
    \end{align}
 where in \cref{eq:holder} we apply H\"older's inequality with proper re-scaling, \cref{eq:dual_norm} follows from the definitions of $f^*(\cdot),g^*(\cdot)$ and some simple relaxation, \cref{eq:compatible_cons} is due to   \cref{eq:defi_Clam12} and $\mathcal{C}^*\subset \mathbb{S}^{n+m-1}$, and in the last inequality we use the definition of Gaussian width. 
The proof is complete.
\end{proof}

  We close this subsection with another set of geometric quantities estimates   specialized to the analysis of generative priors. 

\begin{proposition}
    \label{proadd1}
    Let $\mathcal{K}_{\bm{x}}$ and $\mathcal{K}_{\bm{v}}$ be described in   \cref{assump5}, then the following statements hold true for any $\eta_1\in (0,Lr)$, $\eta_2\in (0,L'r')$:

    \noindent
    {\rm (a)} The Kolmogorov entropy of $\mathcal{K}_{\bm{x}}$ and $\mathcal{K}_{\bm{v}}$ is bounded as follows:
    \begin{align}
        \mathscr{H}(\mathcal{K}_{\bm{x}},\eta_1)\leq k\log\frac{3Lr}{\eta_1},~~\mathscr{H}(\mathcal{K}_{\bm{v}},\eta_2)\leq k'\log\frac{3L'r'}{\eta_2}.\label{eq:bound_on_entropy}
    \end{align}
   

    \noindent 
    {\rm (b)} Let $\mathcal{K}_{\bm{x}}^-= \mathcal{K}_{\bm{x}}-\mathcal{K}_{\bm{x}}$, $\mathcal{K}_{\bm{v}}^-=\mathcal{K}_{\bm{v}}-\mathcal{K}_{\bm{v}}$, then we have \begin{align}\label{eq:secant_entropy}
        \mathscr{H}(\mathcal{K}_{\bm{x}}^-,\eta_1) \leq 2k\log\frac{6Lr}{\eta_1},~~\mathscr{H}(\mathcal{K}_{\bm{v}}^-,\eta_2)\leq 2k'\log\frac{6L'r'}{\eta_2}.
    \end{align} Moreover, for any $0<\eta<\min\{Lr,L'r'\}$ it holds that
    \begin{align}\label{eq:times_secant}
        \mathscr{H}(\mathcal{K}^-_{\bm{x}}\times \mathcal{K}_{\bm{v}}^-,\eta)\leq 2k\log \frac{12Lr}{\eta}+2k'\log\frac{12L'r'}{\eta}. 
    \end{align} 

    \noindent
    {\rm(c)} Given any $\mu\in (0,1)$, we define $\mathcal{E}$ and $\mathcal{E}^*$ as  
        \begin{gather} 
         \label{eq:gene_E}\mathcal{E}=\big\{(\bm{c},\bm{d}):\bm{c}\in \mathcal{K}_{\bm{x}}^-,~\bm{d}\in \mathcal{K}_{\bm{v}}^-,~(\|\bm{c}\|_2^2+\|\bm{d}\|_2^2)^{1/2}\geq 2\mu\big\}\\\label{eq:gene_E_star}
         \mathcal{E}^*=\big\{(\bm{c},\bm{d})/(\|\bm{c}\|^2_2+\|\bm{d}\|_2^2)^{1/2}:(\bm{c},\bm{d})\in \mathcal{E}\big\},
        \end{gather}
    then for any $0<\eta<\min\{Lr,L'r'\}$ we have \begin{align}\label{eq:bound_Estar_entropy}
        \mathscr{H}(\mathcal{E}^*,\eta)\leq 2k\log \frac{24Lr}{\mu \eta}+2k'\log\frac{24L'r'}{\mu\eta}. 
    \end{align} 
    Moreover, we have \begin{align}    \label{eq:gw_E_star}\omega(\mathcal{E}^*)\lesssim \Big(k\log\frac{Lr}{\mu}+k'\log\frac{L'r'}{\mu}\Big)^{1/2}.
    \end{align} 
\end{proposition}
 \begin{proof}We present the proofs of (a), (b) and (c) separately. 
     \subsubsection*{The Proof of (a)} 
     By the covering number of the $\ell_2$-ball \cite[Coro. 4.2.13]{vershynin2018high}, we can construct $\mathcal{G}$ as a $\big(\frac{\eta_1}{L}\big)$-net of $\mathbb{B}_2^k(r)$ with cardinality not exceeding $\big(\frac{2Lr}{\eta_1}+1\big)^k$, and hence not exceeding $\big(\frac{3Lr}{\eta_1}\big)^k$ because $\eta_1<Lr$. Since $G(\cdot)$ is $L$-Lipschitz, $G(\mathcal{G})$ is an $\eta_1$-net of $\mathcal{K}_{\bm{x}}=G(\mathbb{B}_2^k(r))$, thus implying $\mathscr{H}(\mathcal{K}_{\bm{x}},\eta_1)\leq k\log\frac{3Lr}{\eta_1}$. One can similarly derive the bound for $\mathscr{H}(\mathcal{K}_{\bm{v}},\eta_2)$.

     \subsubsection*{The Proof of (b)} From the result $(a)$ that we just proved, we can construct $\mathcal{G}_1$ as an $\big(\frac{\eta_1}{2}\big)$-net of $\mathcal{K}_{\bm{x}}$ such that $\log|\mathcal{G}_1|\leq k\log\frac{6Lr}{\eta_1}$. It is not hard to show $\mathcal{G}_1-\mathcal{G}_1$ is a $\eta_1$-net of $\mathcal{K}_{\bm{x}}^-$, and note that $\log|\mathcal{G}_1-\mathcal{G}_1|\leq\log|\mathcal{G}_1|^2\leq 2k\log\frac{6Lr}{\eta_1}$. We can similarly prove $\mathscr{H}(\mathcal{K}_{\bm{v}}^-,\eta_2)\leq 2k'\log\frac{6L'r'}{\eta_2}$, hence \cref{eq:secant_entropy} follows. To prove \cref{eq:times_secant}, it is sufficient to note the following simple fact:  if $\mathcal{G}_2$ is an $\frac{\eta}{2}$-net of $\mathcal{K}_{\bm{x}}^-$, $\mathcal{G}_3$ is an $\frac{\eta}{2}$-net of $\mathcal{K}_{\bm{v}}^-$, then $\mathcal{G}_2\times \mathcal{G}_3$ is an $\eta$-net of $\mathcal{K}_{\bm{x}}^-\times \mathcal{K}_{\bm{v}}^-$.

      \subsubsection*{The Proof of (c)} Because $\mathcal{E}\subset\mathcal{K}_{\bm{x}}^-\times \mathcal{K}_{\bm{v}}^-$, by \cref{monoto} and \cref{eq:times_secant}, we can let $\mathcal{G}_4$ be an ($\mu\eta$)-net of $\mathcal{E}$ satisfying  $\log|\mathcal{G}_4|\leq 2k\log\frac{24Lr}{\mu\eta}+2k'\log\frac{24L'r'}{\mu\eta}$. Then it suffices to show that $$\mathcal{G}_4^*=\Big\{\frac{(\bm{a},\bm{b})}{(\|\bm{a}\|^2_2+\|\bm{b}\|^2_2)^{1/2}}:(\bm{a},\bm{b})\in \mathcal{G}_4\Big\}$$ is an $\eta$-net of $\mathcal{E}^*$, which we prove in the following two dot points:\begin{itemize}[leftmargin=5ex,topsep=0.25ex]
          \item By \cref{eq:gene_E_star}, any $(\bm{a}_1,\bm{b}_1)\in\mathcal{E}^*$ can be written as ${\bm{c}_2}/{\|\bm{c}_2\|_2}=:(\bm{a}_2,\bm{b}_2)/(\|\bm{a}_2\|_2^2+\|\bm{b}_2\|_2^2)^{1/2}$ for some $\bm{c}_2:=(\bm{a}_2,\bm{b}_2)\in\mathcal{E}$. Since $\mathcal{G}_4$ is a $(\mu\eta)$-net of $\mathcal{E}$, we can pick $\bm{c}_3:=(\bm{a}_3,\bm{b}_3)\in \mathcal{G}_4$ such that $\|\bm{c}_3- \bm{c}_2\|_2\leq \mu\eta$. 
          \item Note that $\bm{c}_3/\|\bm{c}_3\|_2\in \mathcal{G}_4^*$, and we have
      \begin{align*} 
          \left\|\frac{\bm{c}_2}{\|\bm{c}_2\|_2}-\frac{\bm{c}_3}{\|\bm{c}_3\|_2}\right\|_2 &\leq \left\|\frac{\bm{c}_2}{\|\bm{c}_2\|_2}-\frac{\bm{c}_3}{\|\bm{c}_2\|_2}\right\|_2+\left\|\frac{\bm{c}_3}{\|\bm{c}_2\|_2}-\frac{\bm{c}_3}{\|\bm{c}_3\|_3}\right\|_2\\&\leq \frac{\|\bm{c}_2-\bm{c}_3\|_2}{\|\bm{c}_2\|_2}+\frac{\big|\|\bm{c}_3\|_2-\|\bm{c}_2\|_2\big|}{\|\bm{c}_2\|_2}\\&\leq \frac{2\|\bm{c}_2-\bm{c}_3\|}{\|\bm{c}_2\|_2}\leq \frac{2\eta\mu}{2\mu}=\eta,
          \end{align*}
          where the last inequality holds because $\bm{c}_2\in \mathcal{E}$ satisfies $\|\bm{c}_2\|_2 \ge 2\mu$. Therefore, $\mathcal{G}_4^*$ is an $\eta$-net of $\mathcal{E}^*$. 
      \end{itemize}
      Therefore, $\mathscr{H}(\mathcal{E}^*,\eta)\leq \log |\mathcal{G}_4|$, the   bound in \cref{eq:bound_Estar_entropy} follows.
It remains to prove \cref{eq:gw_E_star}, and our strategy is to   estimate $\omega(\mathcal{E}^*)$ from $\mathscr{H}(\mathcal{E}^*,\eta)$ via Dudley's inequality \cref{dudley}. Note that $\mathcal{E}^*\subset \mathbb{S}^{n+m-1}$, we thus have
   
          \begin{align*} 
\omega(\mathcal{E}^*)&\lesssim \int_0^\infty \sqrt{  \mathscr{H}(\mathcal{E}^*,\eta)} ~\mathrm{d}\eta = \int_0^2\sqrt{\mathscr{H}(\mathcal{E}^*,\eta)}~\mathrm{d}\eta \\&\leq \sqrt{2k}\int_0^2\sqrt{\log \frac{24Lr}{\mu\eta}}~\mathrm{d}\eta   +\sqrt{2k'}\int_0^2\sqrt{\log\frac{24L'r'}{\mu\eta}}~\mathrm{d}\eta \\&\lesssim \Big(k\log\frac{Lr}{\mu}+k'\log\frac{L'r'}{\mu}\Big)^{1/2},
          \end{align*}
  where we apply \cref{eq:bound_Estar_entropy} in the second inequality. The proof is complete.  
 \end{proof}

\section{Quantized Product Embedding}\label{qpe}
In this appendix, we develop   quantized product embedding (QPE) property for analyzing the uniform recovery of Lasso in quantized corrupted sensing. In brief,     QPE   states that the dithered uniform quantization universally preserves inner product. More precisely, with random dithering given by $\bm{\tau}\sim \mathscr{U}([-\frac{\tau}{2},\frac{\tau}{2}]^m)$, QPE ensures that  $\langle\mathcal{Q}_\delta(\bm{a}+\bm{\tau}),\bm{b}\rangle$ is close to $\langle\bm{a},\bm{b}\rangle$, with the closeness holding uniformly for all $(\bm{a},\bm{b})$ in some constraint sets of interest.

\subsection{An Observation and Local QPE}
We start from a  simple observation due to Xu and Jacques \cite[Lem. 6.1]{xu2020quantized}, which bounds the number of measurements exhibiting some ``discontinuity'' due to the quantizer $\mathcal{Q}_\delta(\cdot)$. 
\begin{lemma}\label{lem1}
    Given some $\delta,\zeta>0$ satisfying $\zeta\in (0,\frac{\delta}{2})$, and a fixed $\bm{a}=[a_i]\in\mathbb{R}^m$. Associated with $\bm{\tau}=[\tau_i]\sim \mathscr{U}[-\frac{\delta}{2},\frac{\delta}{2}]^m$, we define the random set
    \begin{equation}\label{B.1}
        \mathcal{Z}=\big\{i\in [m]: \mathcal{Q}_\delta(a_i+\tau_i+t)\text{ is discontinuous in }[-\zeta,\zeta]\big\}.
    \end{equation}
    Then, for any $t\geq 0$, we have \begin{equation}
        \label{B.2}\mathbbm{P}\Big(|\mathcal{Z}|\geq \frac{5m\zeta}{\delta}\Big)\leq \exp\Big(-\frac{m\zeta}{\delta}\Big),
    \end{equation}
    where $|\mathcal{Z}|$ is the random variable denoting the cardinality of $\mathcal{Z}$. 
\end{lemma}
\begin{proof}
First, we note that $\mathcal{Q}_\delta(a_i+\tau_i+t)$ is discontinuous in $[-\zeta,\zeta]$ if and only if $\mathcal{Q}_\delta(t)$ is discontinuous in $[a_i+\tau_i-\zeta,a_i+\tau_i+\zeta]$, and further, the latter statement is evidently equivalent to the event $$E_1=\big\{[a_i+\tau_i-\zeta,a_i+\tau_i+\zeta]\cap (\delta \mathbb{Z})=\varnothing\big\}.$$ Moreover, under our assumptions of $\tau_i\sim \mathscr{U}([-\frac{\delta}{2},\frac{\delta}{2}])$ and $\zeta\in (0,\frac{\delta}{2})$, it is not hard to see that, $\mathbbm{P}(E_1)=\frac{2\zeta}{\delta}$ holds true independent of the location of $[a_i-\zeta,a_i+\zeta]$. Therefore, $|\mathcal{Z}|\sim \mathrm{Bin}(m,\frac{2\zeta}{\delta})$ (i.e., $|\mathcal{Z}|$ follows a binomial distribution with $m$ trials and probability of success $\frac{2\zeta}{\delta}$ for each), hence it can be written as $|\mathcal{Z}|=\sum_{k=1}^mZ_k$  with i.i.d. $Z_k\sim \mathrm{Bernoulli}(\frac{2\zeta}{\delta})$. Because for any integer $q\geq 3$ we have $\sum_{k=1}^m \mathbbm{E}|Z_k|^q=\frac{2m\zeta}{\delta}$, we can invoke Bernstein's inequality   \cite[Thm. 2.10]{boucheron2013concentration} to obtain that, for any $t>0$, the event 
$\sum_{k=1}^m(Z_k-\mathbbm{E}Z_k)\le 2(\frac{m\zeta t}{\delta})^{1/2}+t$
holds with probability exceeding $1-\exp(-t)$. Setting $t=\frac{m\zeta}{\delta}$ and substituting $\mathbbm{E}|\mathcal{Z}|=\sum_{k=1}^m \mathbbm{E}Z_k = \frac{2m\zeta}{\delta}$ yield the desired claim. 
\end{proof}

We then establish the QPE with   quantized measurements generated by a fixed signal-corruption pair $(\bm{a},\bm{b})$. This ``local'' QPE will be sufficient for establishing non-uniform guarantee, and indeed, the lemma below readily follows from \cite[Lem. 2]{sun2022quantized}. We provide a proof for completeness.  
\begin{lemma}[Local QPE]\label{lem:localqpe}
 Given some bounded sets $\calA\subset\mathbb{R}^n,~\calB\subset \mathbb{R}^m,~\mathcal{E}\subset \mathbb{R}^{n+m}$ and some $\delta>0$, we fix $\ba\in \calA$ and $\bm{b}\in\calB$, and assume that the sub-Gaussian matrix $\bm{\Phi}\in\mathbb{R}^{m\times n}$, sub-Gaussian noise $\bm{\epsilon}$ and the random dither $\bm{\tau}\sim \mathscr{U}[-\frac{\delta}{2},\frac{\delta}{2}]^m$ are as described in   \cref{assump1}, 
    then for $\bm{c}\in \mathbb{R}^n$ and $\bm{d}\in \mathbb{R}^m$, and some absolute constant $C$, the event 
    \begin{align}
       \sup_{(\bm{c},\bm{d})\in \mathcal{E}}\big|\big\langle \bm{\xi}_{\bm{a},\bm{b}},\bm{\Phi c}+\sqrt{m}\bm{d}\big\rangle\big| \le C\sqrt{m}\delta \big(\omega(\mathcal{E})+u\cdot\rad(\calE)\big)
    \end{align}
    holds with probability at least $1-2\exp(-u^2)$, where we denote the quantization noise of $\bm{\Phi a}+\sqrt{m}\bm{b}$ by the shorthand 
    \begin{align}\label{eq:short_xi_a_b}
        \bm{\xi}_{\bm{a},\bm{b}} := \mathcal{Q}_\delta(\bm{\Phi a}+\sqrt{m}\bm{b}+\bm{\epsilon}+\bm{\tau})-(\bm{\Phi a}+\sqrt{m}\bm{b}+\bm{\epsilon}).
    \end{align}
\end{lemma}
\begin{proof}\label{lem:localQPE}
     We note that     it suffices to consider $\bm{\epsilon}=0$; for the general setting, since $\bm{\epsilon}$ is independent of $(\bm{\Phi},\bm{\tau})$, we can condition on $\bm{\epsilon}$ and write 
     \begin{align}\label{eq:reduction}
         \bm{\xi}_{\bm{a},\bm{b}}= \mathcal{Q}_\delta(\bm{\Phi a}+\sqrt{m}\bm{\tilde{b}}+\bm{\tau})-(\bm{\Phi a}+\sqrt{m}\bm{\tilde{b}}) 
     \end{align} with $\bm{\tilde{b}}:=\bm{b}+\frac{\bm{\epsilon}}{\sqrt{m}}$, then the result can be obtained from the case of $\bm{\epsilon}=0$.

     To get the desired bound on  $\sup_{(\bm{c},\bm{d})\in \calE}|\langle \bm{\xi}_{\bm{a},\bm{b}},\bm{\Phi c}+\sqrt{m}\bm{d}\rangle|$, we view $\langle\bm{\xi}_{\bm{a},\bm{b}},\bm{\Phi c}+\sqrt{m}\bm{d}\rangle$ as a random process indexed by $(\bm{c},\bm{d})\in \calE$ and seek to apply   \cref{pro2}. Given any $(\bm{c}_1,\bm{d}_1),(\bm{c}_2,\bm{d}_2)\in \mathcal{E}\cup\{0\}$, 
    we first verify \cref{eq:increment} as follows:
    \begin{align}
        &\big\|\langle\bm{\xi}_{\bm{a},\bm{b}},\bm{\Phi c}_1+ \sqrt{m}\bm{d}_1\rangle -\langle\bm{\xi}_{\bm{a},\bm{b}},\bm{\Phi c}_2+ \sqrt{m}\bm{d}_2\rangle \big\|_{\psi_2}\\
        &\leq \big\|\langle \bm{\xi}_{\bm{a},\bm{b}},\bm{\Phi}(\bm{c}_1-\bm{c}_2)\rangle\big\|_{\psi_2}+\big\|\langle \bm{\xi}_{\bm{a},\bm{b}},\sqrt{m}(\bm{d}_1-\bm{d}_2)\rangle\big\|_{\psi_2} \\\label{eq:sumpsi2}
        &\lesssim \sqrt{m}\delta \cdot\|\bm{c}_1-\bm{c}_2\|_2 + \sqrt{m}\delta\cdot\|\bm{d}_1-\bm{d}_2\|_2 \\&\lesssim \sqrt{m}\delta \cdot \|(\bm{c}_1,\bm{d}_1)-(\bm{c}_2,\bm{d}_2)\|_2,
    \end{align}
    where \cref{eq:sumpsi2} is due to 
    \begin{gather}
       \label{eq:explain1} \|\langle\bm{\xi}_{\bm{a},\bm{b}},\sqrt{m}(\bm{d}_1-\bm{d}_2)\rangle\|_{\psi_2}\lesssim\sqrt{m}\delta\cdot\|\bm{d}_1-\bm{d}_2\|_2\\\label{eq:explain2}
\|\langle\bm{\xi}_{\bm{a},\bm{b}},\bm{\Phi}(\bm{c}_1-\bm{c}_2)\rangle\|_{\psi_2}\lesssim\sqrt{m}\delta \cdot \|\bm{c}_1-\bm{c}_2\|_2. 
    \end{gather}    
    Note that \cref{eq:explain1}
      follows from 
$\|\bm{\xi}_{\bm{a},\bm{b}}\|_{\psi_2}=O(\delta)$ (see \cref{eq:sgxi}). To explain \cref{eq:explain2}, by 
$\|\bm{\xi}_{\bm{a},\bm{b}}\|_\infty\le\delta$ (see \cref{eq:boundedness}), {\color{black}and  with respect to the randomness of $\tau_k$   we have $\mathbbm{E}_{\tau_k}[(\bm{\xi}_{\bm{a},\bm{b}})_k]=0$ (see  \cref{eq:zero_mean}) that implies}
\begin{align}
    \label{eq:response}
    \mathbbm{E}\big((\bm{\xi}_{\bm{a},\bm{b}})_k\bm{\Phi}_k^\top (\bm{c}_1-\bm{c}_2)\big) =\mathbbm{E}\big(\mathbbm{E}_{\tau_k}[(\bm{\xi}_{\bm{a},\bm{b}})_k]\bm{\Phi}_k^\top(\bm{c}_1-\bm{c}_2)\big) = 0,
\end{align}
we can use \cref{eq:sum_psi2} to obtain
    \begin{align}
        \big\|\langle\bm{\xi}_{\bm{a},\bm{b}},\bm{\Phi}(\bm{c}_1-\bm{c}_2)\rangle\big\|_{\psi_2}^2 &\le \sum_{k=1}^m \big\|(\bm{\xi}_{\bm{a},\bm{b}})_k\bm{\Phi}_k^\top(\bm{c}_1-\bm{c}_2)\big\|_{\psi_2}^2 \lesssim m\delta^2 \|\bm{c}_1-\bm{c}_2\|_2^2. 
    \end{align}
    Now we invoke   \cref{pro2} to obtain that for some absolute constant $C$, the event \begin{align}
        \sup_{(\bm{c},\bm{d})\in \mathcal{E}}\big|\langle\bm{\xi}_{\bm{a},\bm{b}},\bm{\Phi c}+\sqrt{m}\bm{d}\rangle\big|\le C\sqrt{m}\delta \big(\omega(\mathcal{E})+u\cdot\rad(\calE)\big)
    \end{align}
    holds with probability exceeding $1-\exp(-u^2)$, as desired. 
\end{proof}

\subsection{Global QPE: The General Version}\label{app:generalqpe}
In pursuit of a uniform guarantee, the local QPE in   \cref{lem:localqpe} is no longer enough. Rather, we will need a global QPE property that holds universally for all $\bm{\xi}_{\bm{a},\bm{b}}$   with $(\bm{a},\bm{b})\in\calA\times \calB$ (see \cref{eq:short_xi_a_b} for $\bm{\xi}_{\bm{a},\bm{b}}$). To achieve the desired uniformity, we follow the standard approach of the covering argument that comprises two steps: (i) establish the desired property over the   discrete nets $\mathcal{G}_1$ and $\mathcal{G}_2$ that respectively approximate $\calA$ and $\calB$, (ii) extend the desired property from $(\bm{a},\bm{b})\in \mathcal{G}_1\times \mathcal{G}_2$ to $(\bm{a},\bm{b})\in \calA\times \calB$. While (i) can be done by applying   \cref{lem:localqpe} along with a union bound, the key challenge lies in (ii) due to the discontinuity of $\mathcal{Q}_\delta(\cdot)$ (note that the extension in (ii) typically relies on certain notion of continuity).  We overcome the difficulty by a strategy similar to \cite[Prop. 6.1]{xu2020quantized} (see \cref{eq:E1} below that characterizes certain continuity of the dithered quantizer), but with the sub-Gaussianity of $\bm{\Phi}$ we manage to provide refinements by using \cref{prop:bound_l_largest} (see more discussions in \cref{app:by}).    
\begin{theorem}[Global QPE] \label{thm:globalqpe}
      Given some bounded sets $\calA\subset\mathbb{R}^n,~\calB\subset \mathbb{R}^m,~\mathcal{E}\subset \mathbb{R}^{n+m}$ and some $\delta>0$, we further define $\calC\subset \mathbb{R}^n$ and $\calD\subset \mathbb{R}^m$ as  
      \begin{gather}\label{eq:defineC}
          \calC = \big\{\bm{c}\in \mathbb{R}^n: \exists\bm{d}\in \mathbb{R}^m,\text{ such that }(\bm{c},\bm{d})\in \calE\big\},\\\label{eq:defineD}
          \calD =\big\{\bm{d}\in \mathbb{R}^m: \exists\bm{c}\in \mathbb{R}^n,\text{ such that }(\bm{c},\bm{d})\in \calE\big\},
      \end{gather}
      and assume that the sub-Gaussian matrix $\bm{\Phi}\in\mathbb{R}^{m\times n}$, sub-Gaussian noise $\bm{\epsilon}$ and the random dither $\bm{\tau}\sim \mathscr{U}[-\frac{\delta}{2},\frac{\delta}{2}]^m$ are as described in   \cref{assump1}. Suppose that $(\zeta,\rho_1,\rho_2)$ are positive scalars satisfying 
      \begin{align}\label{eq:scaling_3cons}
          \zeta \le c_1\delta,~\rho_1\le \frac{c_1\zeta}{(\log\frac{\delta}{\zeta})^{1/2}},~\rho_2\le c_1\zeta
      \end{align}
for some sufficiently small $c_1$.
If for some large enough absolute constant $C_2>0$ we have 
\begin{align}\label{eq:qpe_sample_size}
    m\geq \frac{2\delta \cdot \scrH(\calA,\rho_1)+2\delta\cdot\scrH(\calB,\rho_2)}{\zeta} + \frac{C_2\cdot \omega^2(\calA^{(\rho_1)}_{\loc})}{\zeta^2}, 
\end{align}
then  with the quantization noise $\bm{\xi}_{\bm{a},\bm{b}}$ being given in \cref{eq:short_xi_a_b}, the event 
\begin{align}
    &\sup_{\bm{a}\in\calA}\sup_{\bm{b}\in\calB}\sup_{(\bm{c},\bm{d})\in\mathcal{E}}\big|\langle\bm{\xi}_{\bm{a},\bm{b}},\bm{\Phi c}+\sqrt{m}\bm{d}\rangle\big|\lesssim \delta \sqrt{m}\cdot \omega(\calE)+\delta\sqrt{m\hat{U}}\cdot\rad(\calD)\label{eq:thmbound1}\\&\quad+\delta\sqrt{m} \cdot \rad(\calC)\cdot\Big(\sqrt{\scrH(\calA,\rho_1)}+\sqrt{\scrH(\calB,\rho_2)}+\frac{\hat{U}}{\sqrt{m}}\Big[\log\frac{\delta}{\zeta}\Big]^{1/2}\Big).
    \label{eq:thmbound2}
\end{align}  
holds with probability exceeding  $1-12 \exp(-c_3\scrH(\calA,\rho_1)-c_3\scrH(\calB,\rho_2))$  on a single   draw of $(\bm{\Phi},\bm{\tau})$. In \cref{eq:thmbound2}, $\hat{U}$ is a quantity scaling as  
\begin{align}\label{hatU}
    \hat{U} \asymp \frac{m\zeta}{\delta}+\frac{m\rho_2^2}{\zeta^2}+\frac{\omega^2(\calA^{(\rho_1)}_{\loc})}{\zeta^2} .
\end{align}
\end{theorem}
\begin{proof}
Again, we can simply concentrate on the case where $\bm{\epsilon}=0$;  we can reduce the general case to the case with $\bm{\epsilon}=0$ by 
conditioning on $\bm{\epsilon}$ and writing \cref{eq:reduction} with $\bm{\tilde{b}}=\bm{b}+\frac{\bm{\epsilon}}{\sqrt{m}}\in \tilde{\mathcal{B}}:=\mathcal{B}+\frac{\bm{\epsilon}}{\sqrt{m}}$; since $\mathscr{H}(\mathcal{B},\cdot)=\mathscr{H}(\tilde{\mathcal{B}},\cdot)$ always holds, the result can be readily obtained from the case of $\bm{\epsilon}=0$.  
We seek to bound 
\begin{align}
    \sup_{\bm{a}\in\calA}\sup_{\bm{b}\in\calB}\sup_{(\bm{c},\bm{d})\in \mathcal{E}}\big|\langle\bm{\xi}_{\bm{a},\bm{b}},\bm{\Phi c}+\sqrt{m}\bm{d}\rangle\big|.
\end{align}
We pause to provide an outline for this most technical proof in this work (see \cref{tbl:notation} in \cref{app:table} for a table of the recurring notation):
\begin{itemize}
[leftmargin=5ex,topsep=0.25ex]
    \item In \textbf{Step 1} we bound the random process over nets of $\mathcal{A}$ and $\mathcal{B}$, which is done by applying local QPE (\cref{lem:localqpe}) and then a union bound. 
    \item In \textbf{Step 2} we build three useful events $E_1,E_2,E_3$: $E_1$ in \cref{eq:E1} bounds the number of measurements exhibiting some discontinuity due to the quantizer (those in $\mathcal{Z}_{\bm{a},\bm{b}}$ as per \cref{eq:Zab}); $E_2$ in \cref{eq:E2} bounds the number of measurements exhibiting large perturbations associated with $\bm{a}\in\mathcal{A}$ (those in $\mathcal{J}_{\bm{a}}^{\mathcal{A}}$ as per \cref{eq:bad1});  and $E_3$ in \cref{eq:E3} bounds the number of measurements with large perturbations associated with $\bm{b}\in \mathcal{B}$ (those in $\mathcal{J}_{\bm{b}}^{\mathcal{B}}$ as per \cref{eq:bad2}).   

    \item In \textbf{Step 3} we strengthen the bound from Step 1 to a uniform bound over $\mathcal{A}\times\mathcal{B}$. We provide different treatments to two classes of measurements. The first class collects the ``bad'' ones in $\mathcal{Z}_{\bm{a},\bm{b}}\cup \mathcal{J}_{\bm{a}}^{\mathcal{A}}\cup \mathcal{J}_{\bm{b}}^{\mathcal{B}}$ \cref{eq:bad_index} that only account for a few measurements (see \cref{eq:U0bound}), the second class collects the remaining   benign measurements that enjoy some nice property (see \cref{eq:simplify}). 
    \item In \textbf{Step 4} we choose parameters and use \cref{eq:scaling_3cons}  and \cref{eq:qpe_sample_size} to get the bound and probability in the theorem statement.
\end{itemize}

\subsubsection*{Step 1: Uniform Bound Over Nets} For some $\rho_1,\rho_2>0$ to be chosen, we let $\calG_1$ be the minimal $\rho_1$-net of $\calA$ with $\log|\mathcal{G}_1|=\scrH(\calA,\rho_1)$, $\calG_2$ be the minimal $\rho_2$-net of $\calB$ with $\log|\mathcal{G}_2|=\scrH(\calB,\rho_2)$. Then, for any $t_1>0$, we apply the non-uniform bound in   \cref{lem:localqpe} along with a union bound over $(\bm{a},\bm{b})\in \calG_1\times \calG_2$, yielding that the event 
\begin{align}\label{eq:net_bound}
    \sup_{\bm{a}\in \calG_1}\sup_{\bm{b}\in \calG_2}\sup_{(\bm{c},\bm{d})\in\mathcal{E}}\big|\langle\bm{\xi}_{\bm{a},\bm{b}},\bm{\Phi c}+\sqrt{m}\bm{d}\rangle\big|\le C\sqrt{m}\delta \big(\omega(\mathcal{E})+t_1\cdot\rad(\calE)\big)
\end{align}
holds with probability at least $1-2\exp(\scrH(\calA,\rho_1)+\scrH(\calB,\rho_2)-t_1^2)$.   

\subsubsection*{Step 2: Some Useful Events} Recall that we need to further strengthen \cref{eq:net_bound} from $(\bm{a},\bm{b})\in \mathcal{G}_1\times \mathcal{G}_2$ to $(\bm{a},\bm{b})\in \calA\times\calB$. To this end, we pause to build some useful high-probability events that aid the subsequent analysis. For given $(\bm{a},\bm{b})\in \calA\times \calB$ and $\zeta>0$, {\color{black}we define the random set} 
\begin{align}\label{eq:Zab}
    \calZ_{\bm{a},\bm{b}} = \big\{i\in [m]:\calQ_\delta(\bm{\Phi}_i^\top\bm{a}+\sqrt{m}\cdot b_i+\tau_i+t)\text{ is discontinuous in }t\in[-\zeta,\zeta]\big\}
\end{align}
and let its cardinality be $|\mathcal{Z}_{\bm{a},\bm{b}}|$.

\textbf{Bounding $|\mathcal{Z}_{\bm{a},\bm{b}}|$ over nets:} For some $\zeta\in (0,\frac{\delta}{2})$ to be chosen later,
given $(\bm{a},\bm{b})$ and conditioning on $\bm{\Phi}$,  we utilize   \cref{lem1} (with respect to the randomness of $\bm{\tau}$) to obtain 
\begin{align}
    \mathbbm{P}\Big(|\mathcal{Z}_{\bm{a},\bm{b}}|\ge \frac{5m\zeta}{\delta}\Big) \le \exp \Big(-\frac{m\zeta}{\delta}\Big).
\end{align}
Further taking a union bound over $(\bm{a},\bm{b})\in \calG_1\times\calG_2$, we obtain that the event 
\begin{align}\label{eq:E1}
    E_1= \Big\{\sup_{\ba\in\mathcal{G}_1}\sup_{\bm{b}\in \mathcal{G}_2}|\mathcal{Z}_{\bm{a},\bm{b}}|\le \frac{5m\zeta}{\delta}\Big\}
\end{align}
holds with probability exceeding $1-\exp\big(\scrH(\calA,\rho_1)+\scrH(\calB,\rho_2)-\frac{m\zeta}{\delta}\big)$. Note that for $i\notin \mathcal{Z}_{\bm{a},\bm{b}}$, $\calQ_{\delta}(\bm{\Phi}_i^\top\bm{a}+\sqrt{m}\cdot b_i+\tau_i+t)$ is continuous in $t\in[-\zeta,\zeta]$, which along with the definition of $\mathcal{Q}_\delta(\cdot)$ means that $\mathcal{Q}_\delta(\bm{\Phi}_i^\top\bm{a}+\sqrt{m}\cdot b_i+\tau_i+t)$ remains constant in $t\in [-\zeta,\zeta]$. In essence, $E_1$ bounds $|\mathcal{Z}_{\bm{a},\bm{b}}|$ to be no larger than $\frac{5m\zeta}{\delta}$, universally over the nets  $(\bm{a},\bm{b})\in\mathcal{G}_1\times \mathcal{G}_2$, and thus indicates certain continuity of the quantization in the following sense: for measurement in $[m]\setminus \mathcal{Z}_{\bm{a},\bm{b}}$ (that is the majority under small enough $\zeta$) a perturbation smaller than $\zeta$ will not change the quantized value.

   For any $(\bm{a},\bm{b})\in \mathcal{A}\times\mathcal{B}$ there exists   $(\bm{a}',\bm{b}')\in \mathcal{G}_1\times \mathcal{G}_2$ satisfying $\|\bm{a}-\bm{a}'\|_2 \le \rho_1$ and $\|\bm{b}-\bm{b}'\|_2\le\rho_2$, and we suppose that such pair of $(\bm{a}',\bm{b}')$ has been chosen for every $(\bm{a},\bm{b})\in\mathcal{A}\times\mathcal{B}$ ($\bm{a}'$ and $\bm{b}'$ evidently depend on $(\bm{a},\bm{b})$, but we omit such dependence to avoid cumbersome notation). To get uniform bound over $\mathcal{A}\times\mathcal{B}$ from a bound over nets (see \cref{eq:net_bound}), we will need to compare $\mathcal{Q}_\delta(\bm{\Phi}_i^\top \bm{a}+\sqrt{m}\cdot b_i+\tau_i)$ with $\mathcal{Q}_\delta(\bm{\Phi}_i^\top \bm{a}'+\sqrt{m}\cdot b_i'+\tau_i)$, and we note that the former can be written as 
\begin{align}
    \mathcal{Q}_\delta \big(\bm{\Phi}_i^\top\bm{a}'+\sqrt{m}\cdot b_i'+ \tau_i+ \bm{\Phi}_i^\top(\bm{a}-\bm{a}')+\sqrt{m}\cdot(b_i-b_i')\big),
\end{align}
thus we need to study how the perturbations $\bm{\Phi}_i^\top(\bm{a}-\bm{a}')$ and $\sqrt{m}\cdot(b_i-b_i')$ affect the quantized value. As explained above, on the event $E_1$, for any $i\notin \mathcal{Z}_{\bm{a}',\bm{b}'}$, if $|\bm{\Phi}_i^\top (\bm{a}-\bm{a}')+\sqrt{m}\cdot (b_i-b_i')|\le\zeta$, then it holds that
\begin{align}
    \mathcal{Q}_\delta(\bm{\Phi}_i^\top \bm{a}+\sqrt{m}\cdot b_i+\tau_i)=\mathcal{Q}_\delta(\bm{\Phi}_i^\top \bm{a}'+\sqrt{m}\cdot b_i'+\tau_i). \label{eq:observe1}
\end{align}
In order to utilize this property, we proceed to bound the number of large perturbations $\bm{\Phi}_i^\top(\bm{a}-\bm{a}')$ associated with $\bm{a}$, and similarly large perturbations $\sqrt{m}\cdot(b_i-b_i')$ associated with $\bm{b}$. More precisely, given $\bm{a}\in\mathcal{A}$ and $\bm{b}\in\mathcal{B}$, we define the index sets for large perturbations as
\begin{align}\label{eq:bad1}
    \mathcal{J}^{\mathcal{A}}_{\bm{a}}& = \Big\{i\in [m]: |\bm{\Phi}_i^\top(\bm{a}-\bm{a}')|\ge\frac{\zeta}{2}\Big\},\\\label{eq:bad2}
    \mathcal{J}^{\mathcal{B}}_{\bm{b}}& = \Big\{i\in [m]: \sqrt{m}\cdot\big|b_i-b_i'\big|\ge\frac{\zeta}{2}\Big\},
\end{align}
and denote their cardinalities by $|\mathcal{J}_{\bm{a}}^{\mathcal{A}}|$ and $|\mathcal{J}_{\bm{b}}^{\mathcal{B}}|$, respectively.

\textbf{Bounding $|\mathcal{J}^{\mathcal{A}}_{\bm{a}}|$ over $\calA$:} For some $l\in[1,m]$ to be chosen, we apply   \cref{prop:bound_l_largest} with $\mathcal{T}=\mathcal{A}_{\loc}^{(\rho_1)}:=(\mathcal{A}-\mathcal{A})\cap \mathbbm{B}_2^n(\rho)$ to obtain that the event 
\begin{align}
    \label{eq:bound_ell_large}
    \sup_{\bm{v}\in \mathcal{A}^{(\rho_1)}_{\loc}}\max_{\substack{I\subset [m]\\ |I| \le l}} \Big(\frac{1}{l}\sum_{i\in I}|\langle \bm{\Phi}_i,\bm{v}\rangle|^2\Big)^{1/2}\le \frac{\zeta}{3}
\end{align}
holds with probability at least $1-2\exp(C_1 l\log\frac{em}{l})$, as long as 
\begin{align}\label{eq:con_for_bound_l}
    \frac{\omega(\mathcal{A}^{(\rho_1)}_{\loc})}{\sqrt{l}} + \rho_1 \sqrt{\log\frac{em}{l}}\le c_2\zeta
\end{align}
holds with sufficiently small $c_2$,
as dictated by the right-hand side of \cref{eq:bound_l_large}.\footnote{The right-hand side of \cref{eq:bound_l_large} dictates that, to ensure \cref{eq:bound_ell_large} holding with high probability, it suffices to have $\frac{\omega(\mathcal{A}_{\loc}^{(\rho_1)})}{l^{1/2}}+\rho_1(\log\frac{em}{l})^{1/2}\le c_2\zeta$ with sufficiently small $C_2$.} We suppose that we are on the event \cref{eq:bound_ell_large} and will choose $(\rho_1,l,\zeta)$ satisfying \cref{eq:con_for_bound_l}  later. Then, given $\bm{a}\in\mathcal{A}$ and the corresponding $\bm{a}'\in \mathcal{G}_1$ satisfying $\|\bm{a}-\bm{a}'\|_2\le \rho_1$, we have
  $\bm{a}-\bm{a}'\in \mathcal{A}^{(\rho_1)}_{\loc}$ and thus \cref{eq:bound_ell_large} yields 
  \begin{align}
      \max_{\substack{I\subset [m]\\ |I| \le l}} \Big(\frac{1}{l}\sum_{i\in I}\big|\langle \bm{\Phi}_i,\bm{a}-\bm{a}'\rangle\big|^2\Big)^{1/2}\le \frac{\zeta}{3}. \label{eq:bound_l_special}
  \end{align}
  Observe that the left-hand side of \cref{eq:bound_l_special} is an upper bound on the $l$-th largest elements in $\{|\bm{\Phi}_i^\top(\bm{a}-\bm{a}')|:i=1,...,m\}$, we thus obtain $|\mathcal{J}_{\bm{a}}^{\mathcal{A}}|\le l$. Since this argument applies to any $\bm{a}\in \mathcal{A}$ (and the corresponding $\bm{a}'\in\mathcal{G}_1$), \cref{eq:bound_ell_large} implies the event 
  \begin{align}\label{eq:E2}
      E_2 = \Big\{\sup_{\bm{a}\in \mathcal{A}}|\mathcal{J}_{\bm{a}}^{\mathcal{A}}| \le l\Big\}. 
  \end{align}

\textbf{Bounding $|\mathcal{J}^{\mathcal{B}}_{\bm{b}}|$ over $\calB$:} We consider $\bm{b}\in\mathcal{B}$ and the corresponding $\bm{b}'\in \mathcal{G}_2$ satisfying $\|\bm{b}-\bm{b}'\|_2\le \rho_2$. Without the modulation of $\bm{\Phi}$, we will have less available information on  $\{\sqrt{m}\cdot|b_i-b_i'|:i\in[m]\}$ but merely $\|\bm{b}-\bm{b}'\|_2\le \rho_2$. To still get a bound on $|\mathcal{J}_{\bm{b}}^{\mathcal{B}}|$, we observe that 
\begin{align}
    \rho_2^2 \ge \|\bm{b}-\bm{b}'\|_2^2 \ge |\mathcal{J}_{\bm{b}}^{\mathcal{B}}|\cdot \Big(\frac{\zeta}{2\sqrt{m}}\Big)^2 = \frac{\zeta^2|\mathcal{J}_{\bm{b}}^{\mathcal{B}}|}{4m},
\end{align}
which implies $|\mathcal{J}_{\bm{b}}^{\mathcal{B}}|\le \frac{4m\rho_2^2}{\zeta^2}$. Note that this holds deterministically for all $\bm{b}\in\mathcal{B}$ (and the corresponding $\bm{b}'\in \mathcal{G}_2$), and hence the event \begin{align}\label{eq:E3}
    E_3=\Big\{\sup_{\bm{b}\in\mathcal{B}}|\mathcal{J}_{\bm{b}}^{\mathcal{B}}|\le \frac{4m\rho_2^2}{\zeta^2}\Big\}
\end{align}   
holds deterministically.

\subsubsection*{Step 3: Extension to the Whole Sets}
Equipped with the high-probability events $E_1$, $E_2$ and $E_3$, we are in a position to strengthen the bound over $ \mathcal{G}_1\times\mathcal{G}_2$ (see \cref{eq:net_bound}) to $(\bm{a},\bm{b})\in\mathcal{A}\times\mathcal{B}$. For any $(\bm{a},\bm{b})\in \mathcal{A}\times\mathcal{B}$, recall that we have chosen $(\bm{a}',\bm{b}')\in \mathcal{G}_1\times\mathcal{G}_2$ satisfying $\|\bm{a}-\bm{a}'\|_2\le\rho_1$ and $\|\bm{b}-\bm{b}'\|_2\le \rho_2$, and we begin with 
\begin{align}
    &\sup_{(\bm{c},\bm{d})\in\mathcal{E}}\big|\langle\bm{\xi}_{\bm{a},\bm{b}},\bm{\Phi c}+\sqrt{m}\bm{d}\rangle\big|\le \sup_{(\bm{c},\bm{d})\in\mathcal{E}}\big|\langle \bm{\xi}_{\bm{a},\bm{b}}-\bm{\xi}_{\bm{a}',\bm{b}'},\bm{\Phi c}+\sqrt{m}\bm{d}\rangle\big|+ \sup_{(\bm{c},\bm{d})\in \mathcal{E}}\big|\langle \bm{\xi}_{\bm{a}',\bm{b}'},\bm{\Phi c}+\sqrt{m}\bm{d}\rangle\big|\\
    \label{eq:use_net_bound}&\quad\quad\le \underbrace{\sup_{\bm{c}\in \mathcal{C}}\big|\langle \bm{\xi}_{\bm{a},\bm{b}}-\bm{\xi}_{\bm{a}',\bm{b}'},\bm{\Phi c}\rangle\big|}_{:=I_1}+\underbrace{\sqrt{m}\cdot\sup_{\bm{d}\in \mathcal{D}}\big|\langle \bm{\xi}_{\bm{a},\bm{b}}-\bm{\xi}_{\bm{a}',\bm{b}'},\bm{d}\rangle\big|}_{:=I_2}+  C\sqrt{m}\delta\cdot\big(\omega(\mathcal{E})+t_1\cdot\rad(\calE)\big),
\end{align}
where in \cref{eq:use_net_bound} we apply \cref{eq:net_bound}. We will need to separately bound $I_1$ and $I_2$, while we discuss two kinds of measurements before proceeding.

\textbf{Bad Measurements:} We define for any $(\bm{a},\bm{b})\in \mathcal{A}\times \mathcal{B}$ (and the corresponding $(\bm{a}',\bm{b}')\in\mathcal{G}_1\times\mathcal{G}_2$)
the index set \begin{align}
    \label{eq:bad_index}
    \mathcal{U}_{\bm{a},\bm{b}} = \mathcal{Z}_{\bm{a}',\bm{b}'}\cup \mathcal{J}_{\bm{a}}^{\mathcal{A}}\cup \mathcal{J}_{\bm{b}}^{\mathcal{B}},
\end{align}
which collect the ``bad'' measurements that either lack certain continuity regarding the quantizer (i.e., measurements in $\mathcal{Z}_{\bm{a},\bm{b}}$) or present large perturbations regarding $\bm{a}$ or $\bm{b}$  (i.e., measurements in $\mathcal{J}_{\bm{a}}^{\mathcal{A}}\cup \mathcal{J}_{\bm{b}}^{\mathcal{B}}$). Fortunately, 
the ``bad'' measurements are not that many, since on the events $E_1$ \cref{eq:E1}, $E_2$ \cref{eq:E2}, $E_3$ \cref{eq:E3} we have 
\begin{align}
    \sup_{\bm{a}\in\mathcal{A}}\sup_{\bm{b}\in \mathcal{B}}|\mathcal{U}_{\bm{a},\bm{b}}|&\le \sup_{\bm{a}'\in \mathcal{G}_1}\sup_{\bm{b}'\in \mathcal{G}_2}|\mathcal{Z}_{\bm{a}',\bm{b}'}| + \sup_{\bm{a}\in\mathcal{A}}|\mathcal{J}_{\bm{a}}^{\mathcal{A}}| + \sup_{\bm{b}\in \mathcal{B}}|\mathcal{J}_{\bm{b}}^{\mathcal{B}}|\\
    &\le \frac{5m\zeta}{\delta}+\frac{4m\rho_2^2}{\zeta^2}+l:=U_0. \label{eq:U0bound}
\end{align}
By \cref{eq:scaling_3cons} $\frac{5m\zeta}{\delta}+\frac{4m\rho_2^2}{\zeta^2}\lesssim m$ with small enough implied constant, and we will choose $\ell$ in \cref{eq:choosel} below satisfying $\ell\lesssim m$ for small enough constant (see \cref{eq:U0_value} for the value of $U_0$ after choosing $l$). Thus, by rounding that has minimal impact on our analysis, we can assume that $U_0$ is an integer in $[1,m]$. 
To further control the impact of bad measurements in the worst case, we note the following deterministic bound that holds for any $i\in [m]$:
\begin{align}
    |(\bm{\xi}_{\bm{a},\bm{b}})_{i}-(\bm{\xi}_{\bm{a}',\bm{b}'})_{i}|&=\big|[\mathcal{Q}_\delta(\bm{\Phi}_i^\top \bm{a}+\sqrt{m}b_i + \tau_i) - (\bm{\Phi}_i^\top \bm{a}+\sqrt{m}b_i)]\\
    &\quad-[\mathcal{Q}_\delta(\bm{\Phi}_i^\top \bm{a}'+\sqrt{m}b'_i + \tau_i) - (\bm{\Phi}_i^\top \bm{a}'+\sqrt{m}b_i')]\big|\\
    &=\big|[\mathcal{Q}_\delta(\bm{\Phi}_i^\top \bm{a}+\sqrt{m}b_i + \tau_i) - (\bm{\Phi}_i^\top \bm{a}+\sqrt{m}b_i+\tau_i)]\\
    &\quad-[\mathcal{Q}_\delta(\bm{\Phi}_i^\top \bm{a}'+\sqrt{m}b'_i + \tau_i) - (\bm{\Phi}_i^\top \bm{a}'+\sqrt{m}b_i'+\tau_i)]\big|\\
    & \le 2\cdot \sup_{x\in \mathbb{R}}|\mathcal{Q}_\delta(a)-a|\le \delta. \label{eq:deter_bound}
\end{align}

\textbf{Benign Measurements:} By contrast, measurements not in $\mathcal{U}_{\bm{a},\bm{b}}$ enjoy some nice property; In particular, for $i\notin \mathcal{U}_{\bm{a},\bm{b}}$ we have $i\notin \mathcal{J}_{\bm{a}}^{\mathcal{A}}\cup \mathcal{J}_{\bm{b}}^{\mathcal{B}}$ and hence 
\begin{align}
    &|\bm{\Phi}_i^\top(\bm{a}-\bm{a}')+\sqrt{m}\cdot(b_i-b_i')|\\
    &\le |\bm{\Phi}_i^\top(\bm{a}-\bm{a}')|+|\sqrt{m}\cdot(b_i-b_i')| \le \frac{\zeta}{2}+\frac{\zeta}{2}=\zeta,
\end{align}
thus  \cref{eq:observe1}  holds true, which allows us to simplify the $i$-th entry of $\bm{\xi}_{\bm{a},\bm{b}}-\bm{\xi}_{\bm{a}',\bm{b}'}$ as
\begin{align}
    (\bm{\xi}_{\bm{a},\bm{b}})_i-(\bm{\xi}_{\bm{a}',\bm{b}'})_i& = \mathcal{Q}_\delta(\bm{\Phi}_i^\top \bm{a}+\sqrt{m}b_i +\tau_i) - (\bm{\Phi}_i^\top\bm{a}+\sqrt{m}b_i)\\& \quad\quad-\mathcal{Q}_\delta(\bm{\Phi}_i^\top \bm{a}'+\sqrt{m}b'_i +\tau_i) +(\bm{\Phi}_i^\top\bm{a}'+\sqrt{m}b'_i)\\
    & = \bm{\Phi}_i^\top(\bm{a}'-\bm{a})+\sqrt{m}\cdot(b_i'-b_i). \label{eq:simplify}
\end{align}

\textbf{Decomposition:} According to $\mathcal{U}_{\bm{a},\bm{b}}$ we can always decompose $\bm{\xi}_{\bm{a},\bm{b}}-\bm{\xi}_{\bm{a}',\bm{b}'}$ into 
\begin{align}\label{eq:decomposition}
    \bm{\xi}_{\bm{a},\bm{b}}-\bm{\xi}_{\bm{a}',\bm{b}'} = \bm{h}_{\bm{a},\bm{b}}^{(1)}+\bm{h}_{\bm{a},\bm{b}}^{(2)},
\end{align}
with $\bm{h}_{\bm{a},\bm{b}}^{(1)}$ and $\bm{h}_{\bm{a},\bm{b}}^{(2)}$ respectively  accommodating the entries in $\mathcal{U}_{\bm{a},\bm{b}}$ and $[m]\setminus \mathcal{U}_{\bm{a},\bm{b}}$, i.e., \begin{align}
    (\bm{h}^{(1)}_{\bm{a},\bm{b}})_{i} =( \bm{\xi}_{\bm{a},\bm{b}})_{i} - ( \bm{\xi}_{\bm{a}',\bm{b}'})_{i},~ (\bm{h}^{(2)}_{\bm{a},\bm{b}})_{i}=0 
\end{align}
when $i\in \mathcal{U}_{\bm{a},\bm{b}}$;  otherwise, $ (\bm{h}^{(1)}_{\bm{a},\bm{b}})_{i} =0$ and  
\begin{align}\label{eq:simplify_h2}
    (\bm{h}^{(2)}_{\bm{a},\bm{b}})_{i}=( \bm{\xi}_{\bm{a},\bm{b}})_{i} - ( \bm{\xi}_{\bm{a}',\bm{b}'})_{i} = \bm{\Phi}_i^\top(\bm{a}'-\bm{a})+\sqrt{m}\cdot(b_i'-b_i)
\end{align}
when $i\notin \mathcal{U}_{\bm{a},\bm{b}}$, with the second equality following from \cref{eq:simplify}.

\textbf{Bounding $I_1$:} By substituting \cref{eq:decomposition} we can start with 
\begin{align}\label{eq:decompose_I1}
    I_1 \le \sup_{\bm{c}\in\calC}\big|\langle \bm{h}^{(1)}_{\bm{a},\bm{b}},\bm{\Phi c}\rangle\big|+\sup_{\bm{c}\in\calC}\big|\langle \bm{h}^{(2)}_{\bm{a},\bm{b}},\bm{\Phi c}\rangle\big|.
\end{align}
Recall from \cref{eq:U0bound} and \cref{eq:deter_bound} that  $\|\bm{h}^{(1)}_{\bm{a},\bm{b}}\|_0\le U_0$ and $\|\bm{h}^{(1)}_{\bm{a},\bm{b}}\|_\infty\le\delta$ hold uniformly for all $(\bm{a},\bm{b})\in\mathcal{A}\times\mathcal{B}$. Thus, to bound the first term in \cref{eq:decompose_I1}, we can  restrict our attention to entries in the support of $\bm{h}_{\bm{a},\bm{b}}^{(1)}$ and apply Cauchy-Schwarz inequality to obtain   
\begin{align}
    \sup_{\bm{c}\in \calC}\big|\langle \bm{h}_{\bm{a},\bm{b}}^{(1)},\bm{\Phi c}\rangle\big|& \le \|\bm{h}^{(1)}_{\bm{a},\bm{b}}\|_2 \cdot \sup_{\bm{c}\in\calC}\Big(\sum_{i\in \supp(\bm{h}^{(1)}_{\bm{a},\bm{b}})}\big|\bm{\Phi}_i^\top \bm{c}\big|^2\Big)^{1/2}\\
    &\le \delta \sqrt{U_0}\cdot \sup_{\bm{c}\in\calC}\max_{\substack{I\subset [m]\\|I|\le U_0}}\Big(\sum_{i\in I}|\bm{\Phi}_i^\top \bm{c}|^2\Big)^{1/2}\\
    &\label{eq:bound_U0_large}\le C_3 \delta\sqrt{U_0}\cdot \Big( \omega(\calC)+\rad(\calC)\sqrt{U_0\log\frac{em}{U_0}}\Big),
\end{align}
where  \cref{eq:bound_U0_large} holds with probability at least $1-2\exp(-C_4 U_0 \log\frac{em}{U_0})$ due to a straightforward application of \cref{prop:bound_l_largest}. Next, we seek to bound the second term in \cref{eq:decompose_I1}. By  \cref{eq:simplify_h2} we can proceed as \begin{align}
    \sup_{\bm{c}\in\mathcal{C}}\big|\langle\bm{h}^{(2)}_{\bm{a},\bm{b}},\bm{\Phi c}\rangle\big|&\le \big\|\bm{\Phi}(\bm{a}'-\bm{a})+\sqrt{m}(\bm{b}'-\bm{b})\big\|_2\cdot\sup_{\bm{c}\in\calC}\|\bm{\Phi c}\|_2 \\
    &\le \Big(\sup_{\bm{v}\in \mathcal{A}^{(\rho_1)}_{\loc}}\|\bm{\Phi v}\|_2+\sqrt{m}\rho_2\Big)\cdot \sup_{\bm{c}\in\calC}\|\bm{\Phi c}\|_2.\label{eq:bound_h2_Phic}
\end{align}
Now we apply \cref{prop:bound_l_largest} to achieve the following two bounds (for some absolute constants $C_5,C_6$):\footnote{Alternatively, one can achieve this by using matrix deviation inequality; see \cref{pro1} with $\mathcal{T}=\mathcal{T}_0\times\{0\}$ for some $\mathcal{T}_0\subset \mathbb{R}^n$.}
\begin{itemize}
  [leftmargin=5ex,topsep=0.25ex]

    \item  
    \cref{prop:bound_l_largest} with $\mathcal{T}=\mathcal{A}^{(\rho_1)}_{\loc}$ and $l=m$ yields that the event \begin{align}
      \label{eq:bound_Phiv}  \sup_{\bm{v}\in \mathcal{A}^{(\rho_1)}_{\loc}}\|\bm{\Phi v}\|_2 &\le C_5 \Big(\omega\big(\mathcal{A}^{(\rho_1)}_{\loc}\big)+\sqrt{m}\cdot\rho_1\Big)  
    \end{align}
    holds with probability exceeding $1-2\exp(-C_6m)$. 
    
    \item   \cref{prop:bound_l_largest} with $\mathcal{T}= \calC$ and $l=m$ yields that the event 
    \begin{align}\label{eq:bound_Phic}
        \sup_{\bm{c}\in\calC}\|\bm{\Phi c}\|_2 \le C_5 \Big(\omega(\calC)+\sqrt{m}\cdot \rad(\calC)\Big)
    \end{align}
    holds with probability exceeding $1-2\exp(-C_6m)$. 
\end{itemize}
Substituting \cref{eq:bound_Phiv} and \cref{eq:bound_Phic} into \cref{eq:bound_h2_Phic} we obtain 
\begin{align}
    \sup_{\bm{c}\in\mathcal{C}}\big|\langle\bm{h}^{(2)}_{\bm{a},\bm{b}},\bm{\Phi c}\rangle\big|&\le C_7 \big(\omega(\calC)+\sqrt{m}\cdot\rad(\calC)\big)\Big(\omega\big(\mathcal{A}^{(\rho_1)}_{\loc}\big)+\sqrt{m}\rho_1+\sqrt{m}\rho_2\Big). \label{eq:h2_final_bound}
\end{align}
Note that all arguments in bounding $I_1$ hold universally for all $(\bm{a},\bm{b})\in\calA\times\calB$. Thus,
combining \cref{eq:decompose_I1}, \cref{eq:bound_U0_large} and \cref{eq:h2_final_bound} immediately yields 
\begin{align}
    \nonumber \sup_{\bm{a}\in\mathcal{A}}\sup_{\bm{b}\in\mathcal{B}} I_1&\le C_3 \delta\sqrt{U_0}\cdot \Big( \omega(\calC)+\rad(\calC)\sqrt{U_0\log\frac{em}{U_0}}\Big)\\
    &\quad+C_7 \big(\omega(\calC)+\sqrt{m}\cdot\rad(\calC)\big)\Big(\omega\big(\mathcal{A}^{(\rho_1)}_{\loc}\big)+\sqrt{m}\rho_1+\sqrt{m}\rho_2\Big),\label{eq:final_boundI1}
\end{align}
where $U_0$ is given in \cref{eq:U0bound}.

\textbf{Bounding $I_2$:} By substituting \cref{eq:decomposition}, we proceed as 
\begin{align}
    I_2 & \le \sqrt{m}\cdot\rad(\calD)\cdot \|\bm{\xi}_{\bm{a},\bm{b}}-\bm{\xi}_{\bm{a}',\bm{b}'}\|_2\\&\le   \sqrt{m}\cdot\rad(\calD)\cdot\big(\|\bm{h}^{(1)}_{\bm{a},\bm{b}}\|_2+ \|\bm{h}^{(2)}_{\bm{a},\bm{b}}\|_2  \big)\label{eq:I2bound}
\end{align}
Recall from \cref{eq:U0bound} and \cref{eq:deter_bound} that  $\|\bm{h}^{(1)}_{\bm{a},\bm{b}}\|_0\le U_0$ and $\|\bm{h}^{(1)}_{\bm{a},\bm{b}}\|_\infty\le\delta$ hold uniformly for all $(\bm{a},\bm{b})\in\mathcal{A}\times\mathcal{B}$, and hence we have $ \|\bm{h}^{(1)}_{\bm{a},\bm{b}}\|_2 \le \delta\sqrt{U_0}$.
Then,  \cref{eq:simplify_h2} gives 
\begin{align}
    \|\bm{h}^{(2)}_{\bm{a},\bm{b}}\|_2 &\le \|\bm{\Phi}(\bm{a}-\bm{a}')\|_2 + \sqrt{m}\|\bm{b}-\bm{b}'\|_2\\
    &\le \sup_{\bm{v}\in \mathcal{A}^{(\rho_1)}_{\loc}}\|\bm{\Phi v}\|_2 + \sqrt{m}\cdot \rho_2\\
    &\le  C_5 \Big(\omega\big(\mathcal{A}^{(\rho_1)}_{\loc}\big)+\sqrt{m}\rho_1+\sqrt{m}\rho_2\Big) ,
    \label{eq:h2bound}
\end{align}
where we use \cref{eq:bound_Phiv} in \cref{eq:h2bound}. We note that all arguments in bounding $I_2$ hold universally for all $(\bm{a},\bm{b})\in\calA\times\calB$. 
Substituting the bounds $\|\bm{h}^{(1)}_{\bm{a},\bm{b}}\|_2 \le \delta\sqrt{U_0}$ and \cref{eq:h2bound} into \cref{eq:I2bound}, we thus obtain 
\begin{align}\label{eq:final_boundI2}
    \sup_{\bm{a}\in\calA}\sup_{\bm{b}\in\calB}I_2 \le \sqrt{m}\cdot\rad(\calD)\Big(\delta\sqrt{U_0}+C_5 \Big(\omega\big(\mathcal{A}^{(\rho_1)}_{\loc}\big)+\sqrt{m}\rho_1+\sqrt{m}\rho_2\Big)\Big) 
\end{align}
where $U_0$ is given in \cref{eq:U0bound}.

\subsubsection*{Step 4: Combining and Parameters Selection} 
We are in a position to combine everything together. Taking supremum  over $(\bm{a},\bm{b})\in \calA\times\calB$ in \cref{eq:use_net_bound}  and then  
substituting \cref{eq:final_boundI1} and \cref{eq:final_boundI2},  we obtain that as long as $\rho_1>0,l\in [1,m],\zeta\in (0,\frac{\delta}{2})$ are chosen such that \cref{eq:con_for_bound_l} holds, then the event 
\begin{align}\label{eq:qpebound1}
&\sup_{\bm{a}\in\calA}\sup_{\bm{b}\in\calB}\sup_{(\bm{c},\bm{d})\in\mathcal{E}}\big|\langle\bm{\xi}_{\bm{a},\bm{b}},\bm{\Phi c}+\sqrt{m}\bm{d}\rangle\big|\\\label{eq:qpebound2}
&\lesssim \sqrt{m}\delta\cdot\big(\omega(\mathcal{E})+t_1\cdot\rad(\calE)\big) + \delta\sqrt{U_0}\Big(\omega(\calC)+\sqrt{U_0\log\frac{em}{U_0}}\cdot\rad(\calC)+\sqrt{m}\cdot\rad(\calD)\Big)\\\label{eq:qpebound3}
&\quad + \Big(\omega\big(\mathcal{A}_{\loc}^{(\rho_1)}\big)+\sqrt{m}\rho_1+\sqrt{m}\rho_2\Big)\big(\omega(\calC)+\sqrt{m}\cdot\rad(\calC)+\sqrt{m}\cdot \rad(\calD)\big)
\end{align}
 holds with probability exceeding \begin{align}
     \label{eq:probterm1}
     1&- 2\exp\big(\mathscr{H}(\calA,\rho_1)+\mathscr{H}(\calB,\rho_2)-t_1^2\big) 
      - \exp \Big(\mathscr{H}(\calA,\rho_1)+\mathscr{H}(\calB,\rho_2)-\frac{m\zeta}{\delta}\Big)\\\label{eq:probterm3}
     & - 2 \exp\big(-C_1l\log\frac{em}{l}\big) - 2\exp \big(-C_4 U_0\log\frac{em}{U_0}\big) -4\exp(-C_6m),
 \end{align}
 where the terms in \cref{eq:probterm1} stem from \cref{eq:net_bound} and the event $E_1$ \cref{eq:E1}, the terms in \cref{eq:probterm3} are from   \cref{eq:bound_ell_large}, \cref{eq:bound_U0_large}, and \cref{eq:bound_Phiv}--\cref{eq:bound_Phic}.

 \textbf{Choosing Parameters:} We specify the parameter whose (near) optimal choice is clear at this stage, while we still leave other parameters generic since their optimal values may depend on $(\mathcal{A},\mathcal{B},\mathcal{E})$. Specifically, we set (we suppose that $l$ below is chosen as an integer in $[1,m]$ without loss of generality, since we can we just round otherwise)  
 \begin{gather}\label{eq:chooset1}
     t_1 = 2 \sqrt{\mathscr{H}(\calA,\rho_1)}+2\sqrt{\mathscr{H}(\calB,\rho_2)},\\
     l = \frac{m\zeta}{\delta} + \frac{C_2\cdot  \omega^2(\calA^{(\rho_1)}_{\loc})}{\zeta^2}\label{eq:choosel}
 \end{gather}
 with sufficiently large $C_2$ such that $\frac{\omega(\mathcal{A}^{(\rho_1)}_{\loc})}{\sqrt{l}}\le \frac{c_2\zeta}{2}$. We show that our choice \cref{eq:choosel} satisfies \cref{eq:con_for_bound_l} that is needed to ensure \cref{eq:bound_ell_large}. Specifically, $l \geq \frac{C_2\cdot \omega^2(\mathcal{A}^{(\rho_1)}_{\loc})}{\zeta^2}$  implies $\frac{\omega(\mathcal{A}^{(\rho_1)}_{\loc})}{\sqrt{l}}\lesssim\zeta$ with small enough implied constant, and $l \ge \frac{m\zeta}{\delta}$ along with $\zeta\lesssim\delta$ from \cref{eq:scaling_3cons} implies $\rho_1 \sqrt{\log\frac{em}{l}}\le \rho_1\sqrt{\log\frac{e\delta}{\zeta}}\lesssim \zeta$ with small enough implied constant.  
 We recall the value of $U_0$ given in \cref{eq:U0bound}, which together with \cref{eq:choosel} reads as  
 \begin{align}\label{eq:U0_value}
     U_0 = \frac{6m\zeta}{\delta}+ \frac{4m\rho_2^2}{\zeta^2}+ \frac{C_2\cdot \omega^2(\calA^{(\rho_1)}_{\loc})}{\zeta^2}. 
 \end{align}

 \textbf{Simplifying \cref{eq:probterm1}--\cref{eq:probterm3}:} Under \cref{eq:chooset1} and \cref{eq:qpe_sample_size} that implies $\frac{m\zeta}{\delta}\ge 2\mathscr{H}(\mathcal{A},\rho_1)+2\mathscr{H}(\mathcal{B},\rho_2)$, we can relax the probability terms in \cref{eq:probterm1} to 
 \begin{align}
 2\exp\big(\scrH(\calA,\rho_1)+\scrH(\calB,\rho_2)-&t_1^2\big)+\exp\Big(\scrH(\calA,\rho_1)+\scrH(\calB,\rho_2)-\frac{m\zeta}{\delta}\Big)\\
 &\le3\exp\big(-\scrH(\calA,\rho_1)-\scrH(\calB,\rho_2)\big).
 \end{align}
   Besides, since   \cref{eq:U0bound} and \cref{eq:choosel}  give $U_0\ge l\ge \frac{m\zeta}{\delta}$, so we have
   \begin{align}
       2\exp(-C_1l\log\frac{em}{l})+2\exp(-C_4U_0\log\frac{em}{U_0})\le 4\exp(-c_8 \scrH(\calA,\rho_1)-c_8\scrH(\calB,\rho_2))
   \end{align}
for some absolute constant $c_8>0$. Moreover, \cref{eq:qpe_sample_size} and $\zeta \le \frac{\delta}{2}$ imply $m\ge 2\scrH(\calA,\rho_1)+2\scrH(\calB,\rho_2)$, and so we have
\begin{align}
4\exp(-C_6m)\le     4\exp(-c_8 \scrH(\calA,\rho_1)-c_8\scrH(\calB,\rho_2))
\end{align}
 provided that $c_8$ is  chosen sufficiently small. Overall, from \cref{eq:probterm1}--\cref{eq:probterm3} we can promise that \cref{eq:qpebound1}--\cref{eq:qpebound3} holds with   probability exceeding \begin{align}
     1-12 \exp(-c_9\scrH(\calA,\rho_1)-c_9\scrH(\calB,\rho_2))\label{eq:simplify_prob}
 \end{align} with some $c_9>0$.

\textbf{Simplifying \cref{eq:qpebound1}--\cref{eq:qpebound3}:} We enforce some typical scaling such that the terms in \cref{eq:qpebound3} is dominated by those in \cref{eq:qpebound2}, up to multiplicative factors. For clarity, we collect the developments as follows:\begin{itemize}
  [leftmargin=5ex,topsep=0.25ex]

    \item From \cref{eq:U0_value} and $\zeta\in(0,\frac{\delta}{2})$ we have $\delta\sqrt{U}_0\gtrsim \frac{\delta}{\zeta}\omega(\calA^{(\rho_1)}_{\loc})\gtrsim \omega(\calA_{\loc}^{(\rho_1)})$, hence the terms $\omega(\calA^{(\rho_1)}_{\loc})\cdot(\omega(\calC)+\sqrt{m}\cdot \rad(\calD))$ from \cref{eq:qpebound3} are dominated by $\delta\sqrt{U_0}\cdot (\omega(\calC)+\sqrt{m}\cdot \rad(\calD))$ from \cref{eq:qpebound2};
    \item Note that \cref{eq:scaling_3cons} provides $\rho_1+\rho_2\lesssim\zeta$ with small enough implied constant.
    Then, from \cref{eq:U0_value} and $\zeta\in(0,\frac{\delta}{2})$ we have $\delta\sqrt{U}_0 \ge \sqrt{m\zeta\delta}\geq \sqrt{m}\zeta\gtrsim \sqrt{m}(\rho_1+\rho_2)$, and hence the terms $\sqrt{m}(\rho_1+\rho_2)\cdot(\omega(\calC)+\sqrt{m}\cdot \rad(\calD))$ from \cref{eq:qpebound3} are dominated by $\delta\sqrt{U}_0\cdot(\omega(\calC)+\sqrt{m}\cdot \rad(\calD))$ from \cref{eq:qpebound2}; 
    \item We further show that the remaining terms in \cref{eq:qpebound3}, namely $\big(\sqrt{m}\cdot \omega(\calA^{(\rho_1)}_{\loc})+m(\rho_1+\rho_2)\big)\cdot\rad(\calC)$, are dominated by $\delta U_0\sqrt{\log\frac{em}{U_0}}\cdot\rad(\calC)$ from \cref{eq:qpebound2}. It suffices to show 
    \begin{align}
        \sqrt{m}\cdot  \omega(\calA^{(\rho_1)}_{\loc})+m(\rho_1+\rho_2)\lesssim \delta U_0. \label{eq:desired_scaling}
    \end{align}
    By \cref{eq:U0_value} we have $\delta U_0\ge 6m\zeta$. On the other hand, \cref{eq:scaling_3cons} and \cref{eq:qpe_sample_size} evidently imply $\sqrt{m}\cdot \omega(\mathcal{A}^{(\rho_1)}_{\loc})+m(\rho_1+\rho_2) \lesssim m\zeta$. \cref{eq:desired_scaling} hence follows. 
\end{itemize} 
By the above discussions, under the     scaling conditions stated in our theorem statement, the terms in \cref{eq:qpebound2} dominate the ones in \cref{eq:qpebound3}.   Further substituting $t_1$ in \cref{eq:chooset1} yields the simplified bound: \begin{align}
    &\sup_{\bm{a}\in\calA}\sup_{\bm{b}\in\calB}\sup_{(\bm{c},\bm{d})\in\mathcal{E}}\big|\langle\bm{\xi}_{\bm{a},\bm{b}},\bm{\Phi c}+\sqrt{m}\bm{d}\rangle\big|\\&\quad\quad\le \sqrt{m}\delta \cdot \omega(\calE)+\sqrt{m}\delta\cdot \rad(\calE)\cdot\big(\sqrt{\scrH(\calA,\rho_1)}+\sqrt{\scrH(\calB,\rho_2)}\big)\\
    &\quad\quad\quad\quad + \delta \sqrt{U_0} \Big(\omega(\calC)+\sqrt{U_0\log\frac{em}{U_0}}\cdot\rad(\calC)+\sqrt{m}\cdot\rad(\calD)\Big).
\end{align}  
 We denote $U_0$ by $\hat{U}$ as in the theorem statement \cref{hatU}. 
 To complete the proof, it remains to make some final simplification: 
 \begin{itemize}
 [leftmargin=5ex,topsep=0.25ex]

     \item By \cref{eq:U0_value} we have $\log\frac{em}{U_0}\lesssim \log \frac{\delta}{\zeta}$ and hence we can relax $\delta U_0 \sqrt{\log\frac{em}{U_0}}\cdot\rad(\calC)$ to the term $\delta\hat{U}\sqrt{\log\frac{\delta}{\zeta}}\cdot\rad(\calC)$ in \cref{eq:thmbound2}; 
     \item By $U_0\le m$ we have $\delta\sqrt{U_0}\cdot\omega(\calC)\le \delta\sqrt{m}\cdot\omega(\calE)$, and we can only retain $\delta\sqrt{m}\cdot\omega(\calE)$ as in \cref{eq:thmbound1};
     \item By $\rad(\calE)\le \rad(\calC)+\rad(\calD)$ we can bound  $\sqrt{m}\delta\cdot\rad(\calE)\cdot (\sqrt{\scrH(\calA,\rho_1)}+\sqrt{\scrH(\calB,\rho_2)})$ by  $\sqrt{m}\delta\cdot\rad(\calC)\cdot (\sqrt{\scrH(\calA,\rho_1)}+\sqrt{\scrH(\calB,\rho_2)})+\sqrt{m}\delta\cdot\rad(\calD)\cdot (\sqrt{\scrH(\calA,\rho_1)}+\sqrt{\scrH(\calB,\rho_2)})$. Observe that $\sqrt{m}\delta\cdot\rad(\calD)\cdot (\sqrt{\scrH(\calA,\rho_1)}+\sqrt{\scrH(\calB,\rho_2)})$ is dominated by $\delta\sqrt{m\hat{U}}\cdot \rad(\calD)$ (due to \cref{hatU} and \cref{eq:qpe_sample_size}), 
     we can thus simply retain $\sqrt{m}\delta\cdot\rad(\calC)\cdot (\sqrt{\scrH(\calA,\rho_1)}+\sqrt{\scrH(\calB,\rho_2)})$ as in \cref{eq:thmbound2}.
 \end{itemize}
The proof is now complete. 
\end{proof}
\subsection{Global QPE for Structured  Sets}\label{app:special}
We consider the setting where   $\calA$ and $\calB$ in   \cref{thm:globalqpe} are structured sets with Kolmogorov entropy depending on the covering radius in a logarithmic manner; see \cref{defi1}. With properly chosen   parameters $(\zeta,\rho_1,\rho_2)$, \cref{thm:globalqpe} specializes to the following. We will explain in \cref{rem:sufficeqpe} that the QPE below is sufficient for the proofs of our main theorems (\cref{thm1}, \ref{thm2}, \ref{thm3}).
\begin{corollary}[Global QPE for Structured Sets]\label{coro:qpe_structured}
Given some bounded sets $\calA\subset \mathbb{R}^n$, $\calB\subset \mathbb{R}^m$, $\calE\subset \mathbb{R}^{n+m}$ and some $\delta>0$, we assume that the sub-Gaussian matrix $\bm{\Phi}\in \mathbb{R}^{m\times n}$ and the random dither $\bm{\tau}\sim \mathscr{U}[-\frac{\delta}{2},\frac{\delta}{2}]^m$ are as described in \cref{assump1}. Suppose that $(\zeta,\rho_1,\rho_2)$ are positive scalars satisfying 
\begin{gather}\label{eq:zeta_choice}
    \zeta = \frac{4\delta(\scrH(\calA,\rho_1)+\scrH(\calB,\rho_2))}{m}
    \\\label{eq:rho12_choice}
    \rho_1 \le \frac{c_1\zeta}{(\log\frac{\delta}{\zeta})^{1/2}},~\omega\big(\mathcal{A}_{\loc}^{(\rho_1)}\big) \le c_2 \zeta \sqrt{\frac{m\zeta}{\delta}},~
    \rho_2 \le c_3 \zeta \sqrt{\frac{\zeta}{\delta}}
\end{gather}
for some sufficiently small absolute constants $(c_1,c_2,c_3)$, and suppose that \begin{align}\label{eq:sample_coro4}
    m \ge C_4  \big(\mathscr{H}(\calA,\rho_1)+\scrH(\calB,\rho_2)\big)
\end{align}
for large enough $C_4$. Then,   with the quantization noise $\bm{\xi}_{\bm{a},\bm{b}}$ being given in \cref{eq:short_xi_a_b}, for some absolute constant $C_5$ the event 
\begin{align}
    \sup_{\bm{a}\in\calA}\sup_{\bm{b}\in\calB}\sup_{(\bm{c},\bm{d})\in\calE}\big|\langle\bm{\xi}_{\bm{a},\bm{b}},&\bm{\Phi c}+\sqrt{m}\bm{d}\rangle\big|\\& \le C_5\delta\sqrt{m}\Big(\omega(\calE)+\rad(\calE) \sqrt{\scrH(\calA,\rho_1)+\scrH(\calB,\rho_2)}\Big)
\end{align}
holds with probability exceeding $1-12\exp(-c_6\mathscr{H}(\calA,\rho_1)-c_6\scrH(\calB,\rho_2))$  on a single   draw of $(\bm{\Phi},\bm{\tau})$. 
\end{corollary}
\begin{proof}
We prove the statement using the general global QPE property presented in \cref{thm:globalqpe}. Given $\calE$, recall that $\calC$ and $\calD$ are defined in \cref{eq:defineC} and \cref{eq:defineD}. 
\subsubsection*{Verifying \cref{eq:scaling_3cons}--\cref{eq:qpe_sample_size}} First, we verify \cref{eq:scaling_3cons}. Under \cref{eq:sample_coro4}, $\zeta$ in \cref{eq:zeta_choice} evidently satisfies $\zeta\le c_1\delta$ with small enough $c_1$, verifying the first condition in \cref{eq:scaling_3cons}. Then, \cref{eq:rho12_choice} gives $\rho_2\le c_3\zeta \sqrt{\frac{\zeta}{\delta}}\le c_3\zeta$, which provides the third condition in \cref{eq:scaling_3cons}. Combining with $\rho_1\le \frac{c_1\zeta}{(\log\frac{\delta}{\zeta})^{1/2}}$, we know that the scaling conditions in \cref{eq:scaling_3cons} are   satisfied by \cref{eq:zeta_choice}--\cref{eq:rho12_choice}. Next, we verify \cref{eq:qpe_sample_size}. By substituting \cref{eq:zeta_choice} we find that it suffices to verify $m \ge \frac{C_2}{\zeta^2}\omega^2\big(\mathcal{A}_{\loc}^{(\rho_1)}\big)$ for some large enough $C_2$, and note that this is guaranteed by the second condition in \cref{eq:rho12_choice} that provides $m\ge \frac{\omega^2(\mathcal{A}_{\loc}^{(\rho_1)})}{c_2^2\zeta^2}\cdot \frac{\delta}{\zeta}$ for small enough $c_2$ (since $\delta\ge \zeta$).

\subsubsection*{Simplifying \cref{eq:thmbound1}--\cref{eq:thmbound2}} Because the last two conditions in \cref{eq:rho12_choice} imply \begin{align}
    \frac{m\rho_2^2}{\zeta^2}+\frac{\omega^2(\calA^{(\rho_1)}_{\loc})}{\zeta^2}\le (c_2^2+c_3)\frac{m\zeta}{\delta},
\end{align}
$\hat{U}$ given in \cref{hatU} simplifies to $\hat{U}\asymp \frac{m\zeta}{\delta}=4(\mathscr{H}(\calA,\rho_1)+\mathscr{H}(\calB,\rho_2))$, with the equality following from \cref{eq:zeta_choice}. Thus, we have
\begin{align}
    \frac{\hat{U}}{\sqrt{m}}\Big[\log\frac{\delta}{\zeta}\Big]^{1/2}& \lesssim \frac{\mathscr{H}(\calA,\rho_1)+\mathscr{H}(\calB,\rho_2)}{\sqrt{m}}\sqrt{\log \Big(\frac{m}{\mathscr{H}(\calA,\rho_1)+\mathscr{H}(\calB,\rho_2)}\Big)}\\
    &\lesssim \sqrt{\scrH(\calA,\rho_1)}+\sqrt{\scrH(\calB,\rho_2)}, 
\end{align}
where the second inequality follows from   \cref{eq:sample_coro4}. Therefore, the bound  in \cref{eq:thmbound1}--\cref{eq:thmbound2} simplifies to  \begin{align}
    &O\Big(\delta\sqrt{m}\cdot\omega(\calE)+\delta \sqrt{m\hat{U}}\cdot\rad(\calD) +\delta\sqrt{m}\cdot\rad(\calC)\sqrt{\mathscr{H}(\calA,\rho_1)+\mathscr{H}(\calB,\rho_2)}\Big)\\
    &=O\Big(\delta\sqrt{m}\big(\omega(\calE)+\rad(\calE)\sqrt{\mathscr{H}(\calA,\rho_1)+\mathscr{H}(\calB,\rho_2)}\big)\Big),
\end{align}
where in the second line we use  $\hat{U}\lesssim \mathscr{H}(\calA,\rho_1)+\mathscr{H}(\calB,\rho_2)$ and $\max\{\rad(\calC),\rad(\calD)\}\le \rad(\calE)$. We have arrived at the desired bound in \cref{coro:qpe_structured}, and note that the promised probability directly follows from \cref{thm:globalqpe}. The proof is complete. 
\end{proof}
\begin{rem}\label{rem:sufficeqpe}
  \cref{coro:qpe_structured} is tailored to fit the case where $\mathcal{A}$ and $\mathcal{B}$ are structured sets as per \cref{defi1}, but more generally put, it works well for $\mathcal{A},\mathcal{B}$ with Kolmogorov entropy logarithmically depending on the covering radius. By \cref{proadd1}(a), this is also the case when $\mathcal{A}$ and $\mathcal{B}$ are the ranges of some Lipschitz generative models (as per \cref{assump5}), thus \cref{coro:qpe_structured}  applies to the analysis of generative prior. Therefore, \cref{coro:qpe_structured} is a version of QPE sufficient for proving our main theorems, and we will further present other implications of \cref{thm:globalqpe} for the case where 
$\mathcal{A},\mathcal{B}$ are arbitrary sets in \cref{app:implication}.   
\end{rem}

 \section{Deferred Proofs}\label{appendixc}
 We collect the   proofs of  \cref{coro1} and \cref{coro2} (concrete outcomes of \cref{thm1}), \cref{coro3} and \cref{coro4} (concrete outcomes of \cref{thm2}), and  \cref{thm3} (for generative prior) in this appendix. 

 In the first two proofs we will invoke \cref{thm1}, with the two major steps being: (i) bounding the geometric complexity quantities $\gamma(\mathcal{D}_{\bm{x}}^*)$ and $\gamma(\mathcal{D}_{\bm{v}}^*)$; (ii) selecting $(\rho_1,\rho_2)$ to render \cref{eq:qpecon1}--\cref{eq:qpecon2}. 
 
 \subsection{The Proof of  \cref{coro1} (Recovering Sparse Signal and Sparse Corruption via Constrained Lasso)}
\begin{proof}~
    \subsubsection*{Step 1: Bounding  $\gamma(\mathcal{D}_{\bm{x}}^*)$ and $\gamma(\mathcal{D}_{\bm{v}}^*)$}
    For any $\bm{v}\in \mathcal{D}_{\bm{x}}$, there exists some $s$-sparse $\bm{x}$ such that  $\|\bm{x}+t\bm{v}\|_1\leq \|\bm{x}\|_1$ holds for some $t>0$. We let $\mathcal{S}=\mathrm{supp}(\bm{x})$. Because $\mathcal{S}\subset [n]$ and $|\mathcal{S}|\leq s$, we have (given $\bm{v}\in \mathbb{R}^n$ and $\mathcal{S}\subset [n]$ we obtain $\bm{v}_{\mathcal{S}}$ from $\bm{v}$ by only retaining entries in $\mathcal{S}$ while setting others zero)
    \begin{equation}
        \begin{aligned}\label{C.1}
            \|\bm{x}\|_1\geq \|\bm{x}+t\bm{v}_{\mathcal{S}}+t\bm{v}_{\mathcal{S}^c}\|_1=\|\bm{x}+t\bm{v}_\mathcal{S}\|_1+t\|\bm{v}_{\mathcal{S}^c}\|_1\geq \|\bm{x}\|_1-t\|\bm{v}_{\mathcal{S}}\|_1+t\|\bm{v}_{\mathcal{S}^c}\|_1,
        \end{aligned}
    \end{equation}
    which provides $\|\bm{v}_{\mathcal{S}^c}\|_1\leq \|\bm{v}_{\mathcal{S}}\|_1$. Thus, we obtain \begin{align}
        \|\bm{v}\|_1\leq 2\|\bm{v}_{\mathcal{S}}\|_1\leq 2\sqrt{s}\|\bm{v}_{\mathcal{S}}\|_2\leq 2\sqrt{s}\|\bm{v}\|_2.
    \end{align} Because $\mathcal{D}_{\bm{x}}^*=\mathcal{D}_{\bm{x}}\cap\mathbb{S}^{n-1}$, we have $\mathcal{D}_{\bm{x}}^*\subset \mathbb{B}_2^n\cap \mathbb{B}_1^n(2\sqrt{s})$, hence   \cref{pro5} gives $\omega(\mathcal{D}_{\bm{x}}^*)\lesssim \sqrt{s\log\frac{en}{s}}$. By   \cref{widthcomple}, this also bounds $\gamma(\mathcal{D}^*_{\bm{x}})$ (up to multiplicative constant). Similarly, we have $\gamma(\mathcal{D}^*_{\bm{v}})\lesssim \sqrt{k\log\frac{em}{k}}$.

    \subsubsection*{Step 2:  Selecting $(\rho_1,\rho_2)$} Recall that $\zeta = \frac{4\delta}{m}(\mathscr{H}(\mathcal{K}_{\bm{x}},\rho_1)+\mathscr{H}(\mathcal{K}_{\bm{v}},\rho_2))$ as per \cref{eq:qpecon1}, and we claim that setting
    \begin{align}\label{eq:coro1_rho12}
        \rho_1=c\delta\Big(\frac{s}{m}\Big)^{3/2},~\rho_2=c \delta\Big(\frac{k}{m}\Big)^{3/2}
    \end{align}  with small enough $c$  satisfies \cref{eq:qpecon2}. The reasoning is as follows: 
    \begin{itemize}
    [leftmargin=5ex,topsep=0.25ex]
        \item In view of \cref{pro6}, $m\gtrsim s\log(\frac{nm^{3/2}}{s^{5/2}\delta})+k\log(\frac{m^{5/2}}{k^{5/2}\delta})$ in \cref{coro1} implies \cref{eq:thm1_sam_size} needed in \cref{thm1}; 
        \item \textbf{Verifying $\rho_1\lesssim\frac{\zeta}{(\log\frac{\delta}{\zeta})^{1/2}}$, $\rho_2 \lesssim \zeta\sqrt{\frac{\zeta}{\delta}}$:} By \cref{eq:thm1_sam_size} we have $m\gtrsim \mathscr{H}(\mathcal{K}_{\bm{x}},\rho_1)+ \mathscr{H}(\mathcal{K}_{\bm{v}},\rho_2)$ that implies $\zeta\lesssim\delta$. Substituting $\zeta = \frac{4\delta}{m}(\mathscr{H}(\mathcal{K}_{\bm{x}},\rho_1)+\mathscr{H}(\mathcal{K}_{\bm{v}},\rho_2))$  finds   
        $$\rho_1+\rho_2\lesssim \delta\Big(\frac{s}{m}\Big)^{3/2}+\delta\Big(\frac{k}{m}\Big)^{3/2}\lesssim \zeta\sqrt{\frac{\zeta}{\delta}}.$$ Because  $\frac{\delta}{\zeta}\geq C_1$ holds for some large $C_1$, and so $\rho_1+\rho_2\le c_2 \zeta\sqrt{\frac{\zeta}{\delta}}$ with small enough $c_2$ suffices to ensure $\rho_1\lesssim\frac{\zeta}{(\log\frac{\delta}{\zeta})^{1/2}}$ and $\rho_2 \lesssim \zeta\sqrt{\frac{\zeta}{\delta}}$.  
        \item  \textbf{Verifying $\omega(\mathcal{A}^{(\rho_1)}_{\loc})\lesssim \zeta\sqrt{\frac{m\zeta}{\delta}}$:} Observe that $\mathcal{A}^{(\rho_1)}_{\loc}\subset \Sigma^n_{2s}\cap \rho_1\mathbb{B}_2^n$, and so $\omega(\mathcal{A}^{(\rho_1)}_{\loc})\lesssim \rho_1 s\log\frac{en}{s}$. Further, by $\rho_1\lesssim\zeta$ and $s\log\frac{em}{s}\lesssim \sqrt{\mathscr{H}(\mathcal{A},\rho_1)}\lesssim \sqrt{\frac{m\zeta}{\delta}}$ we arrive at $\omega(\mathcal{A}^{(\rho_1)}_{\loc})\lesssim\zeta\sqrt{\frac{m\zeta}{\delta}}$.   
    \end{itemize}
    Note that we have derived explicit bounds on the geometric quantities and 
    chosen $\rho_1,\rho_2$ to satisfy the conditions in \cref{thm1}. Now we can simply invoke \cref{thm1} to prove the desired claim. 
\end{proof}

\subsection{The Proof of   \cref{coro2} (Recovering Low-Rank Signal and Sparse Corruption via Constrained Lasso)}

 \begin{proof}~
      \subsubsection*{Step 1: Bounding  $\gamma(\mathcal{D}_{\bm{x}}^*)$ and $\gamma(\mathcal{D}_{\bm{v}}^*)$}
            \cite[Coro. 2.1]{fuchs2012proof} gives  $\gamma(\mathcal{D}^*_{\bm{x}})\lesssim \sqrt{r(p+q)}$. As shown in the proof of   \cref{coro1}, we have $\gamma(\mathcal{D}_{\bm{v}}^*)\lesssim \sqrt{k\log \frac{em}{k}}$.

      \subsubsection*{Step 2:  Selecting $(\rho_1,\rho_2)$} This is similar to Step 2 in the proof of \cref{coro1}. We recall  $\zeta = \frac{4\delta}{m}(\mathscr{H}(\mathcal{K}_{\bm{x}},\rho_1)+\mathscr{H}(\mathcal{K}_{\bm{v}},\rho_2))$ as per \cref{eq:qpecon1} and claim that the choice 
      \begin{align}
          \rho_1=c\delta \Big(\frac{r(p+q
          )}{m}\Big)^{3/2},~\rho_2= c\delta \Big(\frac{k}{m}\Big)^{3/2}
      \end{align}with small enough $c$
      satisfies the \cref{eq:qpecon1}--\cref{eq:qpecon2} in \cref{thm1}. We note the following dot points to explain this: 
      \begin{itemize}
       [leftmargin=5ex,topsep=0.25ex]
          \item In view of \cref{pro6}, $m\gtrsim r(p+q)\log(\frac{m^{3/2}}{\delta(r(p+q))^{3/2}})+k\log(\frac{m^{5/2}}{k^{5/2}\delta})$ in \cref{coro2} implies \cref{eq:thm1_sam_size} needed in \cref{thm1};
          \item \textbf{$\rho_1\lesssim\frac{\zeta}{(\log\frac{\delta}{\zeta})^{1/2}}$, $\rho_2 \lesssim \zeta\sqrt{\frac{\zeta}{\delta}}$:} \cref{eq:thm1_sam_size} implies $m\gtrsim \mathscr{H}(\mathcal{K}_{\bm{x}},\rho_1)+ \mathscr{H}(\mathcal{K}_{\bm{v}},\rho_2)$ and hence $\zeta\lesssim\delta$. By $\zeta=\frac{4\delta}{m}(\mathscr{H}(\mathcal{K}_{\bm{x}},\rho_1)+\mathscr{H}(\mathcal{K}_{\bm{v}},\rho_2))$ it is easy to verify $$\rho_1+\rho_2\lesssim \delta\Big(\frac{r(p+q)}{m}\Big)^{3/2}+\delta\Big(\frac{k}{m}\Big)^{3/2}\lesssim \zeta\sqrt{\frac{\zeta}{\delta}}.$$ 
          Because $\frac{\delta}{\zeta}\geq C_1$ holds for some large $C_1$,  $\rho_1+\rho_2\lesssim\zeta\sqrt{\frac{\zeta}{\delta}}$  ensures  $\rho_1\lesssim\frac{\zeta}{(\log\frac{\delta}{\zeta})^{1/2}}$ and $\rho_2 \lesssim \zeta\sqrt{\frac{\zeta}{\delta}}$.  
          \item \textbf{$\omega(\mathcal{A}^{(\rho_1)}_{\loc})\lesssim \zeta\sqrt{\frac{m\zeta}{\delta}}$:} Observe that $\mathcal{A}^{(\rho_1)}_{\loc}\subset M^{p,q}_{2r}\cap \mathbb{B}_{\text{F}}^{p,q}(\rho_1)$, and so $\omega(\mathcal{A}^{(\rho_1)}_{\loc})\lesssim \rho_1r(p+q)$. Further,  $\rho_1\lesssim \zeta$ and $r(p+q)\lesssim \sqrt{\mathscr{H}(\mathcal{A},\rho_1)}\lesssim \sqrt{\frac{m\zeta}{\delta}}$ imply $\omega(\mathcal{A}_{\loc}^{(\rho_1)})\lesssim \zeta\sqrt{\frac{m\zeta}{\delta}}$.    
      \end{itemize}
      Applying \cref{thm1} yields the desired \cref{coro2}.
 \end{proof}

Next, we prove the recovery guarantees for unconstrained Lasso. We apply \cref{thm2} by several steps: (1) Verifying \cref{assump4}, (2) Selecting $(\rho_1,\rho_2)$ to render \cref{eq:qpecon1un}--\cref{eq:qpecon2un}, and (3) Estimating the geometric quantities.
 \subsection{The Proof of   \cref{coro3} (Recovering Sparse Signal and Sparse Corruption via Unconstrained Lasso)}

\begin{proof}

We present the proof in three steps.

    \subsubsection*{Step 1: Verifying  \cref{assump4}}

    Given any $\bm{a}\in\mathcal{K}_{\bm{x}}$, we   take $\mathcal{X}_{\bm{a}}=\overline{\mathcal{X}}_{\bm{a}}=\{\bm{w}\in \mathbb{R}^n:\mathrm{supp}(\bm{w})\subset \mathrm{supp}(\bm{a})\}$. Then we have $\overline{\mathcal{X}}_{\bm{a}}^\bot= \{\bm{w}\in \mathbb{R}^n:\supp(\bm{w})\subset ([n]\setminus\supp(\bm{a}) )\}$,
    and note that decomposibility \cref{3.39} immediately follows since $\|\bm{w}_1+\bm{w}_2\|_1=\|\bm{w}_1\|_1+\|\bm{w}_2\|_1$ holds for any $\bm{w}_1,\bm{w}_2\in \mathbb{R}^n$ with $\supp(\bm{w}_1)\subset \supp(\bm{a})$ and $\supp(\bm{w}_2)\subset ([n]\setminus \supp(\bm{a}))$. Moreover, for any $\bm{w}\in \overline{\mathcal{X}}_{\bm{a}}$ we have $\|\bm{w}\|_1\leq \sqrt{s}\|\bm{w}\|_2$, and hence $\alpha_{\|\cdot\|_1}(\overline{\mathcal{X}}_{\bm{a}})\leq \sqrt{s}$; This holds uniformly for all $\bm{a}\in\mathcal{K}_{\bm{x}}$ and hence we can take $\alpha_{\bm{x}}=\sqrt{s}$. Similarly, regarding the $k$-sparse corruption, \cref{assump4} is also satisfied  with $\alpha_{\bm{v}}=\sqrt{k}$.

     \subsubsection*{Step 2: Selecting $(\rho_1,\rho_2)$}

     As shown in Step 2 of the proof of   \cref{coro1}, setting $\rho_1=c\delta(\frac{s}{m})^{3/2}$ and $\rho_2=c\delta(\frac{k}{m})^{3/2}$ with small enough $c$ satisfies \cref{eq:qpecon1un}--\cref{eq:qpecon2un}.

     \subsubsection*{Step 3: Estimating Geometric Quantities}
     For the estimations of $\mathscr{H}(\mathcal{K}_{\bm{x}},\cdot)$ and  $\mathscr{H}(\mathcal{K}_{\bm{v}},\cdot)$ we use   \cref{pro6}. For the $\ell_1$-ball we have $\omega(\mathbb{B}_f^n)\asymp \sqrt{\log n}$ and $\omega(\mathbb{B}_g^m)\asymp \sqrt{\log m}$ \cite[Example 7.5.9]{vershynin2018high}. Now we can easily see that $(\lambda_1,\lambda_2)$ in \cref{coro3} satisfies \cref{eq:lam1choice}--\cref{eq:lam2choice}, and the sample complexity stated in \cref{coro3} satisfies \cref{eq:thm2sample}.

     With the above preparations,  the result immediately follows from   \cref{thm2}.
     \end{proof}

\subsection{The Proof of   \cref{coro4} (Recovering Low-Rank Signal and Sparse Corruption via Unconstrained Lasso)}

\begin{proof}
We present the proof in three steps. 

\subsubsection*{Step 1: Verifying   \cref{assump4}}
 
Given any $\bm{a}\in M_r^{p,q}$, we let its singular value decomposition be 
\begin{equation}
    \bm{a} = \begin{bmatrix}
    \bm{U}_1&\bm{U}_2
    \end{bmatrix}\begin{bmatrix}
        \bm{\Sigma}&\bm{0}\\\bm{0} &\bm{0}
    \end{bmatrix}\begin{bmatrix}
        \bm{V}_1^\top\\\bm{V}_2^\top
    \end{bmatrix}.
\end{equation}
Then we define $\mathcal{X}_{\bm{a}}=\{\bm{U}_1\bm{A}\bm{V}_1^\top:\bm{A}\in \mathbb{R}^{r\times r}\}$, $\overline{\mathcal{X}}_{\bm{a}}=\{\bm{U}_1\bm{A}\bm{V}_1^\top+\bm{U}_1\bm{B}\bm{V}_2^\top+\bm{U}_2\bm{C}\bm{V}_1^\top:\bm{A}\in \mathbb{R}^{r\times r},\bm{B}\in \mathbb{R}^{r\times(q-r)},\bm{C}\in\mathbb{R}^{(p-r)\times r}\}$. Note that $\overline{\mathcal{X}}^\bot_{\bm{a}}= \{\bm{U}_2\bm{A}\bm{V}_2^\top:\bm{A}\in\mathbb{R}^{(p-r)\times(q-r)}\}$, then it is not hard to verify the decomposibility \cref{3.39}, since  $\|\bm{w}_1+\bm{w}_2\|_{\rm nu}=\|\bm{w}_1\|_{\rm nu}+\|\bm{w}_2\|_{\rm nu}$ holds for any $\bm{w}_1\in \mathcal{X}_{\bm{a}}$ and $\bm{w}_2\in \overline{\mathcal{X}}^\bot_{\bm{a}}$. Moreover,   any $\bm{w}\in \overline{\mathcal{X}}_{\bm{a}}$ has rank not exceeding $2r$, and hence we have $\|\bm{x}\|_{\rm nu}\leq \sqrt{2r}\|\bm{x}\|_{\rm F}$. Note that $\|\cdot\|_{\rm F}$ is just the $\ell_2$-norm when $\bm{x}$ is viewed as vector, thus we have $\alpha_{\|\cdot\|_{\rm nu}}(\mathcal{X}_{\bm{a}})\le \sqrt{2r}$. This holds uniformly for all $\bm{a}\in M^{p,q}_r$, thus we can take $\alpha_{\bm{x}}=\sqrt{2r}$ in \cref{assump4}. It has been shown in  the proof of   \cref{coro3} that the $k$-sparse corruption satisfies \cref{assump4} with $\alpha_{\bm{v}}=\sqrt{k}$.

  \subsubsection*{Step 2: Selecting $(\rho_1,\rho_2)$} As shown in Step 2 in the proof of \cref{coro2}, setting $\rho_1=c\delta (\frac{r(p+q)}{m})^{3/2}$ and $\rho_2=c\delta(\frac{k}{m})^{3/2}$ with small enough $c$, along with the sample complexity stated in \cref{coro4}, satisfies \cref{eq:qpecon1un}--\cref{eq:qpecon2}.

     \subsubsection*{Step 3: Estimating Geometric Quantities}
 For estimations of $\mathscr{H}(\mathcal{K}_{\bm{x}},\cdot)$ and $\mathscr{H}(\mathcal{K}_{\bm{v}},\cdot)$ we use   \cref{pro6}. \cite[Example 7.5.9]{vershynin2018high} gives $\omega(\mathbb{B}_g^m)\asymp \sqrt{\log m}$. Moreover, we show  $\omega(\mathbb{B}_f^n)=\omega(\mathbb{B}^{p,q}_{\rm nu})\asymp \sqrt{p+q}$ in the following. First, note that a matrix that has only one non-zero row in $\mathbb{B}^{q}_2$ (or only one non-zero column in $\mathbb{B}_2^{p}$) belongs to $\mathbb{B}_{\rm nu}^{p,q}$, which implies $\omega(\mathbb{B}^{p,q}_{\rm nu})\geq \max\{\omega(\mathbb{B}_2^p),\omega(\mathbb{B}_2^q)\}=\Omega(\sqrt{p+q})$. Second, let $\bm{G}\sim \mathcal{N}^{p\times q}(0,1)$ we have $\omega(\mathbb{B}_{\rm nu}^{p,q})=\mathbbm{E}\sup_{\|\bm{A}\|_{\rm nu}=1}\langle\bm{G},\bm{A}\rangle\leq \mathbbm{E}\|\bm{G}\|_{op}\lesssim \sqrt{p+q}$ \cite[Exercise 4.4.6]{vershynin2018high}. 
 
 With the above preparations, we are ready to invoke   \cref{thm2} to obtain the desired claim.
\end{proof}

\subsection{The Proof of \cref{thm3} (Uniform Recovery Guarantee under Generative Priors)}
\begin{proof}
By writing   $\mathcal{K}_{\bm{x}}^-=\mathcal{K}_{\bm{x}}-\mathcal{K}_{\bm{x}}$ and $\mathcal{K}_{\bm{v}}^-=\mathcal{K}_{\bm{v}}-\mathcal{K}_{\bm{v}}$, it is immediate from the constraint of \cref{geneprogram} that 
\begin{equation}\label{generange}
   \bm{\Delta_x}=\bm{\hat{x}}-\bm{x^\star} \in \mathcal{K}_{\bm{x}}^-,~\bm{\Delta_v}=\bm{\hat{v}}-\bm{v^\star} \in \mathcal{K}_{\bm{v}}^-.
\end{equation}
    We may omit some details  because the techniques are analogous to those for proving   \cref{thm1}. We present the proofs in three steps. 

    \subsubsection*{Step 1: Problem Reduction} We first reduce the proof to bounding several random processes.

    \textbf{Identifying Constraint Sets:} Note  that we want to prove $\sqrt{\|\bm{\Delta_x}\|_2^2+\|\bm{\Delta_v}\|_2^2}\leq \mu$ for all $(\bm{x^\star},\bm{v^\star})\in \mathcal{K}_{\bm{x}}\times \mathcal{K}_{\bm{v}}$ and for some given accuracy $\mu\in (0,1)$. Up to rescaling it suffices to prove $\|\bm{\Delta}\|_2=\sqrt{\|\bm{\Delta_x}\|_2^2+\|\bm{\Delta_v}\|_2^2}\leq 3\mu$. Hence, we can assume \begin{equation}
        \label{bigerror} \sqrt{\|\bm{\Delta_x}\|_2^2+\|\bm{\Delta_v}\|_2^2}\geq 2\mu
    \end{equation} since the bound holds trivially when $(\|\bm{\Delta_x}\|_2^2+\|\bm{\Delta_v}\|_2^2)^{1/2}<2\mu$. Therefore, we can proceed with the constraint 
    \begin{align}
        (\bm{\Delta_x},\bm{\Delta_v})\in \mathcal{E} :=\big\{(\bm{c},\bm{d}):\bm{c}\in \mathcal{K}_{\bm{x}}^-,~\bm{d}\in \mathcal{K}_{\bm{v}}^-,~(\|\bm{c}\|_2^2+\|\bm{d}\|_2^2)^{1/2}\geq 2\mu\big\},
    \end{align}
    where the constraint set $\mathcal{E}$ is defined as per \cref{eq:gene_E}. To accmmodate the normalized error, 
    we   introduce $\mathcal{E}^*$   \begin{align}\label{eq:con_estar}
     \frac{(\bm{\Delta_x},\bm{\Delta_v})}{(\|\bm{\Delta_x}\|_2^2+\|\bm{\Delta_v}\|_2^2)^{1/2}} \in  \mathcal{E}^*=\big\{(\bm{c},\bm{d})/(\|\bm{c}\|^2_2+\|\bm{d}\|_2^2)^{1/2}:(\bm{c},\bm{d})\in \mathcal{E}\big\} 
    \end{align}
    as per \cref{eq:gene_E_star}. Besides, we further define 
    \begin{gather}\label{eq:constrain_C}
        \mathcal{C} = \big\{\bm{c}\in \mathbb{R}^n:(\bm{c},\bm{d})\in \mathcal{E}^*\text{ for some }\bm{d}\in\mathbb{R}^m\big\}\\\label{eq:constrain_D}
         \mathcal{D} = \big\{\bm{d}\in \mathbb{R}^m:(\bm{c},\bm{d})\in \mathcal{E}^*\text{ for some }\bm{c}\in\mathbb{R}^n\big\},
    \end{gather}
    then we note the   relation  \begin{align}\label{eq:GWbound}
        \max\{\omega(\mathcal{C}),\omega(\mathcal{D})\}\le \omega(\mathcal{E}^*)\lesssim \Big(k\log\frac{Lr}{\mu}+k'\log\frac{L'r'}{\mu}\Big)^{1/2},
    \end{align}
    where the first inequality can be seen by \cite[Exercise 7.5.4]{vershynin2018high}, the second inequality follows from \cref{proadd1}(c). 
    In the sequel, we will proceed with the constraints
    \begin{align}
        &\bm{\Delta_x}/(\|\bm{\Delta_x}\|_2^2+\|\bm{\Delta_v}\|_2^2)^{1/2} \in \mathcal{C},\\
        & \bm{\Delta_v}/(\|\bm{\Delta_x}\|_2^2+\|\bm{\Delta_v}\|_2^2)^{1/2}\in\mathcal{D}.
    \end{align}    
    Note that all above constraints hold universally for all $(\bm{x^\star},\bm{v^\star})$ {\it that need further consideration}; those $(\bm{x^\star},\bm{v^\star})$ that fail to satisfy these constraints must satisfy $(\|\bm{\Delta_x}\|_2^2+\|\bm{\Delta_v}\|_2^2)^{1/2}<2\mu$ and are already done.

\textbf{Using Optimality:} From $\|\bm{\dot{y}}-\bm{\Phi}\bm{\hat{x}}-\sqrt{m}\bm{\hat{v}}\|_2\leq \|\bm{\dot{y}}-\bm{\Phi x^\star}-\sqrt{m}\bm{v^\star}\|_2$, we substitute $\bm{\hat{x}}=\bm{x^\star}+\bm{\Delta_x}$ and $\bm{\hat{v}}=\bm{v^\star}+\bm{\Delta_v}$, expand the square, and then substitute \cref{eq:quan_noise} to obtain  
    \begin{equation}
        \label{C.4}\|\bm{\Phi\Delta_x}+\sqrt{m}\bm{\Delta_v}\|_2^2\leq 2\langle\bm{\epsilon}+\bm{\xi}_{\bm{x^\star},\bm{v^\star}},\bm{\Phi\Delta_x}+\sqrt{m}\bm{\Delta_v}\rangle,
    \end{equation}
    see \cref{eq:xi_xv} and \cref{eq:xiabfirst} for    $\bm{\xi}_{\bm{x^\star},\bm{v^\star}}$.  Combining with the constraints in \cref{eq:con_estar}, \cref{eq:constrain_C} and \cref{eq:constrain_D}, we bound both sides of \cref{C.4} to arrive at \begin{align}
        \label{eq:gene_I1234}\big(\|\bm{\Delta}_{\bm{x}}\|_2^2+\|\bm{\Delta}_{\bm{v}}\|_2^2\big)\cdot I_1\le 2\big(\|\bm{\Delta_x}\|_2^2+\|\bm{\Delta_v}\|_2^2\big)^{1/2}\cdot\big(  I_2 +   I_3 +    I_4\big)
    \end{align}
    where the random terms (to be bounded) are given by  
    \begin{gather}
        I_1: = \inf_{(\bm{c},\bm{d})\in \mathcal{E}^*}\|\bm{\Phi c}+\sqrt{m}\bm{d}\|_2^2,~I_2 :=\sup_{\bm{c}\in \mathcal{C}}\langle\bm{\epsilon},\bm{\Phi c}\rangle \\
        I_3:= \sup_{\bm{d}\in\mathcal{D}}\langle \bm{\epsilon},\sqrt{m}\bm{d}\rangle,~I_4:= \sup_{\bm{a}\in \mathcal{K}_{\bm{x}}} \sup_{\bm{b}\in \mathcal{K}_{\bm{v}}}\sup_{(\bm{c},\bm{d})\in \mathcal{E}^*}\langle\bm{\xi}_{\bm{a},\bm{b}},\bm{\Phi c}+\sqrt{m}\bm{d}\rangle
    \end{gather}

\subsubsection*{Step 2: Bounding $I_1,I_2,I_3,I_4$}

Parallel to the proof of \cref{thm1}, our techniques to bound $I_1,I_2,I_3,I_4$ are \cref{pro1}, \cref{pro4}, \cref{pro2} and \cref{coro:qpe_structured}, respectively. The additional technicalities are the estimates on Gaussian width and Kolmogorov entropy developed in \cref{proadd1}.

 {\bf Bounding $I_1$:}       
Because $\mathcal{E}^*\subset \mathbb{S}^{n+m-1}$, \cref{pro1} yields that the event 
\begin{align}\label{eq:gene_extend}
    \sup_{(\bm{c},\bm{d})\in \mathcal{E}^*}\Big|\|\bm{\Phi c}+\sqrt{m}\bm{d}\|_2 -\sqrt{m}\Big|\le C_1\big(\gamma(\mathcal{E}^*)+t\big)
\end{align}
holds with probability exceeding $1-\exp(-t^2)$. Notice that  
 the bound on 
$\omega(\mathcal{E}^*)$ in \cref{eq:GWbound} remains valid for $\gamma(\mathcal{E}^*)$ due to \cref{widthcomple}. Since \cref{genesample} implies $m\gtrsim  k\log\frac{Lr}{\mu}+k'\log\frac{L'r'}{\mu}$,  we can 
set $t = \sqrt{k\log\frac{Lr}{\mu}+k'\log\frac{L'r'}{\mu}}$ and still assume that the right-hand side of \cref{eq:gene_extend} is bounded by $\frac{\sqrt{m}}{2}$, thus obtaining that the event 
\begin{align}\label{C.77}
     \sqrt{I_1}=&\inf_{(\bm{c},\bm{d})\in \mathcal{E}^*}\|\bm{\Phi c}+\sqrt{m}\bm{d}\|_2\\&\quad \ge \sqrt{m}- \sup_{(\bm{c},\bm{d})\in \mathcal{E}^*}\big|\|\bm{\Phi c}+\sqrt{m}\bm{d}\|_2 -\sqrt{m}\big|\ge \frac{\sqrt{m}}{2} \label{eq:gene_bound_on_I1}
\end{align}
holds with probability exceeding $1-\exp(-k\log\frac{Lr}{\mu}-k'\log\frac{L'r'}{\mu})$. 

{\bf Bounding $I_2$:} We derive the bound following similar courses as in the corresponding part in the proof of \cref{thm1}. Conditioning on $\bm{\epsilon}$, for any $t\ge0$, \cref{pro4} gives that the event 
\begin{align}\label{eq:gene_boundII2}
    I_2 \le C_2 \|\bm{\epsilon}\|_2 \big(\omega(\mathcal{C})+t\big)
\end{align}
holds with probability exceeding $1-2\exp(-t^2)$. By repeating the argument in \cref{3.1999} we can show that $\|\bm{\epsilon}\|_2\lesssim E\sqrt{m}$ holds with probability exceeding $1-\exp(-\Omega(m))$, and hence also exceeding $1-\exp(-\Omega(k\log\frac{Lr}{\mu}+k'\log\frac{L'r'}{\mu}))$ by \cref{genesample}. 
Combining with \cref{eq:GWbound}, we set $t= (k\log\frac{Lr}{\mu}+k'\log\frac{L'r'}{\mu})^{1/2}$ in \cref{eq:gene_boundII2} to obtain the bound on $I_2$ \begin{align}\label{eq:gene_bound_I2}
    I_2 \lesssim E\sqrt{m}\cdot \Big(k\log\frac{Lr}{\mu}+k'\log\frac{L'r'}{\mu}\Big)^{1/2},
\end{align}
with probability exceeding $1-3\exp(-\Omega(k\log\frac{Lr}{\mu}+k'\log\frac{L'r'}{\mu}))$.

{\bf Bounding $I_3$:} Due to \cref{eq:verify_proa2}, for any $t\ge 0$, \cref{pro2} yields
that the event $I_3 \lesssim E\sqrt{m}(\omega(\mathcal{D})+t)$ holds with probability exceeding $1-2\exp(-t^2)$. Combining with \cref{eq:GWbound}, we set $t =(k\log\frac{Lr}{\mu}+k'\log\frac{L'r'}{\mu})^{1/2}$  to obtain that the bound on $I_3$ 
\begin{align}
    \label{eq:gene_bound_I3}
    I_3 \lesssim E\sqrt{m}\cdot \Big(k\log\frac{Lr}{\mu}+k'\log\frac{L'r'}{\mu}\Big)^{1/2}
\end{align}
holds with probability exceeding $1-2\exp(-k\log\frac{Lr}{\mu}-k'\log\frac{L'r'}{\mu})$. 

{\bf Bounding $I_4$:} We apply \cref{coro:qpe_structured} to bound $I_4$, so the major work lies in selecting $(\rho_1,\rho_2)$  such that \cref{eq:zeta_choice}--\cref{eq:rho12_choice} hold. 
We claim that setting 
\begin{align}
    \rho_1 = c\delta \Big(\frac{k}{n}\Big)^{3/2},~\rho_2 = c \delta\Big(\frac{k'}{m}\Big)^{3/2} 
\end{align}
works, with the reasoning provided below:\footnote{Unlike in the case of structured priors, we do not aim to carefully choose $\rho_1$ but simply set it small enough to justify \cref{eq:rho12_choice}. The reason is that other parameters $(L,r,L',r')$ appearing in the logarithm   typically dominate $(m,n)$. (As a result, most works in generative compressed sensing do not refine logarithmic factor.)}  
\begin{itemize}
    [leftmargin=5ex,topsep=0.25ex]
    \item In general, we do not have lower bound on the Kolmogorov entropy of $\mathcal{K}_{\bm{x}}$ and $\mathcal{K}_{\bm{v}}$ in \cref{eq:gene_prior}, but we note that the $\mathscr{H}(\mathcal{A},\rho_1)$ and the  $\mathscr{H}(\mathcal{B},\rho_2)$ appearing in the statement of \cref{coro:qpe_structured} can be replaced by their upper bounds, and we will simply use \cref{eq:bound_on_entropy}. 
    \item That being mentioned, from \cref{eq:zeta_choice} we have $\zeta\gtrsim \frac{\delta(k+k')}{m}$, and so $\rho_2 = c\delta(\frac{k'}{m})^{3/2}\lesssim \zeta\sqrt{\frac{\zeta}{\delta}}$. Since \cref{genesample} implies $m\gtrsim k+k'$, we can assume that $\frac{\delta}{\zeta}$ is sufficiently large. Hence, to ensure $\rho_1 \lesssim \frac{\zeta}{(\log \frac{\delta}{\zeta})^{1/2}}$, it suffices to ensure $\rho_1\lesssim \zeta\sqrt{\frac{\zeta}{\delta}}$. Under $m=O(n)$,  $\rho_1=c\delta(\frac{k}{n})^{3/2}$ with small enough $c$ evidently satisfies this. 

    \item It remains to verify the second condition in \cref{eq:rho12_choice}, namely $\omega(\mathcal{A}^{(\rho_1)}_{\loc})\lesssim \zeta\sqrt{\frac{m\zeta}{\delta}}$, and by $\zeta\gtrsim \frac{\delta k}{m}$ it suffices to ensure $\omega(\mathcal{A}^{(\rho_1)}_{\loc})\lesssim \frac{\delta k ^{3/2}}{m}$. This can be justified by $\omega(\mathcal{A}^{(\rho_1)}_{\loc})\le \rho_1 \omega(\mathbb{B}^n_2)\le \rho_1\sqrt{n}=c\delta \frac{k^{3/2}}{n}$, where we use $\omega(\mathbb{B}_2^n)\le\sqrt{n}$ from \cite[Prop. 7.5.2(f)]{vershynin2018high}.
\end{itemize}
Moreover, note that \cref{genesample} implies $m\gtrsim k\log (\frac{Lrn^{3/2}}{\delta k^{3/2}})+k'\log (\frac{L'r'm^{3/2}}{\delta (k')^{3/2}})$, and thus we can apply \cref{coro:qpe_structured} to obtain that the bound 
\begin{align}
    \label{eq:gene_bound_on_I4}
    I_4\lesssim \delta\sqrt{mk\log\Big(\frac{Lrn^{3/2}}{\mu \delta k^{3/2}}\Big)+mk'\log\Big(\frac{L'r'm^{3/2}}{\mu\delta (k')^{3/2}}\Big)}
\end{align}
that holds with probability exceeding $1-12\exp(-\Omega(k\log(\frac{Lrn^{3/2}}{\mu\delta k^{3/2}})+k'\log(\frac{L'r'm^{3/2}}{\mu\delta(k')^{3/2}})))$.

\subsubsection*{Step 3: Combining Everything}  We substitute the bounds \cref{eq:gene_bound_on_I1}, \cref{eq:gene_bound_I2}, \cref{eq:gene_bound_I3}, \cref{eq:gene_bound_on_I4} into \cref{eq:gene_I1234} and perform simple rearrangement. This yields the   bound (universally for all  $(\bm{x^\star},\bm{v^\star})$ that may violate $(\|\bm{\Delta_x}\|_2^2+\|\bm{\Delta_v}\|_2^2)^{1/2}\le 2\mu$) 
\begin{align}\nonumber
    \big(\|\bm{\Delta_x}\|_2^2+
    \|\bm{\Delta_v}\|_2^2\big)^{1/2}\lesssim \frac{E\sqrt{k\log(\frac{Lr}{\mu})+k'\log(\frac{L'r'}{\mu})}+\delta \sqrt{k\log(\frac{Lrn^{3/2}}{\mu\delta k^{3/2}})+k'\log(\frac{L'r'm^{3/2}}{\mu\delta(k')^{3/2}})}}{\sqrt{m}}
\end{align}
with probability exceeding $1-C_3 \exp(-\Omega(k\log(Lr)+k'\log(L'r')))$. Therefore, under the sample size given in \cref{genesample}, we again obtain $(\|\bm{\Delta_x}\|_2^2+\|\bm{\Delta_v}\|_2^2)^{1/2}\le 3\mu$, completing the proof. 
\end{proof}
\section{Technical By-Product} \label{app:by}
We demonstrate that our global QPE property \cref{thm:globalqpe} is a generalization and instance-wise improvement (under Gaussian sensing matrix $\bm{\Phi}$) of the one developed in \cite{xu2020quantized}. Then, as an interesting enough technical by-product, we improve the uniform error rate of the {\it projected back-projection} (PBP) estimator over bounded convex signal set in \cite{xu2020quantized} from $O(m^{-1/16})$ to $O(m^{-1/8})$.

\subsection{Implications of \cref{thm:globalqpe}} \label{app:implication}
Recall that we have specialized \cref{thm:globalqpe} to the case where $\mathcal{A}$ and $\mathcal{B}$ are {\it structured sets} (see \cref{defi1}) in \cref{coro:qpe_structured}. Here, for {\it arbitrary} sets $\mathcal{A}\subset \mathbb{R}^n$ and $\mathcal{B}\subset \mathbb{R}^m$, we further present some direct outcomes of our general \cref{thm:globalqpe} and compare them with \cite{xu2020quantized}.

The key ingredient in \cite{xu2020quantized} for achieving global QPE is their Proposition 6.1, which can be recovered from our \cref{thm:globalqpe} by setting $\mathcal{A}=\mathcal{C}=\{0\}$ (see \cref{rem:recover}). 

\begin{corollary}[QPE (Almost) Coincident with Proposition 6.1 in \cite{xu2020quantized}] \label{coro:recover}Given $\mathcal{B},\mathcal{D}\subset \mathbb{B}_2^m$, we consider the uniform quantizer $\mathcal{Q}_\delta(\cdot)$ associated with uniform dither $\bm{\tau}\sim \mathscr{U}([-\frac{\delta}{2},\frac{\delta}{2}]^m)$. Given any small enough $\epsilon>0$, if \begin{align}
    m\ge C_1(\frac{\mathscr{H}(\mathcal{B},\delta\epsilon^3)}{\epsilon^2}+\frac{\omega^2(\mathcal{D})}{\epsilon^2}) \label{eq:recover_sam}
\end{align} holds for some sufficiently large $C_1$, then the event 
\begin{align}\label{eq:recover_qpe}
    \sup_{\bm{b}\in \mathcal{B}}\sup_{\bm{d}\in\mathcal{D}}\big|\langle \mathcal{Q}_\delta(\sqrt{m}\bm{b}+\bm{\tau})-\sqrt{m}\bm{b},\sqrt{m}\bm{d}\rangle\big| \le C_2\delta m\epsilon  
\end{align}
holds with probability exceeding $1-12\exp(-c_3\mathscr{H}(\mathcal{B},\delta\epsilon^3))$.
\end{corollary}
\begin{proof}
\textbf{Applying \cref{thm:globalqpe}:}
    We invoke \cref{thm:globalqpe} with $\mathcal{A}=\mathcal{C}=\{0\}$, $\mathcal{E}= \{0\}\times \mathcal{D}$, $\bm{\epsilon}=0$. Using arbitrarily small $\rho_1$ that renders the second condition in \cref{eq:scaling_3cons}, we always have $\mathscr{H}(\mathcal{A},\rho_1)=0$ and $\omega^2(\mathcal{A}_{\loc}^{(\rho_1)})=0$ due to $\mathcal{A}=\{0\}$. Combining with     $\mathcal{B},\mathcal{D}\subset \mathbb{B}_2^m$, \cref{thm:globalqpe} gives that, if $\zeta \in (0,\frac{\delta}{2})$ and $\rho_2\lesssim\zeta$, $m\gtrsim \frac{2\delta\cdot \mathscr{H}(\mathcal{B},\rho_2)}{\zeta}$, then the event 
    \begin{align}\label{eq:initial_bound}
        \sup_{\bm{b}\in \mathcal{B}}\sup_{\bm{d}\in \mathcal{D}}\big|\langle \mathcal{Q}_\delta(\sqrt{m}\bm{b}+\bm{\tau})-\sqrt{m}\bm{b}, \sqrt{m}\bm{d}\rangle\big|\lesssim \delta m\Big(\frac{\omega(\mathcal{D})}{\sqrt{m}}+\sqrt{\frac{\zeta}{\delta}}+\frac{\rho_2}{\zeta}\Big)
    \end{align}
    holds with probability exceeding $1-12\exp(-\Omega(\mathscr{H}(\mathcal{B},\rho_2)))$.

    \textbf{Choosing Parameters:} Given sufficiently small $\epsilon>0$, we set $$\zeta=\delta\epsilon^2,~\rho_2=\epsilon\zeta = \delta \epsilon^3$$ that satisfy the required conditions $\zeta\in(0,\frac{\delta}{2})$  and $\rho_2\lesssim\zeta$. Under such choice, the required sample size of $m\gtrsim \frac{2\delta\cdot \mathscr{H}(\mathcal{B},\rho_2)}{\zeta}$ reads as $m\gtrsim \frac{\mathscr{H}(\mathcal{B},\delta\epsilon^3)}{\epsilon^2}$, which is satisfied due to \cref{eq:recover_sam}, and the right-hand side of \cref{eq:initial_bound} becomes $\delta m (\frac{\omega(\calD)}{\sqrt{m}}+2\epsilon)$. Combining with  $m\gtrsim \frac{\omega^2(\mathcal{D})}{\epsilon^2}$, we arrive at the desired bound of $O(\delta m \epsilon)$. The  promised probability is directly dictated from \cref{thm:globalqpe}.  
\end{proof}
\begin{rem}
    \label{rem:recover}
    There is no essential difference between our \cref{coro:recover} and \cite[Prop. 6.1]{xu2020quantized}, and   we simply  note the specific two points: (i) \cite[Prop. 6.1]{xu2020quantized} is stated for a fixed $\bm{d}=\bm{d}_0$, which corresponds to the special case of \cref{coro:recover} with $\mathcal{D}=\{\bm{d}_0\}$; (ii) \cite[Prop. 6.1]{xu2020quantized}
    is stated for    $\bm{\tilde{b}}:=\sqrt{m}\bm{b}\in\tilde{\mathcal{B}}:=\sqrt{m}\mathcal{B}$, so the sample complexity is consistent since $\mathscr{H}(\mathcal{B},\delta\epsilon^3)=\mathscr{H}(\tilde{\mathcal{B}},\sqrt{m}\delta\epsilon^3)$ always holds. 
\end{rem}

Next, we show that   improvement can be obtained if $(\bm{b},\bm{d})$ in \cref{eq:recover_qpe} is {\it modulated} by a sub-Gaussian sensing matrix (rather than being simply re-scaled by a factor of $\sqrt{m}$), as will be discussed in \cref{rem:improve}. To get the improved QPE property, we invoke \cref{thm:globalqpe} with $\mathcal{B}=\mathcal{D}=\{0\}$. 

\begin{corollary}
    [Improved QPE under Sub-Gaussian Matrix] \label{coro:improved_qpe}
    Given $\mathcal{A},\mathcal{C}\subset \mathbb{B}_2^n$, we assume that the sub-Gaussian sensing matrix $\bm{\Phi}$ and uniform dither $\bm{\tau}$ are as described in \cref{assump1}. Given any small enough $\epsilon>0$, we let $\rho=\frac{c_0\delta\epsilon}{\sqrt{\log\epsilon^{-1}}}$ for sufficiently small absolute constant $c_0$. If 
    \begin{align}\label{eq:improve_sam}
        m\ge C_1\Big(\frac{\mathscr{H}(\mathcal{A},\rho)}{\epsilon^2\log(\epsilon^{-1})}+ \frac{\omega^2(\mathcal{A}^{(\rho)}_{\loc})}{\delta^2\epsilon^3}+\frac{\omega^2(\mathcal{C})}{\epsilon^2\log(\epsilon^{-1})}\Big)
    \end{align}
    holds for some sufficiently large $C_1$, then the event 
    \begin{align}
        \label{eq:imp_qpe_bound}\sup_{\bm{a}\in\mathcal{A}}\sup_{\bm{c}\in \mathcal{C}}\big|\langle\mathcal{Q}_\delta(\bm{\Phi a}+\bm{\tau})-\bm{\Phi a},\bm{\Phi c}\rangle\big|\le C_2 \delta m \epsilon\sqrt{\log\epsilon^{-1}}
    \end{align}
    holds with probability exceeding $1-12\exp(-c_3\mathscr{H}(\mathcal{A},\rho))$.
\end{corollary}
\begin{proof}
\textbf{Applying \cref{thm:globalqpe}:} 
    We invoke \cref{thm:globalqpe} with $\mathcal{B}=\mathcal{D}=\{0\}$, $\mathcal{E}=\mathcal{C}\times\{0\}$, $\bm{\epsilon}=0$. Using arbitrarily small $\rho_2$ to render $\rho_2\lesssim\zeta$ needed in \cref{eq:scaling_3cons}, we always have $\mathscr{H}(\mathcal{B},\rho_2)=0$ since $\mathscr{B}=\{0\}$. Combining with $\mathcal{A},\mathcal{C}\subset \mathbb{B}_2^n$, \cref{thm:globalqpe} gives that, if for some positive scalars $\rho=\rho_1,\zeta$ we have
    \begin{align}\label{eq:needed_improve}
        \zeta\in \big(0,\frac{\delta}{2}\big),~\rho \lesssim \frac{\zeta}{(\log\frac{\delta}{\zeta})^{1/2}},~m\gtrsim \frac{\delta\cdot \mathscr{H}(\calA,\rho)}{\zeta}+ \frac{\omega^2(\mathcal{A}^{(\rho)}_{\loc})}{\zeta^2},
    \end{align}
    then the event 
    \begin{align}
        \label{eq:quantity_int}&\sup_{\bm{a}\in\mathcal{A}}\sup_{\bm{c}\in \mathcal{C}}\big|\langle\mathcal{Q}_\delta(\bm{\Phi a}+\bm{\tau})-\bm{\Phi a},\bm{\Phi c}\rangle\big| \\&\quad\quad\quad\lesssim \delta m \Big(\frac{\omega(\mathcal{C})}{\sqrt{m}}+\frac{\sqrt{\mathscr{H}(\mathcal{A},\rho})}{\sqrt{m}}+\frac{\zeta}{\delta}\sqrt{\log\frac{\delta}{\zeta}}+ \frac{\omega^2(\mathcal{A}^{(\rho)}_{\loc})}{m\zeta^2}\sqrt{\log\frac{\delta}{\zeta}}\Big)\label{eq:imp_initial}
    \end{align}
    holds with probability exceeding $1-12\exp(-\Omega(\mathscr{H}(\mathcal{A},\rho)))$.

    \textbf{Choosing Parameters:} We proceed with the parametrization $\zeta = \delta \epsilon$ with the given  small enough $\epsilon$, which along  with $\rho=\frac{c_0\delta\epsilon}{\sqrt{\log\epsilon^{-1}}}$ ensures the first two conditions in \cref{eq:needed_improve}. Note that   \cref{eq:improve_sam}, additionally implies $$m\gtrsim \frac{\mathscr{H}(\mathcal{A},\rho)}{\epsilon}+\frac{\omega^2(\mathcal{A}^{(\rho)}_{\loc})}{\delta^2\epsilon^2},$$ which is just the third condition in \cref{eq:needed_improve} due to $\zeta=\delta\epsilon$. Therefore,  the bound \cref{eq:imp_initial} on \cref{eq:quantity_int} holds with the promised probability, and by substituting $\zeta=\delta\epsilon$ it reads as \begin{align}\label{eq:imp_2bound}
        O\left(\delta m \Big(\frac{\omega(\mathcal{C})}{\sqrt{m}}+\frac{\sqrt{\mathscr{H}(\mathcal{A},\rho_1})}{\sqrt{m}}+\epsilon\sqrt{\log\epsilon^{-1}}+\frac{\omega^2(\mathcal{A}^{(\rho)}_{\loc})\sqrt{\log\epsilon^{-1}}}{m\delta^2\epsilon^2}\Big)\right)
    \end{align}
     The desired claim \cref{eq:imp_qpe_bound} thus follows under the sample complexity in \cref{eq:improve_sam} since under the assumed sample complexity \cref{eq:improve_sam}, \cref{eq:imp_2bound} scales as $O(\delta m \epsilon\sqrt{\log\epsilon^{-1}})$ ---  specifically, $m\gtrsim \frac{\mathscr{H}(\mathcal{A},\rho_1)}{\epsilon^2\log(\epsilon^{-1})}$ ensures $\frac{\sqrt{\mathscr{H}(\mathcal{A},\rho)}}{\sqrt{m}}\lesssim \epsilon\sqrt{\log\epsilon^{-1}}$, $m\gtrsim \frac{\omega^2(\mathcal{A}^{(\rho)}_{\loc})}{\delta^2\epsilon^3}$ ensures $\frac{\omega^2(\mathcal{A}^{(\rho)}_{\loc})\sqrt{\log\epsilon^{-1}}}{m\delta^2\epsilon^2}\lesssim \epsilon\sqrt{\log\epsilon^{-1}}$, and $m\gtrsim \frac{\omega^2(\mathcal{C})}{\epsilon^2\log(\epsilon^{-1})}$ implies $\frac{\omega(\mathcal{C})}{\sqrt{m}}\lesssim \epsilon\sqrt{\log\epsilon^{-1}}$. The proof is complete.
\end{proof}
\begin{rem}
    [Comparing \cref{coro:recover} and \cref{coro:improved_qpe}] \label{rem:improve} Note that the random process in \cref{eq:imp_qpe_bound} reduces to the one in \cref{eq:recover_qpe} 
    when $\bPhi=\sqrt{m}\bm{I}$. In this remark, we show that distortion in \cref{eq:imp_qpe_bound} exhibits a  decaying rate in $m$ faster than \cref{eq:recover_qpe}, due to the modulation of the sub-Gaussian matrix $\bm{\Phi}$. Recall that \cref{coro:recover} and \cref{coro:improved_qpe} aim to handle arbitrary signal sets $\mathcal{A}$ and $\mathcal{B}$, thus Sudakov's inequality \cref{sudakov} is tight. Moreover, we use the simple bound (that follows from $\rho\le1$ and \cite[Prop. 7.5.2(e)]{vershynin2018high})\begin{align}\label{eq:local_bound}
        \omega^2(\mathcal{A}^{(\rho)}_{\loc}) \le \omega^2(\mathcal{A}-\mathcal{A})=4\omega^2(\mathcal{A}).
    \end{align}
    Now we present \cref{coro:recover} and \cref{coro:improved_qpe} in the form of error rate: \begin{itemize}
        [leftmargin=5ex,topsep=0.25ex]
     \item \textbf{The Decaying Rate of \cref{coro:recover}:} By Sudakov's inequality, $m\gtrsim \frac{\omega^2(\mathcal{B})} {\delta^2\epsilon^8}+\frac{\omega^2(\mathcal{D})}{\epsilon^2}$ (with large enough implied constant, implicitly below), or equivalently $$\delta\epsilon\gtrsim \delta^{3/4}\Big(\frac{\omega^2(\mathcal{B})}{m}\Big)^{1/8}+\delta\Big(\frac{\omega^2(\mathcal{D})}{m}\Big)^{1/2},$$  suffices for ensuring  \cref{eq:recover_sam}. Therefore, provided that $m\gtrsim \frac{\omega^2(\mathcal{B})}{\delta^2}+\omega^2(\mathcal{D})$, \cref{eq:recover_qpe} in \cref{coro:recover} implies the following bound on the QPE distortion: \begin{align}
          \sup_{\bm{b}\in \mathcal{B}}\sup_{\bm{d}\in\mathcal{D}}\frac{1}{m}\big|\langle \mathcal{Q}_\delta(\sqrt{m}\bm{b}+\bm{\tau})-\sqrt{m}\bm{b},\sqrt{m}\bm{d}\rangle\big| \lesssim \delta^{3/4}\Big(\frac{\omega^2(\mathcal{B})}{m}\Big)^{1/8}+ \delta \Big(\frac{\omega^2(\mathcal{D})}{m}\Big)^{1/2}. \label{eq:rate_recover}
     \end{align}
        \item \textbf{The Decaying Rate of \cref{coro:improved_qpe}:} By Sudakov's inequality and \cref{eq:local_bound}, the condition \begin{align}
            m\gtrsim \frac{\omega^2(\mathcal{A})}{\rho^2\epsilon^2\log(\epsilon^{-1})}+\frac{\omega^2(\mathcal{A})}{\delta^2\epsilon^3}+\frac{\omega^2(\mathcal{C})}{\epsilon^2\log(\epsilon^{-1})} \label{eq:suff_con_im}
        \end{align} suffices for ensuring \cref{eq:improve_sam}. By substituting $\rho = \frac{c_0\delta\epsilon}{\sqrt{\log\epsilon^{-1}}}$ and   $\epsilon<1$, we can write \cref{eq:suff_con_im} as $m\gtrsim \frac{\omega^2(\calA)}{\delta^2\epsilon^4}+\frac{\omega^2(\calC)}{\epsilon^2\log(\epsilon^{-1})}$, and further note that this can be guaranteed by two conditions: 
        \begin{align}
            \label{eq:twocon_1}
            \epsilon\gtrsim  \Big(\frac{\omega^2(\mathcal{A})}{\delta^2m}\Big)^{1/4}~\text{and}~\delta\epsilon \sqrt{\log\epsilon^{-1}}\gtrsim \delta \Big(\frac{\omega^2(\mathcal{C})}{m}\Big)^{1/2}
        \end{align}
        Moreover, since $\epsilon\sqrt{\log\epsilon^{-1}}$ is monotonically increasing with $\epsilon$ when $\epsilon$ is sufficiently small,  the first condition in \cref{eq:twocon_1} is equivalent to \begin{align}
            \delta\epsilon\sqrt{\log \epsilon^{-1}}\gtrsim \sqrt{\delta}\Big(\frac{\omega^2(\mathcal{A})}{m}\Big)^{1/4}\Big(\log\frac{\delta^2m}{\omega^2(\mathcal{A})}\Big)^{1/2}. 
        \end{align} Overall, the above analysis shows that 
        \begin{align}
            \delta\epsilon\sqrt{\log\epsilon^{-1}} \gtrsim \sqrt{\delta}\Big(\frac{\omega^2(\mathcal{A})}{m}\log^{2}\Big(\frac{\delta^2m}{\omega^2(\mathcal{A})}\Big)\Big)^{1/4}+\delta\Big(\frac{\omega^2(\mathcal{C})}{m}\Big)^{1/2}
        \end{align}
        with sufficiently large implied constant can imply \cref{eq:improve_sam}. Therefore, provided that $m\gtrsim \frac{\omega^2(\mathcal{A})}{\delta^2}+\omega^2(\mathcal{C})$, \cref{eq:imp_qpe_bound} in  \cref{coro:improved_qpe} implies the following bound on the QPE distortion: 
        \begin{align}
        \nonumber  & \sup_{\bm{a}\in\mathcal{A}}\sup_{\bm{c}\in \mathcal{C}}\frac{1}{m}\big|\langle\mathcal{Q}_\delta(\bm{\Phi a}+\bm{\tau})-\bm{\Phi a},\bm{\Phi c}\rangle\big| \\&\quad\quad\quad\lesssim \sqrt{\delta}\Big(\frac{\omega^2(\mathcal{A})}{m}\log^{2}\Big(\frac{\delta^2m}{\omega^2(\mathcal{A})}\Big)\Big)^{1/4} +\delta\Big(\frac{\omega^2(\mathcal{C})}{m}\Big)^{1/2}.\label{eq:rate_improved}
        \end{align}
    \end{itemize} 
    Comparing \cref{eq:rate_recover} and \cref{eq:rate_improved}, it shall be clear that \cref{coro:improved_qpe} provides   decaying rate of the QPE distortion faster than \cref{coro:recover}. 
\end{rem}

\subsection{Improving    Uniform Error Decaying Rate for PBP} \label{app:improve}

We present an interesting by-product as our final technical development: with \cref{coro:improved_qpe}, under sub-Gaussian measurement matrix, we are able to improve  
the uniform recovery guarantee for the projected-back projection (PBP) estimator over signals from a convex and symmetric set $\mathcal{K}$ in \cite[Sec. 7.3\textbf{B}]{xu2020quantized}. 

\subsubsection*{PBP Estimator and Uniform Guarantee in \cite{xu2020quantized}} We first review the PBP estimator and the related result in  \cite{xu2020quantized}. Suppose that the signal $\bm{x^\star}\in \mathbb{R}^n$ lies in some convex and symmetric set $\mathcal{K}$ with $\rad(\mathcal{K})\le 1$, and the reader may think of a typical example given by {\it the set of effectively sparse signals} (e.g., \cite{plan2012robust,plan2013one})
\begin{align}
\mathcal{K}=\mathbb{B}_1^n(\sqrt{s})\cap \mathbb{B}_2^n = \{\bm{x}\in \mathbb{R}^n:\|\bm{x}\|_1 \le \sqrt{s},\|\bm{x}\|_2\le 1\},
\end{align}which is essentially   larger   than the set of exactly $s$-sparse signals $\Sigma^n_s\cap \mathbb{B}_2^n$. 
Under the sensing matrix
$\bm{\Phi}\in\mathbb{R}^{m\times n}$ and uniform dither $\bm{\tau}\sim\mathscr{U}([-\frac{\delta}{2},\frac{\delta}{2}]^m)$, 
we observe the quantized measurements $\bm{\dot{y}}=\mathcal{Q}_\delta (\bm{\Phi x^\star}+\bm{\tau})$.
Let $\mathcal{P}_{\mathcal{K}}(\cdot)$ be the projection operator onto $\mathcal{K}$ under $\ell_2$-norm, then the PBP estimator is given by (e.g., \cite{xu2020quantized,plan2017high})
\begin{align}\label{eq:pbp}
    \bm{\hat{x}}_{\rm PBP} = \mathcal{P}_{\mathcal{K}}\Big(\frac{1}{m}\bm{\Phi}^\top\bm{\dot{y}}\Big).  
\end{align}
Given a general RIP matrix $\bm{\Phi}$ (one that satisfies restricted isometry property (RIP)),   
it was shown in \cite[Sec. 7.3\textbf{B}]{xu2020quantized} that PBP achieves uniform recovery over all $\bm{x^\star}\in\mathcal{K}$ with the following error rate (up to logarithmic factors) 
\begin{align}\label{eq:xurate}
    \|\bm{\hat{x}}_{\rm PBP}-\bm{x^\star}\|_2 =\tilde{O}\left((1+\delta)^{\frac{1}{2}}\Big(\frac{\omega^2(\mathcal{K})}{m}\Big)^{\frac{1}{16}}\right). 
\end{align}

We also  recap some argument from the proofs of \cite[Thm. 4.3, Coro. 3.1]{xu2020quantized} for bounding the PBP estimation error. Let $\bm{a}_{\bm{x^\star}}=\frac{1}{m}\bm{\Phi}^\top\bm{\dot{y}}$ (note that this depends on $\bm{x^\star}$) be the intermediate estimator, then the PBP estimator can be written as  $\bm{\hat{x}}_{\rm PBP}=\mathcal{P}_{\mathcal{K}}(\bm{a}_{\bm{x^\star}})$, and we can proceed as
\begin{align}\label{eq:start_bound_pbp}
    &\|\bm{\hat{x}}_{\rm PBP}-\bm{x^\star}\|_2^2\\&\quad= \|\mathcal{P}_{\mathcal{K}}(\bm{a}_{\bm{x^\star}})-\mathcal{P}_{\mathcal{K}}(\bm{x^\star})\|_2^2 \\
    \label{eq:proj_prope}&\quad\le \langle \bm{x^\star}-\bm{a}_{\bm{x^\star}},\mathcal{P}_{\mathcal{K}}(\bm{x^\star})-\mathcal{P}_{\mathcal{K}}(\bm{a}_{\bm{x^\star}})\rangle \\&\quad\le 
| \langle \bm{x^\star}-\bm{a}_{\bm{x^\star}},\mathcal{P}_{\mathcal{K}}(\bm{x^\star})\rangle| +|\langle \bm{x^\star}-\bm{a}_{\bm{x^\star}},\mathcal{P}_{\mathcal{K}}(\bm{a}_{\bm{x^\star}})\rangle|\\ 
    \label{eq:supout1}&\quad\le 2\sup_{\bm{u}\in \mathcal{K}}\Big|\Big\langle \bm{x^\star}- \frac{1}{m}\bm{\Phi}^\top \cdot\mathcal{Q}_\delta(\bm{\Phi x^\star}+\bm{\tau}),\bm{u}\Big\rangle\Big|\\
 \label{eq:supx}&\quad\le 2\sup_{\bm{v}\in\mathcal{K}}\sup_{\bm{u}\in \mathcal{K}}\Big|\langle \bm{v},\bm{u}\rangle - \frac{1}{m}\big\langle \mathcal{Q}_\delta(\bm{\Phi v}+\bm{\tau}),\bm{\Phi u}\big\rangle\Big|\\
 \label{eq:use_tri_3}&\quad\le 2\underbrace{\sup_{\bm{u},\bm{v}\in\mathcal{K}}\Big|\langle\bm{v},\bm{u}\rangle - \frac{1}{m}\langle \bm{\Phi v},\bm{\Phi u}\rangle\Big|}_{:=I_1} +2\underbrace{\sup_{\bm{u},\bm{v}\in\mathcal{K}}\frac{1}{m}\big|\langle \mathcal{Q}_\delta(\bm{\Phi v}+\bm{\tau})-\bm{\Phi v},\bm{\Phi u}\rangle\big|}_{:=I_2}
\end{align}
where \cref{eq:proj_prope} follows from the non-expansivity of the projector onto $\mathcal{K}$, in \cref{eq:supout1} we take the supremum over $\mathcal{P}_{\mathcal{K}}(\bm{x^\star}),\mathcal{P}_{\mathcal{K}}(\bm{a}_{\bm{x^\star}})\in\mathcal{K}$ and substitute $\bm{a}_{\bm{x^\star}}=\frac{1}{m}\bm{\Phi}^\top\bm{\dot{y}}=\frac{1}{m}\bm{\Phi}^\top \cdot\mathcal{Q}_\delta(\bm{\Phi x^\star}+\bm{\tau})$, in \cref{eq:supx} we take the supremum over $\bm{x^\star}\in \mathcal{K}$, and \cref{eq:use_tri_3} follows from triangle inequality. Then, the critical observation made by \cite{xu2020quantized} can be summarized as follows:
$$I_1~\text{\it can be bounded by RIP},~I_2~ \text{\it can be bounded by QPE}.$$
   To see how one can bound $I_1$ via RIP, suppose that $\mathcal{K}\subset \mathbb{B}_2^n$, the RIP over the convex symmetric $\mathcal{K}$ with distortion $\beta$ is formulated as 
\begin{align}\label{eq:rip}
    \sup_{\bm{w}\in \mathcal{K}}\Big|\frac{1}{m}\|\bm{\Phi w}\|_2^2-\|\bm{w}\|_2^2\Big| \le \beta,
\end{align}
and one can bound $I_1$ by \cref{eq:rip} since 
\begin{align}
    I_1 & = \sup_{\bm{u},\bm{v}\in\mathcal{K}}\left|\frac{1}{m}\Big\|\bm{\Phi}\Big(\frac{\bm{u}+\bm{v}}{2}\Big)\Big\|_2^2-\frac{1}{m}\Big\|\bm{\Phi}\Big(\frac{\bm{u}-\bm{v}}{2}\Big)\Big\|_2^2-\Big\|\frac{\bm{u}+\bm{v}}{2}\Big\|_2^2+\Big\|\frac{\bm{u}-\bm{v}}{2}\Big\|_2^2 \right|\\
    &\le 2  \sup_{\bm{w}\in \mathcal{K}}\Big|\frac{1}{m}\|\bm{\Phi w}\|_2^2-\|\bm{w}\|_2^2\Big| \le \beta \le 2\beta, \label{eq:use_tip} 
\end{align}
where \cref{eq:use_tip} holds because $\frac{\bm{u}+\bm{v}}{2},\frac{\bm{u}-\bm{v}}{2}\in\mathcal{K}$ (recall that $\mathcal{K}$ is convex and symmetric).

\subsubsection*{Improved Rate under Sub-Gaussian $\bm{\Phi}$} In the specific instance of sub-Gaussian $\bm{\Phi}$, we are able to improve \cref{eq:xurate} by bounding $I_2$ via our \cref{coro:improved_qpe}. We formally present this as the following statement. 
\begin{proposition}
    [Improved PBP Uniform Rate under sub-Gaussian Matrix]\label{prop:improve} Let the sub-Gaussian matrix $\bm{\Phi}\in \mathbb{R}^{m\times n}$ and the uniform dither $\bm{\tau}\sim \mathscr{U}([-\frac{\delta}{2},\frac{\delta}{2}]^m)$ be  as described in \cref{assump1}, $\bm{x^\star}\in \mathcal{K}$ for some convex symmetric $\mathcal{K}\subset \mathbb{B}_2^n$, and from the quantized observations $\bm{\dot{y}}=\mathcal{Q}_\delta(\bm{\Phi x^\star}+\bm{\tau})$ we recover $\bm{x^\star}$ by PBP as per \cref{eq:pbp}. If $m\gtrsim (1+\delta^{-2})\omega^2(\mathcal{K})$ with large enough implied constant, then for some small enough absolute constants $c_0,c_1$, with probability exceeding $1-\exp(-\omega^2(\mathcal{K}))-12\exp(-c_1\mathscr{H}(\mathcal{K},c_0\delta))$  on a single draw of $(\bm{\Phi},\bm{\tau})$, the   error rate \begin{align}\label{eq:pbp_uniform}
        \|\bm{\hat{x}}_{\rm PBP}-\bm{x^\star}\|_2 \lesssim \delta^{\frac{1}{4}}\Big(\frac{\omega^2(\mathcal{K})}{m}\log^2\Big(\frac{\delta^2 m}{\omega^2(\mathcal{K})}\Big)\Big)^{\frac{1}{8}} + (1+\delta)^{\frac{1}{2}}\Big(\frac{\omega^2(\mathcal{K})}{m}\Big)^{\frac{1}{4}}
    \end{align}
    holds uniformly for all $\bm{x^\star}\in \mathcal{K}$. 
\end{proposition}
\begin{proof}
 We bound $I_1$ and $I_2$ in \cref{eq:use_tri_3} separately. 

\subsubsection*{Bounding $I_1$} 
Recall from \cref{eq:use_tip} that $I_1\le 2\beta$ as long as the RIP in \cref{eq:rip} holds. Thus, we only need to identify the  value of $\beta$ by bounding $\sup_{\bm{w}\in \mathcal{K}}|\frac{1}{m}\|\bm{\Phi w}\|_2^2-\|\bm{w}\|_2^2|$. To achieve this, we set $\mathcal{T}=\mathcal{K}\times\{0\}$ to obtain that for any $t\ge 0$, the event 
$\sup_{\bm{w}\in \mathcal{K}}\big|\frac{\|\bm{\Phi w}\|_2}{\sqrt{m}}- \|\bm{w}\|_2\big| \lesssim \frac{\omega(\mathcal{K})+t}{\sqrt{m}}$ holds with probability exceeding $1-\exp(-t^2)$ (note that for symmetric $\mathcal{K}$ we have $\omega(\mathcal{K})=\gamma(\mathcal{K})$). Then, we set $t = \omega(\mathcal{K})$ to obtain \begin{align}
\sup_{\bm{w}\in\mathcal{K}}\Big|\frac{\bm{\|\Phi w}\|_2}{\sqrt{m}}- \|\bm{w}\|_2\Big|\lesssim \frac{\omega(\mathcal{K})}{\sqrt{m}}
\end{align} 
with probability exceeding $1-\exp(-\omega^2(\mathcal{K}))$. Note that $\frac{\omega(\mathcal{K})}{\sqrt{m}}=O(1)$, and hence we have $\sup_{\bm{w}\in\mathcal{K}}\frac{\|\bm{\Phi w}\|_2}{\sqrt{m}}\le \sup_{\bm{w}\in\mathcal{K}}|\frac{\|\bm{\Phi w}\|_2}{\sqrt{m}}-\|\bm{w}\|_2|+\sup_{\bm{w}\in\mathcal{K}}\|\bm{w}\|_2=O(1)$. Thus we have \begin{align}
    \sup_{\bm{w}\in \mathcal{K}}\Big|\frac{\|\bm{\Phi w}\|_2^2}{m}-\|\bm{w}\|_2^2\Big|\le \sup_{\bm{w}\in \mathcal{K}}\Big|\frac{\|\bm{\Phi w}\|_2}{\sqrt{m}}-\|\bm{w}\|_2\Big|\cdot  \sup_{\bm{w}\in \mathcal{K}}\Big|\frac{\|\bm{\Phi w}\|_2}{\sqrt{m}}+\|\bm{w}\|_2\Big| \lesssim \frac{\omega(\mathcal{K})}{\sqrt{m}}.  \label{eq:rip_bound}
\end{align} 
Combining with \cref{eq:use_tip}, we obtain that $I_1\lesssim \frac{\omega(\mathcal{K})}{\sqrt{m}}$ with probability exceeding $1-\exp(-\omega^2(\mathcal{K}))$.

\subsubsection*{Bounding $I_2$} We bound $I_2$ by \cref{coro:improved_qpe}. As reformulated in \cref{rem:improve}, provided that $m\gtrsim (1+\delta^{-2})\omega^2(\mathcal{K})$, \cref{coro:improved_qpe} implies that \begin{align}
    \label{eq:imp_bound_I2}I_2\lesssim\sqrt{\delta}\Big(\frac{\omega^2(\mathcal{K})}{m}\log^{2}\Big(\frac{\delta^2m}{\omega^2(\mathcal{K})}\Big)\Big)^{1/4} +\delta\Big(\frac{\omega^2(\mathcal{K})}{m}\Big)^{1/2}
\end{align}
holds with probability exceeding $1-12\exp(-c_1\mathscr{H}(\mathcal{K},c_0\delta))$ for some small enough $c_0$ (recall that $\rho\asymp \frac{\delta\epsilon}{\sqrt{\log\epsilon^{-1}}}$ for small enough $\epsilon$ in \cref{coro:improved_qpe}). Substituting $I_1\lesssim \frac{\omega(\mathcal{K})}{\sqrt{m}}$ and \cref{eq:imp_bound_I2} into \cref{eq:use_tri_3} yields the claim.   
\end{proof}
\begin{rem}
    Under the regular scaling of $\delta\asymp 1$, our \cref{prop:improve} provides a uniform error rate of $\tilde{O}((\frac{\omega^2(\mathcal{K})}{m})^{1/8})$, which improves on $\tilde{O}((\frac{\omega^2(\mathcal{K})}{m})^{1/16})$ in \cite{xu2020quantized} in the specific instance of sub-Gaussian $\bm{\Phi}$. 
\end{rem}
\begin{rem}[Improving the Non-Uniform Error Rate]
While for general RIP matrix $\bm{\Phi}$ the non-uniform error rate (for recovering a fixed $\bm{x^\star}\in \mathcal{K}$) reads as $\tilde{O}((1+\delta)^{1/2}(\frac{\omega^2(\mathcal{K})}{m})^{1/8})$ \cite[Sec. 7.3\textbf{B}]{xu2020quantized}, we note that a faster decaying rate of $O(m^{-1/4})$ can also be obtained if $\bm{\Phi}$ is sub-Gaussian:
\begin{itemize}
 [leftmargin=5ex,topsep=0.25ex]
    \item  First note that \cref{eq:start_bound_pbp}--\cref{eq:use_tri_3} bounds the recovery error for a fixed $\bm{x^\star}$ as 
\begin{align}
\|\bm{\hat{x}}_{\rm PBP}-\bm{x^\star}\|_2^2 &\le 
    2\sup_{\bm{u}\in \mathcal{K}}\Big|\langle \bm{x^\star},\bm{u}\rangle - \frac{1}{m}\langle \bm{\Phi x^\star},\bm{\Phi u}\rangle\Big|\\&\quad\quad + 2\sup_{\bm{u}\in \mathcal{K}}\frac{1}{m}\Big|\langle\mathcal{Q}_\delta(\bm{\Phi x^\star}+\bm{\tau})-\bm{\Phi x^\star},\bm{\Phi u}\rangle\Big|:=I_3+I_4; 
\end{align}

\item Then, by \cref{eq:use_tip} and \cref{eq:rip_bound} we obtain $I_3 \lesssim \frac{\omega(\mathcal{K})}{\sqrt{m}}$, and    \cref{lem:localqpe} (local QPE, with $\bm{a}=\bm{x^\star},\bm{b}=0,\bm{\epsilon}=0,\mathcal{E}=\mathcal{K}\times\{0\}$) gives $I_4 \lesssim \frac{\delta\cdot\omega(\mathcal{K})}{\sqrt{m}}$, thus yielding the non-uniform rate $O((1+\delta)^{1/2}(\frac{\omega^2(\mathcal{K})}{m})^{1/4})$. 
\end{itemize}
  Comparing with the uniform rate $O(m^{-1/8})$   \cref{eq:pbp_uniform} provides such implication: when estimating signals living in a convex symmetric set, the cost of uniform recovery is essential. Moreover, we   note that under  a general non-linear model with Gaussian $\bm{\Phi}$, the non-uniform recovery of PBP has been systematically studied in \cite{plan2017high}, and and their rate also reads as $O(m^{-1/4})$ for the set of effectively sparse signals \cite[Sec. 2.6]{plan2017high} (this is the canonical example of bounded convex signal set).  
\end{rem}
\section{A Table of Recurring Notation}\label{app:table}
\begin{table}

\caption{Table of recurring notation. \label{tbl:notation}}
\begin{tabular}{|c|c|}
\hline 
\multicolumn{2}{|c|}{\textbf{Introduced in main text}}\tabularnewline
\hline 
$\bm{x^\star},n$ & Underlying signal and its dimension\tabularnewline
\hline 
$\bm{v^\star},m$ & Underlying corruption and its dimension
\tabularnewline
\hline 
$\bm{\Phi},\bm{\Phi}_i^\top$& Sub-Gaussian sensing matrix and its $i$-th sensing vector (row)
\tabularnewline
\hline 
$\bm{\epsilon},E$ & Independent sub-Gaussian noise with sub-Gaussian norm bounded by $E$
\tabularnewline
\hline 
$\delta,\bm{\tau},\mathcal{Q}_\delta(\cdot) $ & Quantization resolution, uniform dither, uniform quantizer 
\tabularnewline
\hline 
$\bm{y},\bm{\dot{y}}$ & Unquantized noisy measurements \cref{3.1}, quantized measurements \cref{3.2}
\tabularnewline
\hline 
$\bm{\xi}_{\bm{a},\bm{b}}$ & Quantization noise associated with     $(\bm{a},\bm{b})$ \cref{eq:xiabfirst} 
\tabularnewline
\hline 
$f(\cdot),g(\cdot)$&Norms for promoting the structures of $\bm{x^\star}$ and $\bm{v^\star}$ (\cref{assump2}) 
\tabularnewline
\hline 
$\mathcal{D}_f(\bm{x}),\mathcal{D}^*_f(\bm{x})$ & Descent cone of $f$ at $\bm{x}$ and its normalized counterpart \cref{descentcone}
\tabularnewline
\hline 
$\mathscr{N}(\mathcal{K},\varepsilon)$& The covering number of $\mathcal{K}$ with radius $\varepsilon$ under Euclidean distance 
\tabularnewline
\hline 
$\mathscr{H}(\mathcal{K},\varepsilon)$ & The Kolmogorov entropy defined as $\mathscr{H}(\mathcal{K},\varepsilon)=\log \mathscr{N}(\mathcal{K},\varepsilon)$ 
\tabularnewline
\hline 
$\omega(\mathcal{K})$& The Gaussian width of $\mathcal{K}\subset \mathbb{R}^n$: $\omega(\mathcal{K})= \mathbbm{E}_{\bm{g}\sim \mathcal{N}(0,\bm{I}_n)}\sup_{\bm{x}\in \mathcal{K}}\bm{g}^\top\bm{x}$  
\tabularnewline
\hline 
$\gamma(\mathcal{K})$ & The Gaussian complexity of $\mathcal{K}\subset \mathbb{R}^n$: $\gamma(\mathcal{K})= \mathbbm{E}_{\bm{g}\sim \mathcal{N}(0,\bm{I}_n)}\sup_{\bm{x}\in \mathcal{K}}|\bm{g}^\top\bm{x}|$
\tabularnewline
\hline
$\rad(\mathcal{K})$ & The radius of $\mathcal{K}$: $\rad(\mathcal{K})=\sup_{\bm{x}\in \mathcal{K}}\|\bm{x}\|_2$
\tabularnewline
\hline
 $\mathcal{K}^{(\rho)}_{\loc}$ & The localized version of $\mathcal{K}$: $\mathcal{K}^{(\rho)}_{\loc} = (\mathcal{K}-\mathcal{K})\cap  \mathbb{B}_2^n(\rho)$ 
 \tabularnewline
\hline 
$\mathcal{D}_{\bm{x}},\mathcal{D}_{\bm{x}}^*,\mathcal{D}_{\bm{v}},\mathcal{D}_{\bm{v}}^*$& Constraint sets for analyzing constrained Lasso \cref{3.91}--\cref{3.101}
\tabularnewline
\hline 
$\alpha_f(\mathcal{X})$ & Compatibility constant between $f$ and $\ell_2$-norm over set    $\mathcal{X}$ \cref{eq:compatible}
\tabularnewline
\hline
$\bm{\Delta_x},\bm{\Delta_v}$ & Reconstruction error of the signal and the corruption
\tabularnewline
\hline
$\alpha_{\bm{x}},\alpha_{\bm{v}}$& Uniform bounds on $\alpha_f(\overline{\mathcal{X}}_{\bm{x^\star}})$ and $\alpha_g(\overline{\mathcal{V}}_{\bm{v^\star}})$ (\cref{assump4})\tabularnewline
\hline 
$\rho_1,\rho_2,\zeta$ & Parameters to be chosen in \cref{thm1}, \cref{thm2}
\tabularnewline 
\hline
$I_1,I_2,I_3,$ etc. & Random processes that we need to bound 
\tabularnewline 
\hline
$\mathcal{C}(\lambda_1,\lambda_2)$ & Constraint set for analyzing unconstrained Lasso \cref{eq:Clam12}
\tabularnewline 
\hline
$k,L,r,k',L',r'$& Parameters for formulating generative priors (\cref{assump5})
\tabularnewline 
\hline
$\Sigma^n_s$ & The set of $s$-sparse $n$-dimensional vectors 
\tabularnewline 
\hline
$M^{p,q}_r$ & The set of $p\times q$ matrices with rank not exceeding $r$ 
\tabularnewline 
\hline
\multicolumn{2}{|c|}{\textbf{Introduced in appendices}}\tabularnewline
\hline 
$\mathcal{K}_{\bm{x}}^-,\mathcal{K}_{\bm{v}}^-,\mathcal{E},\mathcal{E}^*$ & Constraint sets for analyzing     generative case (\cref{proadd1})
\tabularnewline
\hline 
$\rho_1,\rho_2$ & Covering radius for the covering arguments in the proof of \cref{thm:globalqpe}
\tabularnewline
\hline 
$\zeta$ & A parameter in $(0,\frac{\delta}{2})$ introduced in the proof of \cref{thm:globalqpe}  
\tabularnewline
\hline 
$\mathcal{Z}_{\bm{a},\bm{b}}$& ``Bad'' measurements suffering from discontinuity \cref{eq:Zab}
\tabularnewline
\hline
$\mathcal{J}_{\bm{a}}^{\mathcal{A}},\mathcal{J}_{\bm{b}}^{\mathcal{B}}$ & ``Bad'' measurements   suffering from large perturbations \cref{eq:bad1}--\cref{eq:bad2}
\tabularnewline
\hline
$E_1,E_2,E_3$ &  Events that aid the proof of \cref{thm:globalqpe}: \cref{eq:E1}, \cref{eq:E2}, \cref{eq:E3} 
\tabularnewline
\hline
$\mathcal{U}_{\bm{a},\bm{b}}$& ``Bad'' measurements associated with $(\bm{a},\bm{b})$ \cref{eq:bad_index}
\tabularnewline
\hline
$U_0$ & Uniform upper bound on the cardinality of $\mathcal{U}_{\bm{a},\bm{b}}$ \cref{eq:U0bound}
\tabularnewline
\hline
\end{tabular}

\end{table}
\end{appendix}
\bibliographystyle{siamplain}
\bibliography{libr}
\end{document}